\pdfoutput=1
\RequirePackage{ifpdf}
\ifpdf 
\documentclass[pdftex]{sigma}
\else
\documentclass{sigma}
\fi

\numberwithin{equation}{section}

\newtheorem{Theorem}{Theorem}[section]
\newtheorem*{Theorem*}{Theorem}
\newtheorem{Corollary}[Theorem]{Corollary}
\newtheorem{Lemma}[Theorem]{Lemma}
\newtheorem{Proposition}[Theorem]{Proposition}
 { \theoremstyle{definition}
\newtheorem{Definition}[Theorem]{Definition}
\newtheorem{Remark}[Theorem]{Remark} }

\usepackage{calrsfs}

\usepackage{pst-node}
\usepackage{authblk}
\usepackage{tikz}
\usepackage{tikz-cd}
\usepackage[all,cmtip]{xy}

\DeclareMathOperator*{\res}{Res}
\usepackage{fancyhdr}

\begin{document}

\allowdisplaybreaks

\newcommand{\arXivNumber}{2402.09574}

\renewcommand{\PaperNumber}{038}

\FirstPageHeading

\ShortArticleName{1D Landau--Ginzburg Superpotential of Big Quantum Cohomology of $\mathbb{CP}^2$}

\ArticleName{1D Landau--Ginzburg Superpotential\\ of Big Quantum Cohomology of $\boldsymbol{\mathbb{CP}^2}$}

\Author{Guilherme F. ALMEIDA~$^{\rm ab}$}

\AuthorNameForHeading{G.F.~Almeida}

\Address{$^{\rm a)}$~Mannheim University, Mannheim, Germany}
\Address{$^{\rm b)}$~Max Planck Institute of Molecular Cell Biology and Genetics, Dresden, Germany}
\EmailD{\href{mailto:feitosad@mpi-cbg.de}{feitosad@mpi-cbg.de}}

\ArticleDates{Received February 16, 2024, in final form May 12, 2025; Published online May 30, 2025}

\Abstract{Using the inverse period map of the Gauss--Manin connection associated with $QH^{*}\bigl(\mathbb{CP}^2\bigr)$ and the Dubrovin construction of Landau--Ginzburg superpotential for Dubrovin--Frobenius manifolds, we construct a one-dimensional Landau--Ginzburg superpotential for the quantum cohomology of $\mathbb{CP}^2$. In the case of small quantum cohomology, the Landau--Ginzburg superpotential is expressed in terms of the cubic root of the $j$-invariant function. For big quantum cohomology, the one-dimensional Landau--Ginzburg superpotential is given by Taylor series expansions whose coefficients are expressed in terms of quasi-modular forms. Furthermore, we express the Landau--Ginzburg superpotential for both small and big quantum cohomology of $QH^{*}\bigl(\mathbb{CP}^2\bigr)$ in closed form as the composition of the Weierstrass $\wp$-function and the universal coverings of \smash{$\mathbb{C} \setminus \bigl(\mathbb{Z} \oplus {\rm e}^{\frac{\pi {\rm i}}{3}}\mathbb{Z}\bigr)$} and $\mathbb{C} \setminus (\mathbb{Z} \oplus z\mathbb{Z})$, respectively.}


\Classification{53D45}

\tableofcontents

\section{Introduction}

The main goal of this paper is to investigate the geometric structure of quantum cohomology of~$\mathbb{CP}^2$ by studying the correspondent 1D Landau--Ginzburg superpotential. We start by providing the necessary background and motivation to our goal.

\subsection{Background}

\begin{Definition}[{\cite{B.Dubrovin2,B.Dubrovin3}}]
An analytic function $F(t)$, where $t = \bigl(t^1, t^2, \dots, t^n\bigr) \in U \subset \mathbb{C}^n$ defined in an open subset of $\mathbb{C}^n$, is considered a solution of the WDVV (Witten--Dijkgraaf--Verlinde--Verlinde) equation if its third derivatives
\smash{$c_{\alpha\beta\gamma} = \frac{\partial^3 F}{\partial t^{\alpha} \partial t^{\beta} \partial t^{\gamma}}$}
satisfy the following conditions:
\begin{enumerate}\itemsep=0pt
\item[(1)] The coefficients $\eta_{\alpha\beta} = c_{1\alpha\beta}$ form elements of a constant nondegenerate matrix.
\item[(2)] The quantities $c_{\alpha\beta}^{\gamma} = \eta^{\gamma\delta} c_{\alpha\beta\delta}$ represent the structure constants of an associative algebra.
\item[(3)] The function $F(t)$ must be quasi-homogeneous.
\end{enumerate}
\end{Definition}

In \cite[Chapter~1]{B.Dubrovin2}, Dubrovin formulated a geometric interpretation of the WDVV equation which is given by the following.

\begin{Definition}
A Frobenius algebra $\mathcal{A}$ is a finite-dimensional unital, commutative, associative algebra equipped with an invariant, non-degenerate bilinear pairing
$
\eta \colon \mathcal{A} \otimes \mathcal{A} \mapsto \mathbb{C}$,
which is invariant in the following sense
$
\eta(A \bullet B, C) = \eta(A, B \bullet C)$, $ \forall A, B, C \in \mathcal{A}$.
\end{Definition}

\begin{Definition}[{\cite{B.Dubrovin2,B.Dubrovin3}}]
Let $M$ be a complex manifold of dimension $n$. A Dubrovin--Frobenius structure over $M$ consists of the following compatible objects:
\begin{enumerate}\itemsep=0pt
\item[(1)] A family of Frobenius multiplications $\bullet_p \colon T_pM \times T_pM \mapsto T_pM$ that are analytically dependent on $p \in M$. This family induces a Frobenius multiplication on
\begin{gather*}
\bullet \colon\ \Gamma(TM) \times \Gamma(TM) \mapsto \Gamma(TM).
\end{gather*}
\item[(2)] A flat pseudo-Riemannian metric $\eta$ on $\Gamma(TM)$, also known as the Saito metric.
\item[(3)] A unity vector field $e$ that is covariantly constant with respect to the Levi-Civita connection~$\nabla$ for the metric $\eta$, i.e., $\nabla e = 0$.
\item[(4)] Consider the tensor $c(X, Y, Z) := \eta(X \bullet Y, Z)$. We require the 4-tensor
$(\nabla_W c)(X, Y, Z)$
to be symmetric with respect to $X, Y, Z, W \in \Gamma(TM)$.
\item[(5)] An Euler vector field $E$ with the following properties
\begin{gather*}
\nabla \nabla E = 0, \qquad \mathcal{L}_E \eta(X, Y) = (2 - d) \eta(X, Y), \qquad \mathcal{L}_E c(X, Y, Z) = c(X, Y, Z),
\end{gather*}
where $X, Y, Z \in \Gamma(TM)$. Moreover, we require $\nabla E$ to be diagonalizable.
Let $\bigl(t^1, t^2, \dots, t^n\bigr)$ be the flat coordinates with respect to the metric $\eta$. These coordinates are denoted as Saito flat coordinates. The Euler vector $E$ can be explicitly represented as
\begin{gather*}
E = \sum_{i=1}^n ((1 - q_i) t_i + r_i) \partial_i.
\end{gather*}
\end{enumerate}
\end{Definition}

Roughly speaking, Dubrovin--Frobenius manifold is the geometric structure that naturally arise in the domain of any WDVV solution, which is given by a family of Frobenius algebra on the sheaf of holomorphic vector fields, a flat structure and some suitable marked vector fields. An important example of WDVV solutions are the generating function of Gromov--Witten invariants called Gromov--Witten potential. Another source of Dubrovin--Frobenius manifolds comes from Landau--Ginzburg superpotential, which are unfolding of singularities or family of covering over~$\mathbb{CP}^1$. In the analytic theory of Dubrovin--Frobenius manifold, there exist two flat connection. The 1st structure connection is called Dubrovin connection and it is defined below.

\begin{Definition}[{\cite{B.Dubrovin2,B.Dubrovin3}}]
Consider the following deformation of the Levi-Civita connection defined on a Dubrovin--Frobenius manifold $M$,
$
\tilde{\nabla}_{u}v = \nabla_{u}v + zu\bullet v$, $ u,v\in \Gamma(TM)$,
where $\nabla$ represents the Levi-Civita connection of the metric $\eta$, $\bullet$ denotes the Frobenius product, and $z\in \mathbb{CP}^1$. The Dubrovin connection defined in $M\times\mathbb{CP}^1$ is then given by
\begin{gather}
\tilde{\nabla}_{u}v = \nabla_{u}v + zu\bullet v,\qquad
\tilde{\nabla}_{\frac{\rm d}{{\rm d}z}}\frac{\rm d}{{\rm d}z} = 0, \qquad \tilde{\nabla}_{v} \frac{\rm d}{{\rm d}z} = 0,\nonumber\\
\tilde{\nabla}_{\frac{\rm d}{{\rm d}z}}v = \partial_{z}v + E\bullet v - \frac{1}{z}\mu(v).\label{Dubrovin Connection}
\end{gather}
Here, $\mu$ is a diagonal matrix given by
$\mu_{\alpha\beta} = \bigl(q_{\alpha}-\frac{d}{2}\bigr)\delta_{\alpha\beta}$.
\end{Definition}

The deformation of Levi-Civita connection \eqref{Dubrovin Connection} is again a flat connection. In Saito flat coordinates, the Dubrovin connection flat coordinate system, i.e., the solution of
$
\tilde \nabla {\rm d}\tilde t=0$,
 can be written as
 \begin{gather*}
 \bigl(\tilde \nabla_{\alpha} \omega\bigr)_{\beta}=\partial_{\alpha}\omega_{\beta}-zc_{\alpha\beta}^{\gamma}\omega_{\gamma}=0,\\
\bigl(\tilde \nabla_{\frac{\rm d}{{\rm d}z}} \omega\bigr)_{\beta}=\partial_{z}\omega_{\beta}-E^{\sigma}c_{\sigma\beta}^{\gamma}\omega_{\gamma}+\frac{\mu_{\beta}}{z}\omega_{\beta}=0,
\end{gather*}
 where $\omega={\rm d}\tilde t=\omega_{\alpha}{\rm d}t^{\alpha}$ and $\partial_{\alpha}:=\frac{\partial}{\partial t^{\alpha}}$. Alternatively, we could write the Dubrovin flat coordinate system in the matrix form as
\begin{gather}
\partial_{\alpha}\omega=zC_{\alpha}^{\mathsf T}\omega,\qquad
\partial_{z}\omega=\left(\mathcal{U}^{\mathsf T}-\frac{\mu}{z} \right)\omega,\label{Dubrovin connection flat coordinate system in saito coordinates}
\end{gather}
where
\begin{gather}
\omega=\sum_{\beta=1}^n\partial_{\beta}\tilde t {\rm d}t^{\beta}, \qquad C_{\alpha}=\bigl( c_{\alpha\beta}^{\gamma}\bigr), \qquad \mathcal{U}=\bigl( \mathcal{U}_{\beta}^{\gamma} \bigr):=\bigl( E^{\epsilon}c_{\epsilon\beta}^{\gamma} \bigr)\label{gradient flat saito coordinates def}
\end{gather}
or by using the conjugate system
\begin{gather}
\partial_{\alpha}\xi=zC_{\alpha}\xi,\qquad
\partial_{z}\xi=\left(\mathcal{U}+\frac{\mu}{z} \right)\xi,\qquad \text{where}\quad \xi=\eta^{-1}\omega.\label{Dubrovin connection flat coordinate system in saito coordinates conjugate}
\end{gather}

The compatibility of the system \eqref{gradient flat saito coordinates def} is guaranteed by the vanishing of the Riemann tensor associated with the connection \eqref{Dubrovin Connection}.

The connection \eqref{Dubrovin Connection} is more suitable for presenting the Gromov--Witten potential. The second structure connection, also known as the extended Gauss--Manin connection, is more suitable for presenting the Landau--Ginzburg potential. In order to define this connection, we consider the multiplication by the Euler vector field,
\begin{gather}\label{multiplication by the Euler vector field endomorphism}
E\bullet\colon \ \Gamma(TM) \mapsto \Gamma(TM), \qquad X \in \Gamma(TM) \mapsto E\bullet X \in \Gamma(TM).
\end{gather}
Such an endomorphism gives rise to a bilinear form in the sections of the cotangent bundle of~$M$ as follows. Consider $x = \eta(X ), y = \eta(Y ) \in \Gamma(T^*M)$, where $X, Y \in \Gamma(TM)$. An induced Frobenius algebra is defined on $\Gamma(T^*M)$ by
$x \bullet y = \eta(X \bullet Y)$.

\begin{Definition}[{\cite{B.Dubrovin2,B.Dubrovin3}}]
The intersection form is a bilinear pairing in $\Gamma(T^*M)$ defined by ${g^*(\omega_1, \omega_2) = \iota_E(\omega_1 \bullet \omega_2)}$,
where $\omega_1, \omega_2 \in \Gamma(T^*M)$, and $\bullet$ is the induced Frobenius algebra product in $\Gamma(T^*M)$.
\end{Definition}

In the flat coordinates of the Saito metric is given by
\begin{gather}\label{intersection form generic in flat coordinates}
g^{\alpha\beta} = E^{\epsilon}\eta^{\alpha\mu}\eta^{\beta\lambda}c_{\epsilon\mu\lambda}.
\end{gather}

The intersection form $g^{*}$ of a Dubrovin--Frobenius manifold is a flat almost everywhere nondegenerate metric.
The discriminant is the sub manifold such that the intersection form is degenerate
\begin{gather}\label{discriminant}
 \Sigma=\{t\in M\mid \det( g)=0 \} .
\end{gather}

 Hence, the flat coordinate system of the intersection form in Saito flat coordinates
 \begin{gather}\label{Gauss--Manin connection}
 g^{\alpha\epsilon}(t)\frac{\partial^2 x}{\partial t^{\beta}\partial t^{\epsilon}}+\Gamma_{\beta}^{\alpha\epsilon}(t)\frac{\partial x}{\partial t^{\epsilon}}=0
 \end{gather}
 has poles in \eqref{discriminant}, and consequently its solutions $x_a\bigl(t^1,\dots,t^n\bigr)$ are multivalued. The meromorphic connection \eqref{Gauss--Manin connection} is called Gauss--Manin connection of the Dubrovin--Frobenius manifold. The analytical continuation of the solutions $x_a\bigl(t^1,\dots,t^n\bigr)$ has monodromy corresponding to loops around $\Sigma$. This gives rise to a monodromy representation of
 \begin{gather}\label{monodromy on the moduli}
 \pi_1(M\setminus\Sigma)\mapsto {\rm GL}(\mathbb{C}^n),
 \end{gather}
 which is called monodromy of the Dubrovin--Frobenius manifold. Moreover, we can extend the connection \eqref{Gauss--Manin connection} to a connection on $M\times \mathbb{CP}^1$ as follows:
 \begin{gather}\label{Extended Gauss--Manin connection}
\left( g^{\alpha\epsilon}(t)-\lambda\eta^{\alpha\epsilon}\right)\frac{\partial^2 x}{\partial t^{\beta}\partial t^{\epsilon}}+\Gamma_{\beta}^{\alpha\epsilon}(t)\frac{\partial x}{\partial t^{\epsilon}}=0.
 \end{gather}
 The system \eqref{Extended Gauss--Manin connection} is isomonodromic, then its monodromy representation
 \begin{gather}\label{monodromy of GM on the tolal bundle}
 \pi_1(M\times \mathbb{C}\setminus\Sigma_{\lambda})\mapsto {\rm GL}(\mathbb{C}^n), \qquad \text{where} \quad \Sigma_{\lambda}=\det(\lambda-E_{\bullet})
 \end{gather}
is isomorphic to \eqref{monodromy on the moduli}. This fact is straightforward to check once we realise that if $x_{a}\bigl(t^1,t^2,\dots,\allowbreak t^n\bigr)$ is a solution of \eqref{monodromy on the moduli}, then $x_{a}\bigl(t^1-\lambda,t^2,\dots,t^n\bigr)$ is a solution of \eqref{Extended Gauss--Manin connection}. The meromorphic connection \eqref{Extended Gauss--Manin connection} is called extended Gauss--Manin connection.
 Summarising,
 \begin{Definition}[\cite{B.Dubrovin2,B.Dubrovin3}]
 Let $\tilde\nabla$ be the Levi-Civita connection of the intersection form \eqref{multiplication by the Euler vector field endomorphism}. Then
the extended Gauss--Manin connection is connection in $M\times\mathbb{CP}^1$ given by
\begin{gather*}
\tilde\nabla_{\frac{\partial }{\partial t^{\alpha}} }=\frac{\partial }{\partial t^{\alpha}}+(\lambda-E_{\bullet} )^{-1}C_{\alpha}\left(\frac{1}{2}+\mu \right) ,\\
\tilde\nabla_{\frac{\partial }{\partial \lambda} }=\frac{\partial }{\partial \lambda}-(\lambda-E_{\bullet} )^{-1}\left(\frac{1}{2}+\mu \right) .
\end{gather*}
\end{Definition}

The solutions of the flat coordinate system of \eqref{Dubrovin Connection} and \eqref{Extended Gauss--Manin connection} are related by a Fourier--Laplace transform, see Lemma~\ref{Main lemma Mirror symmetry 1}.

\begin{Remark}
The analytic continuation of solutions of the flat coordinates system \eqref{Gauss--Manin connection} and \eqref{Extended Gauss--Manin connection} along any path of $M\setminus{\Sigma}$ and $M\times\mathbb{CP}^1\setminus{\Sigma_{\lambda}}$, respectively,
\begin{gather}
t=\bigl(t^1,t^2,\dots,t^n\bigr)\mapsto (x_1(t),\dots,x_n(t)),\nonumber\\
(\lambda,t)=\bigl(\lambda,t^1,t^2,\dots,t^n\bigr)\mapsto (x_1(\lambda,t),\dots,x_n(\lambda,t))\label{generic period map and extended period map}
\end{gather}
are called period map and extended period map, respectively.
\end{Remark}

 \begin{Remark}
The solutions of flat coordinate system of \eqref{Gauss--Manin connection} and \eqref{Extended Gauss--Manin connection} are quasi-homogeneous in the following sense:
 \begin{gather}\label{quasi homogeneous condition of the solutions xa and xalambda}
 E(x_a(t))=\frac{1-d}{2}x_a(t),\qquad
\left( \lambda\frac{\partial}{\partial \lambda}+E\right)(x_a(\lambda,t))=\frac{1-d}{2}x_a(\lambda,t).
 \end{gather}
 Moreover, let $g^{ab}$ be the coefficients of the intersection form \eqref{intersection form generic in flat coordinates} in its flat coordinates. Then, there is a polynomial relation between Saito flat coordinates $\bigl(t^1,t^2,\dots,t^n\bigr)$ and intersection form flat coordinates $(x_1,x_2,\dots,x_n)$ given by
 \begin{gather}\label{quadratic relation between Saito and Intersection form flat coordinates}
 t_1:=\eta_{1\alpha}t^{\alpha}=g_{ab}x_ax_b, \qquad \text{where}\quad (g_{ab})=\bigl(g^{ab}\bigr)^{-1}.
 \end{gather}
 See in \cite[Appendix~G, Exercise~G.1]{B.Dubrovin2}.
 \end{Remark}

Recall that a point in a Dubrovin--Frobenius manifold is called semisimple if the Frobenius algebra in $T_pM$ is semisimple. It is worth noting that semisimplicity constitutes an open condition. The Dubrovin--Frobenius structure around semisimple points becomes rather simple. Specifically, the Frobenius algebra becomes trivial. Moreover, both the Saito metric and the endomorphism resulting from the multiplication by the Euler vector field~\eqref{multiplication by the Euler vector field endomorphism} are diagonal around such a point. 

\begin{Proposition}[\cite{B.Dubrovin2,B.Dubrovin3}]
Let $(u_1,u_2,\dots,u_n)$ be pairwise distinct roots of the characteristic equation
\begin{gather}\label{spectral curve}
\det\bigl(g^{\alpha\beta} -u\eta^{\alpha\beta} \bigr) = 0.
\end{gather}
Then, the relation $(u_1(t),u_2(t),\dots,u_n(t))$ can serve as local coordinates, which are called canonical coordinates. In these coordinates, the Frobenius multiplication, Saito metric, unit vector field and Euler vector field can be written as
\begin{gather}
\frac{\partial}{\partial u_i}\bullet \frac{\partial}{\partial u_j}=\delta_{ij}\frac{\partial}{\partial u_i},\qquad \eta=\sum_{i=1}^{n}\psi_{i1}^2 ({\rm d}u_i)^2, \qquad e=\sum_{i=1}^n\frac{\partial}{\partial u_i},\qquad E=\sum_{i=1}^n u_i\frac{\partial}{\partial u_i},\label{Dubrovin Frobenius structure in canonical coordinates 1}
\end{gather}
where the matrix $\Psi=( \psi_{i\alpha})$ is given by
$
\psi_{i\alpha}=\psi_{i1}\frac{\partial u_i}{\partial t^{\alpha}}$.
\end{Proposition}
At this stage, we can define a Landau--Ginzburg superpotental, which can be found in \cite[Definition~5.7]{B.Dubrovin3}.

\begin{Definition}[\cite{B.Dubrovin3}]\label{Landau--Ginzburg superpotential definition}
Let $D$ be an open domain of a Riemann surface. A Landau--Ginzburg superpotential associated with a Dubrovin--Frobenius manifold $M$ of dimension $n$ consists of a~function $\lambda(p,u)$ on $D \times M$ and an Abelian differential $\phi$ in $D$ satisfying
\begin{itemize}\itemsep=0pt
\item The critical values of $\lambda(p,u)$ are the canonical coordinates $(u_1,\dots,u_n)$. In other words, the canonical coordinates $(u_1,u_2,\dots ,u_n)$ are defined by the following system
$\lambda(p_i)=u_i$, $\frac{{\rm d} \lambda}{{\rm d}p}( p_i)=0$.
\item For some cycles $Z_1,\dots,Z_n$ in $D$ the integrals
\begin{gather}\label{Mirror symmetry}
\tilde t_j(z,u)=\frac{1}{z^{\frac{3}{2}}}\int_{Z_j}{\rm e}^{z\lambda(p)}\phi, \qquad j=1,\dots,n,
\end{gather}
converges and give a system of independent flat coordinates for the Dubrovin connection~$\tilde\nabla$ in canonical coordinates, i.e., the matrix
$Y=\Psi\eta^{-1}\omega:=\bigl(\psi_{i\alpha}\eta^{\alpha\beta}\partial_{\beta}\tilde t \bigr)$
is a solution of the following system
\begin{gather}\label{Dubrovin connection in canonical coordinates 1}
\frac{\partial Y}{\partial u_i}=( zE_i+V_i )Y,\qquad
\frac{{\rm d}Y}{\partial z}=\left( U+\frac{V}{z} \right)Y,
\end{gather}
where
\begin{gather}\label{U,V, Vi in canonical coordinates}
U=\Psi\mathcal{U}\Psi^{-1}, \qquad V=\Psi\mu\Psi^{-1}, \qquad E_i=( \delta_{ij}\delta_{ik} ), \qquad V_i:=\frac{\partial \Psi}{\partial u_i}\Psi^{-1}.
\end{gather}

\item
The following expressions for the coefficients of the tensors Saito metric $\eta$, intersection form $g^{*}$ and the structure constants $c$, in any coordinate system $(x_1,x_2,\dots,x_n)$, holds true
\begin{gather}
\eta_{ij}=\sum \res_{{\rm d}\lambda=0} \frac{ \partial_i\lambda \partial_j\lambda }{d_p\lambda }\phi, \qquad
g_{ij}=\sum \res_{{\rm d}\lambda=0} \frac{ \partial_i\log\lambda \partial_j\log\lambda }{d_p\log\lambda } \phi,\nonumber\\
c_{ijk}=\sum \res_{{\rm d}\lambda=0} \frac{ \partial_i\lambda \partial_j\lambda \partial_k\lambda }{d_p\lambda }\phi.\label{residue expression for eta, intersection form and structure constants}
\end{gather}
\end{itemize}
\end{Definition}

\subsection{Problem setting}

Let $N_k$ be the number of rational curves $\mathbb{CP}^1 \rightarrow \mathbb{CP}^2$ of degree $k$ passing through $3k-1$ generic points. Kontsevich \cite{M.Kontsevich} proved that the generating function
\begin{gather}\label{generating function of Gromov Witten CP2}
F\bigl(t^1, t^2, t^3\bigr) = \frac{\bigl(t^1\bigr)^2t^3}{2} + \frac{\bigl(t^2\bigr)^2t^1}{2} + \sum_{k=0}^{\infty} \frac{N_k}{(3k-1)!}{\rm e}^{kt^2}\bigl(t^3\bigr)^{3k-1}
\end{gather}
satisfies a WDVV equation. Moreover, in \cite{M.Kontsevich} the Gromov--Witten invariants $N_k$ were computed as follows. Note that the associativity equation for any 3d Dubrovin--Frobenius manifold is given~by
\begin{gather}\label{associativity equation generic}
c_{223}^2 = c_{333} + c_{222}c_{233}.
\end{gather}
Defining the function
\begin{gather}\label{main Gromov Witten potential}
\Phi(X) = \sum_{k=0}^{\infty} \frac{N_k}{(3k-1)!}{\rm e}^{kX}, \qquad X := t^2 + 3\ln t^3,
\end{gather}
we obtain the following third-order differential equation for the function $\Phi(X)$ by substituting~\eqref{generating function of Gromov Witten CP2} and \eqref{main Gromov Witten potential} into \eqref{associativity equation generic},
\begin{gather}\label{associativity equation}
-6\Phi + 33\Phi' - 54\Phi'' - \bigl(\Phi''\bigr)^2 + \Phi''' \bigl( 27 + 2\Phi' - 3\Phi'' \bigr) = 0.
\end{gather}
Hence, using the Taylor series expansion of \eqref{main Gromov Witten potential} in \eqref{associativity equation}, we obtain the following recursion:
\begin{gather*}
N_d = \sum_{m=1}^{d-1} \left[ {{3d-4}\choose{3m-2}} m^2(d-m)^2 - {{3d-4}\choose{3m-3}} m(d-m)^3 \right]N_mN_{d-m},
\end{gather*}
which explicitly provides all Gromov--Witten $N_k$ by setting $N_1=1$. In principle, the Gromov--Witten potential \eqref{main Gromov Witten potential} is only a formal power series. However, in \cite{DiFrancesco} Di Francesco and Itzykson derived the following asymptotic behavior for~$N_k$:
\begin{gather*}
\frac{N_k}{(3k-1)!} = ba^kk^{-\frac{7}{2}} \left(1 + O\left(\frac{1}{k}\right)\right),
\end{gather*}
where $a$ and $b$ are numerically estimated as
$a = 0.138$, $ b = 6.1$,
which implies that the radius of convergence of \eqref{main Gromov Witten potential} is given by $\frac{1}{a}$. Then, the analytic domain of \eqref{generating function of Gromov Witten CP2} is given by
\begin{gather}\label{WDVV domain}
D = \left\{\bigl(t^1, t^2, t^3\bigr) \in \mathbb{C}^3 \mid \big|{\rm e}^{t^2}\bigl(t^3\bigr)^3\big| < \frac{1}{a}\right\}.
\end{gather}
Furthermore, the Euler vector field has the following form
$E = t^1\partial_1 + 3\partial_2 - t^3\partial_3$.

Our aim is to investigate the analytic properties of \eqref{generating function of Gromov Witten CP2} by constructing the corresponding 1D Landau--Ginzburg superpotential from the Gauss--Manin connection of $QH^{*}\bigl(\mathbb{CP}^2\bigr)$.

\subsection{Main results}
In order to simplify the problem, it is convenient to restrict to the affine small quantum cohomology locus, which is the sub manifold of \eqref{WDVV domain} defined by
\begin{gather*}
D=\left\{ \bigl(t^1,t^2,0\bigr) \in \mathbb{C}^2\times D\left(0,\frac{1}{a}\right) \right\}\cong \mathbb{C}^2.
\end{gather*}
Here the adjective affine means an affine extension of the standard small quantum cohomology locus
\begin{gather*}
D^0=\left\{ \bigl(0,t^2,0\bigr) \in \mathbb{C}^2\times D\left(0,\frac{1}{a}\right) \right\}\cong \mathbb{C}.
\end{gather*}
Moreover, we consider the local change of coordinates
$
Q\colon \mathbb{C}\mapsto \mathbb{C}^{*}$, $ t^2\mapsto Q={\rm e}^{t^2}$.
In this setting, we are able to state our first result.

\begin{Theorem}\quad
\begin{enumerate}\itemsep=0pt
\item[$(1)$] The Landau--Ginzburg superpotential of small affine quantum cohomology of $\mathbb{CP}^2$ is a~family of holomorphic functions
\smash{$
\lambda\bigl(\tilde\tau, t^1, Q^{\frac{1}{3}}\bigr)\colon \mathbb{H} \mapsto \mathbb{C}
$}
with holomorphic dependence in the parameter space
\begin{gather}\label{parameter space small affine quantum cohomology}
\bigl(t^1,Q^{\frac{1}{3}}\bigr)\in \mathbb{C}\times\mathbb{C}^{*}
\end{gather}
 and given by
\begin{gather}\label{lg superpotential of small quantum cohomology j function intro}
\lambda\bigl(\tilde\tau, t^1, Q\bigr)=t^1+3Q^{\frac{1}{3}}J^{\frac{1}{3}}(\tilde\tau),
\end{gather}
where
\[
J(\tilde\tau)=\frac{E_4^3(\tilde\tau)}{E_4^3(\tilde\tau)-E_6^2(\tilde\tau)}
\]
is the $j$-function and $E_k(\tilde\tau)$ are the Eisenstein series of weight~$k$.
In addition, the correspondent Abelian differential $\phi$ is given by
\[
\phi=\frac{\Delta^{\frac{1}{6}}(\tilde\tau)}{Q^{\frac{1}{6}}}{\rm d}\tilde \tau,
\]
where the function
\begin{gather*}
\Delta(\tilde\tau)=(2\pi)^{12}q\prod_{n=1}^{\infty}(1-q^n)^{24}, \qquad q={\rm e}^{2\pi {\rm i}\tilde\tau},
\end{gather*}
is the modular discriminant.
\item[$(2)$] The Landau--Ginzburg superpotential of small affine quantum cohomology of $\mathbb{CP}^2$ is $\Gamma^{(3)}$-invariant, where
\begin{gather}\label{Group gamma 3 intro}
\Gamma^{(3)}=\left\{ \begin{pmatrix}
a & b\\
c& d\\
\end{pmatrix} \in {\rm SL}_2(\mathbb{Z}) \mid \begin{pmatrix}
0 & *\\
*& 0\\
\end{pmatrix} \text{ or } \begin{pmatrix}
* & b\\
b& *\\
\end{pmatrix} \mod 3\right\}
\end{gather}
 is a free Coxeter group with $3$ generators.
\end{enumerate}
\end{Theorem}

After an affine transformation the Landau--Ginzburg superpotential \eqref{lg superpotential of small quantum cohomology j function intro} is the Hauptmodul of the group \eqref{Group gamma 3 intro}. In particular, this group generate the monodromy of the Gauss--Manin connection associated to quantum cohomology of $\mathbb{CP}^2$. The Saito flat coordinates $\bigl(t^1,Q\bigr)$ lives in the parameter space \eqref{parameter space small affine quantum cohomology}, which describes the deformations of Hauptmodul of $\Gamma^{(3)}$ by affine and rescaling transformations, respectively.
See details in Section~\ref{LG superpotential for small Quantum Cohomology}.

In addition, we consider a $t^3$-isomonodromic deformation of \eqref{lg superpotential of small quantum cohomology j function intro}. In essence, the Landau--Ginzburg superpotential for big quantum cohomology of $\mathbb{CP}^2$ represents a power series expansion in the variable $Q^{\frac{1}{3}}t^3$, where its coefficients are quasi-modular forms. We present this result below without delving into specific details on obtaining the coefficients of the expansion. For further information, refer to Section~\ref{Milanov deformation}.

\begin{Theorem}\quad
\begin{enumerate}\itemsep=0pt
\item[$(1)$] The Landau--Ginzburg superpotential of big quantum cohomology of $\mathbb{CP}^2$ is a family of functions
\smash{$
\lambda\bigl(\tilde\tau_{12}, t^1,Q^{\frac{1}{3}}, Q^{\frac{1}{3}}t^3\bigr)\colon \mathbb{H} \mapsto \mathbb{C}
$}
with holomorphic dependence in the parameter space
\begin{gather*}
\left\{\bigl(t^1,Q^{\frac{1}{3}}, Q^{\frac{1}{3}}t^3\bigr) \in \mathbb{C}\times\mathbb{C}^{*}\times \mathbb{C}\mid \big|Q^{\frac{1}{3}}t^3\big|<\left(\frac{1}{a}\right)^{\frac{1}{3}} \right\},
\\
\lambda\bigl(\tilde\tau_{12}, t^1, Q, t^3\bigr)=t^1+3Q^{\frac{1}{3}}J^{\frac{1}{3}}\bigl(\tilde\tau_{12}, Q^{\frac{1}{3}}t^3\bigr),
\end{gather*}
where its pullback via the extended period map \eqref{generic period map and extended period map} coincide with the coordinate $\lambda$.
Here~\smash{${J^{\frac{1}{3}}\bigl(\tilde\tau_{12}, Q^{\frac{1}{3}}t^3\bigr)}$} is a $t^3$-deformation of \smash{$J^{\frac{1}{3}}(\tilde\tau_{12})$} in the following sense:
\[
J^{\frac{1}{3}}\bigl(\tilde\tau_{12}, Q^{\frac{1}{3}}t^3\bigr)=J^{\frac{1}{3}}(\tilde\tau_{12})+\sum_{n=1}^{\infty}J_n^{\frac{1}{3}}(\tilde\tau_{12})\bigl( Q^{\frac{1}{3}}t^3\bigr)^n.
\]
The coefficients \smash{$J_n^{\frac{1}{3}}(\tilde\tau_{12})$} belong to the following ring
\smash{$
J_n^{\frac{1}{3}}(\tilde\tau_{12}) \in \Delta^{\frac{-n}{3}}\mathbb{C}[E_2,E_4,E_6]$}.

 In addition, the correspondent Abelian differential $\phi$ is given by
\[
\phi=\frac{\Delta^{\frac{1}{6}}\bigl(\tilde\tau_{12},Q^{\frac{1}{3}}t^3\bigr)}{Q^{\frac{1}{6}}} {\rm d}\tilde\tau_{12},
\]
where \smash{$\Delta\bigl(\tilde\tau_{12}, Q^{\frac{1}{3}}t^3\bigr)$} is a suitable $t^3$-deformation of $\Delta(\tilde\tau_{12})$ in the following sense:
\begin{gather*}
\Delta\bigl(\tilde\tau_{12},Q^{\frac{1}{3}}t^3\bigr)=\Delta(\tilde\tau_{12})+\sum_{n=1}^{\infty}\Delta_n(\tilde\tau_{12})\bigl( Q^{\frac{1}{3}}t^3\bigr)^n.
\end{gather*}
\item[$(2)$] The Landau--Ginzburg superpotential of big quantum cohomology of $\mathbb{CP}^2$ is $\Gamma^{(3)}$-invariant, under a suitable action.
\end{enumerate}
\end{Theorem}

In order to elucidate the geometric nature of the parameter space \eqref{parameter space small affine quantum cohomology}, we use an appropriate change of coordinates in the context of affine small quantum cohomology. More precisely, we utilize the following factorization of the Hauptmodul of $\Gamma^{(3)}$, which we denoted by Cohn identities~\cite{M.Batchelor,Cohn,Golubev}:
\begin{gather}
1-J(\tau)=4\wp^3\bigl(v(\tau),{\rm e}^{\frac{2\pi {\rm i}}{6}}\bigr)+1,\qquad
(v'(\tau))^6=\Delta(\tau).\label{universal covering map equiharmonic lattice in intro}
\end{gather}
Here, $v(\tilde{\tau})$ represents the universal covering of $\mathbb{C}\setminus\bigl\{\mathbb{Z}\oplus {\rm e}^{\frac{\pi {\rm i}}{3}}\mathbb{Z} \bigr\}$, and $\wp(v,z)$ denotes the Weierstrass~$\wp$ function. Further details can be found in Section~\ref{Cohn interpretation}. By employing the change of coordinates~\eqref{universal covering map equiharmonic lattice in intro} in the small superpotential \eqref{lg superpotential of small quantum cohomology j function intro}, we derive the following theorem.

\begin{Theorem}
The Landau--Ginzburg superpotential of small affine quantum cohomology \linebreak of~$\mathbb{CP}^2$ is a family of functions
$\lambda(\tilde\tau, v_0, \omega)\colon \mathbb{H} \mapsto \mathbb{C}$
with holomorphic dependence in the parameter space
\smash{$ (v_0,\omega) \in \mathbb{C}\setminus\bigl\{\mathbb{Z}\oplus {\rm e}^{\frac{\pi {\rm i}}{3}}\mathbb{Z} \bigr\}\times\mathbb{C}^{*}$}
 and given by
\begin{gather*}
\lambda(\tilde\tau,v_0,\omega)=\frac{\wp\bigl(v(\tilde\tau),{\rm e}^{\frac{\pi {\rm i} }{3}}\bigr)}{(2\omega)^2}-\frac{\wp\bigl(v_0,{\rm e}^{\frac{\pi {\rm i} }{3}}\bigr)}{(2\omega)^2},
\end{gather*}
where $v(\tilde\tau)$ is the universal covering of $\mathbb{C}\setminus\bigl\{\mathbb{Z}\oplus {\rm e}^{\frac{\pi {\rm i}}{3}}\mathbb{Z} \bigr\}$.
 In addition, the correspondent Abelian differential $\phi$ is given by
$
\phi=2\omega {\rm d}v(\tilde\tau)$.

\end{Theorem}

The elliptic curve associated with the lattice \smash{$\mathbb{Z}\oplus {\rm e}^{\frac{2\pi {\rm i}}{6}}\mathbb{Z}$},
$y^2=4\bigl(x-t^1\bigr)^3-27Q$, is referred to as an equianharmonic elliptic curve.
Subsequently, we establish a change of coordinates between \smash{$Q^{\frac{1}{3}}t^3 \in D\bigl(0,\frac{1}{a}\bigr)$} and $\smash{z \in U\bigl({\rm e}^{\frac{2\pi {\rm i}}{3}}\bigr)}\allowbreak \subset \mathbb{H}$, where \smash{$U\bigl({\rm e}^{\frac{2\pi {\rm i}}{3}}\bigr)$} represents a suitable neighbourhood of~\smash{${\rm e}^{\frac{2\pi {\rm i}}{3}}$}. As a consequence of this, we~can derive the following theorem.

\begin{Theorem}
The Landau--Ginzburg superpotential of big quantum cohomology of $\mathbb{CP}^2$ is a~family of functions
$\lambda(\tilde\tau_{12}, \tau_{12}, \tilde\omega,z)\colon \mathbb{H} \mapsto \mathbb{C}$
with holomorphic dependence in the parameter space
\begin{gather}\label{parameter space of elliptic curves near equiharmonic lattice}
\bigl\{(\tau_{12},\tilde\omega,z) \in \mathbb{H}\times\mathbb{C}^{*}\times\mathbb{H} \mid z \in U\bigl({\rm e}^{\frac{2\pi {\rm i}}{3}}\bigr) \bigr\}
\end{gather}
and given by
\[
\lambda(v(\tilde\tau_{12},z),\tau_{12},\tilde\omega,z)=\frac{\wp(v(\tau_{12},z),z)}{\left(2\tilde\omega\right)^2}-\frac{\wp(v(\tilde\tau_{12},z),z)}{\left(2\tilde\omega\right)^2},
\]
where, for $z\in \mathbb{H}$ close enough to \smash{${\rm e}^{\frac{2\pi {\rm i}}{3}}$}, $v(\tilde\tau_{12},z)$ is the universal covering of $\mathbb{C}\setminus\{\mathbb{Z}\oplus z\mathbb{Z} \}$.
 In addition, the correspondent Abelian differential $\phi$ is given by
$\phi= 2\tilde\omega {\rm d}v(\tilde\tau_{12},z)$,
where $\Delta(\tilde\tau_{12},z)$ is a suitable deformation of the modular discriminant.
\end{Theorem}

From this perspective, the domain of the WDVV solution \eqref{generating function of Gromov Witten CP2} could be identified as the parameter space of a family of elliptic functions following a change of coordinates. The parameters~$t^1$ and $Q$ are associated with $v_0$ and $\omega$, representing the affine and rescaling freedom within the elliptic family \eqref{parameter space of elliptic curves near equiharmonic lattice}. Additionally, the $t^3$ deformation corresponds to \smash{$z \in U\bigl({\rm e}^{\frac{2\pi {\rm i}}{3}}\bigr) \subset \mathbb{H}$}, functioning as deformations of the equianharmonic lattice \smash{$\mathbb{Z}\oplus {\rm e}^{\frac{\pi {\rm i} }{3}}\mathbb{Z}$}. Further details can be found in Section \ref{Geometric Isomonodromic deformation}.

Starting from a genus $0$ Gromov--Witten theory applied to suitable K\"ahler manifolds, we observe that the Gromov--Witten potential satisfies a WDVV equation. Consequently, its domain exhibits a Dubrovin--Frobenius structure. This construction is recognized in the realm of the geometry of topological field theories as the A-model or big quantum cohomology. Alternatively, starting from a Landau--Ginzburg superpotential, one can establish a Dubrovin--Frobenius structure within the parameter space of the LG superpotential. This is accomplished through the utilization of Grothendieck residues \eqref{residue expression for eta, intersection form and structure constants}. Such a construction is commonly referred to as LG models or the B model. In this setting, mirror symmetry is a Dubrovin--Frobenius manifold isomorphism between quantum cohomology and LG models.

Restricting to a specific sublocus of big quantum cohomology known as small quantum cohomology, Givental derived several mirror symmetry statements in \cite{Givental3,Givental2}. In particular, there exists a mirror symmetry between the small quantum cohomology of $\mathbb{CP}^n$ and the LG superpotential, represented by the Laurent polynomial
\[
\lambda=\sum_{i=1}^n x_i + \frac{Q}{\prod_{i=1}^n x_i}.
\]
 This result is elucidated in Givental's notes \cite{Givental}. Furthermore, Barannikov, in \cite{Barannikov}, developed an LG model corresponding to the big quantum cohomology of $\mathbb{CP}^n$. Additionally, Douai and Sabbah, in \cite{DouaiSabbah1} and \cite{DouaiSabbah2}, systematically constructed a theory to derive LG models for a wide array of examples, notably including the case studied by Barannikov.

The primary objective of this manuscript is to establish a 1D Landau--Ginzburg (LG) superpotential for both small and big quantum cohomology of $\mathbb{CP}^2$. This signifies that the LG superpotential's domain is one-dimensional, differing from the two-dimensional domains of the Givental LG superpotential for small quantum cohomology and the Barannikov LG superpotential for big quantum cohomology. Our approach relies on the inversion of the period map~\eqref{generic period map and extended period map}, wherein the analytic properties of the associated LG superpotential are contingent on the analytical behavior of the period map governed by the given Gauss--Manin connection. This method circumvents the issue of non-isolated singularities for sufficiently large $t^3$. Despite encountering singular behavior for sufficiently large $t^3$ in our framework, it appears that a more lucid geometric interpretation of these singularities is feasible in terms of the degeneration of the universal covering of $\mathbb{C}\setminus\{\mathbb{Z}\oplus z\mathbb{Z}\}$. Further exploration of this problem will be the focus of subsequent publications.

Our primary motivation is to explore the Gromov--Witten potential with good analytic properties. In \cite[Appendix A]{B.Dubrovin2}, Dubrovin conjectured that the Dubrovin--Frobenius structure of these Gromov--Witten potentials has a monodromy associated with certain reflection groups or their generalizations.

From the standard theory of Dubrovin--Frobenius manifolds, we know that the coefficients of the intersection form in Saito flat coordinates are essentially the Hessian of the corresponding WDVV solution. On the other hand, the intersection form is invariant with respect to the monodromy of the associated Gauss--Manin connection. Therefore, the WDVV solution and the inverse period map can be expressed in terms of monodromy-invariant functions.

In \cite{Hertling}, Hertling proved that a particular class of polynomial WDVV solutions is in one-to-one correspondence with the orbit space of a finite Coxeter group. Examples of WDVV solutions associated with orbit spaces of natural extensions of finite Coxeter groups can be found in \cite{Almeida1,Almeida2, Bertola, Bertola2,B.Dubrovin1,B.Dubrovin34,Zuo}. In particular, the extended affine Weyl groups and Jacobi groups are extensions of finite Coxeter groups, and the corresponding rings of invariant functions are trigonometric functions and Jacobi forms, respectively.

Motivated by these examples, we aim to study in detail the ring of invariant functions related to the quantum cohomology of $\mathbb{CP}^2$. The inverse period map constructed by Milanov in \cite{Milanov} serves as a good example of invariant functions for $QH^{*}\bigl(\mathbb{CP}^2\bigr)$. However, we need to understand how large the period domain found by Milanov is. We expect that the investigation of the 1D Landau--Ginzburg (LG) superpotential for the quantum cohomology of $\mathbb{CP}^2$ would be the first step in understanding the invariant functions of $QH^{*}\bigl(\mathbb{CP}^2\bigr)$.

The inverse period map of $QH^{*}\bigl(\mathbb{CP}^2\bigr)$ possesses intriguing arithmetic significance as it transforms like a Hilbert modular form under a diagonal action of $A_1\times {\rm PSL}_2(\mathbb{Z})$. Furthermore, its Taylor expansion exhibits similarities to Jacobi forms. An interesting avenue of investigation is to explore whether these functions can be derived through methods akin to Cohen--Kuznetsov series.

This paper is organized in the following way. In Section \ref{Construction of Dubrovin superpotential}, we revisit the Dubrovin construction of the Landau--Ginzburg superpotential associated with a Dubrovin--Frobenius manifold, as done in \cite{B.Dubrovin2}. In Section \ref{small Quantum Cohomolology}, we recapitulate the inverse period map for affine small quantum cohomology as presented in \cite{Milanov} and apply this result to derive a Landau--Ginzburg superpotential for affine small quantum cohomology. Moreover, we reinterpret the Landau--Ginzburg superpotential for affine small quantum cohomology as a composition between the Weierstrass~$\wp$ function and the universal covering of \smash{$\mathbb{C}\setminus\bigl\{\mathbb{Z}\oplus {\rm e}^{\frac{\pi {\rm i}}{3}}\mathbb{Z}\bigr\}$} using Cohn identities. In Section \ref{Big Quantum Cohomology CP2}, we revise the Milanov deformations as discussed in \cite{Milanov} and apply these results to obtain a Landau--Ginzburg superpotential for the big quantum cohomology of $\mathbb{CP}^2$. Furthermore, we derive a~change of coordinates using the results of \cite{Guzzetti} to rewrite the Landau--Ginzburg superpotential of big quantum cohomology of $\mathbb{CP}^2$ as a composition between the Weierstrass $\wp$ function and the universal covering of $\mathbb{C}\setminus\{\mathbb{Z}\oplus z\mathbb{Z}\}$ and derive an isomonodromic deformation of Cohn identities.

\section{Construction of Dubrovin superpotential}\label{Construction of Dubrovin superpotential}

The objective of this section is to review the Dubrovin construction of the superpotential associated with any WDVV solution in \cite[Appendix J]{B.Dubrovin2} and \cite[Chapter 5]{B.Dubrovin3}. Additionally, we will employ this approach to deduce the Landau--Ginzburg superpotential for the big quantum cohomology of $\mathbb{CP}^2$. The central element of this construction is the inverse period map of the extended Gauss--Manin connection \eqref{Extended Gauss--Manin connection}.

\subsection{Monodromy of Dubrovin--Frobenius manifolds}

In this subsection, we describe the monodromy associated with the two flat meromorphic connections of a Dubrovin--Frobenius manifold and their relationship.
The Dubrovin connection~\eqref{Dubrovin Connection} is a meromorphic connection with a regular singularity at~${z=0}$ and an irregular singularity at $z=\infty$. The fundamental solution of the flat coordinate system of~\eqref{Dubrovin connection flat coordinate system in saito coordinates conjugate} in Saito flat coordinates near $z=0$ has the following form:
\begin{gather}
 \xi_0(z,u)=\bigl(\eta^{\alpha\gamma}\partial_{\gamma}\tilde t_{\beta}\bigr)=\left(I+\sum_{n=1}^{\infty}H_n(t)z^n \right)z^{\mu}z^R, \qquad \tilde t_{\beta}=\eta_{\beta\epsilon}\tilde t^{\epsilon},\label{solution of Dubrovin connection near 0}
\end{gather}
or alternatively in canonical coordinates
\begin{gather}\label{solution of Dubrovin connection near 0 canonical coordinates}
 Y_0(z,u)=\left(\Psi+\sum_{n=1}^{\infty}\Psi_n(u)z^n \right)z^{\mu}z^R,
\end{gather}
where $R_{\alpha\beta}=0$ if $\mu_{\alpha}-\mu_{\beta} \neq k>0$, $k\in \mathbb{N}$.

The solutions~\eqref{solution of Dubrovin connection near 0} and~\eqref{solution of Dubrovin connection near 0 canonical coordinates} are related by the following Gauge transformation
$ Y_0=\Psi \xi_0$.
The coefficients $\Psi_k(u)$ in~\eqref{solution of Dubrovin connection near 0 canonical coordinates} are obtained by substituting~\eqref{solution of Dubrovin connection near 0 canonical coordinates} in~\eqref{Dubrovin connection in canonical coordinates 1}, which concretely gives
\[
 \frac{\partial \Psi_k}{\partial u_i}=E_i\Psi_{k-1}+V_i\Psi_k.
 \]
The monodromy of the solution~\eqref{solution of Dubrovin connection near 0} is given by
\[
Y_0\bigl(z{\rm e}^{2\pi {\rm i}},u\bigr)=Y_0(z,u)M,
\]
 where \smash{$ M={\rm e}^{2\pi {\rm i}(\mu+R)}$}.

To analyze the analytic behavior at $\infty$, consider a series
$\sum_{n=0}^{\infty} \frac{a_n}{z^n}$
as an asymptotic expansion of the function $f(z)$ for
$
\lvert z \rvert \to \infty$ in the sector $\alpha < \arg(z) < \beta$,
if for any $n$
\[
z^n\left[f(z)- \sum_{n=0}^{n} \frac{a_n}{z^n} \right] \to 0, \qquad \text{as} \quad \lvert z \rvert \to \infty,
\]
uniformly in the sector $ \alpha+\epsilon < \arg(z) < \beta-\epsilon$.
We denote this asymptotic expansion by{\samepage
\[
f(z)\sim\sum_{n=0}^{\infty} \frac{a_n}{z^n} , \qquad \text{as} \quad \lvert z \rvert \to \infty,
\]
uniformly in the sector $\alpha+\epsilon < \arg(z) < \beta-\epsilon$.}

Moreover, a line \(l\) through the origin in the complex \(z\)-plane is called admissible for the system~\eqref{Dubrovin Connection} if
\smash{$
\operatorname{Re}(z(u_i-u_j))|_{z\in l\setminus\{0\}} \neq 0$} for any $ i \neq j$.
Fixing an admissible line \(l\) with slope~\(\phi\) and its respective orientations \(l_{+}\) and \(l_{-}\),
$
 l_{+} = \{ z \in \mathbb{C}\mid \arg(z) = \phi\}$, $
 l_{-} = \{ z \in \mathbb{C}\mid \arg(z) = \phi-\pi\}$,
and the sectors
\[
 \Pi_{\text{right}} = \phi-\pi-\epsilon < \arg(z) < \phi+\epsilon,\qquad
 \Pi_{\text{left}} = \phi-\epsilon < \arg(z) < \phi+\epsilon+\pi,
\]
for small \(\epsilon\). Due to the irregular singularity at \(z=\infty\), there is a formal solution \eqref{Dubrovin Connection} near~${z=\infty}$ of the form
\begin{gather}\label{formal solution of Dubrovin connection at infinity}
Y_{\text{formal}}(z,u) = \left(I+\sum_{n=1}^{\infty}\frac{F_n(u)}{z^n} \right){\rm e}^{zU}.
\end{gather}
In addition, there exist analytic solutions of \eqref{Dubrovin Connection} \(Y_{\text{right}}\), \(Y_{\text{left}}\) in the sectors \(\Pi_{\text{right}},\) \(\Pi_{\text{left}}\) with the following asymptotic expansion
\begin{gather*}
 Y_{\text{formal}}{\rm e}^{-zU} \sim Y_{\text{right}}{\rm e}^{-zU} , \qquad \lvert z \rvert \to \infty, \qquad z \in \Pi_{\text{right}},\\
 Y_{\text{formal}}{\rm e}^{-zU} \sim Y_{\text{left}}{\rm e}^{-zU} , \qquad \lvert z \rvert \to \infty, \qquad z \in \Pi_{\text{left}}.
\end{gather*}
Note that in the sector
$
\Pi_{+} \colon \phi-\epsilon < \arg(z) < \phi+\epsilon$,
the solutions \(Y_{\text{right}}\), \(Y_{\text{left}}\) are defined. Hence, these solutions coincide up to a constant matrix, i.e.,
\begin{gather}\label{Stokes matrix definition}
 Y_{\text{left}}(z) = Y_{\text{right}}(z)S, \qquad z \in \Pi_{+}.
\end{gather}
Similarly, in the opposite narrow sector $\Pi_{-}$
\begin{gather}\label{Stokes matrix definition 2}
 Y_{\text{left}}(z) = Y_{\text{right}}(z)S_{-}, \qquad z \in \Pi_{-}.
\end{gather}
The constant matrices defined in \eqref{Stokes matrix definition} and \eqref{Stokes matrix definition 2} are called the Stokes matrices. From the standard theory of ODE in complex domain, the matrices $S$ and $S_{-}$ satisfy the following relation~${
S_{-}=S^{\mathsf T}}$.
The solution at \(z=0\) and the solution near \(z=\infty\) are connected by a constant matrix \(C\) called the connection matrix, i.e.,
$Y_{0}(z) = Y_{\text{right}}(z)C$, $ z \in \Pi_{+}$.

The set of parameters
\begin{gather}\label{Monodromy data of Dubrovin connection}
 (\mu, R, S, C),
\end{gather}
represents the monodromy data associated with the Dubrovin connection. The flatness property of the Dubrovin connection implies that the monodromy data \eqref{Monodromy data of Dubrovin connection} remains constant with respect to the directions $\frac{\partial}{\partial t^{\alpha}}$, $\alpha=1, \dots, n$. In other words, the system \eqref{Dubrovin connection flat coordinate system in saito coordinates} is isomonodromic.

To provide a more detailed description of the monodromy associated with the second structure connection, it is advantageous to express the extended Gauss--Manin connection \eqref{Extended Gauss--Manin connection} in canonical coordinates. Consider a solution
$
 x\bigl(t^1-\lambda, t^2, \dots, t^n\bigr)$,
of the flat coordinate system of the extended Gauss--Manin connection \eqref{Extended Gauss--Manin connection}. Then, the gradient
${
 \xi^{\alpha} = \eta^{\alpha\beta}\partial_{\beta}x\bigl(t^1-\lambda, t^2, \dots, t^n\bigr)}
$
solves the following system:
\begin{gather}
 (\mathcal{U}-\lambda)\partial_{\beta}\xi + C_{\beta}\left(\frac{1}{2}+\mu\right)\xi = 0,\qquad
 (\mathcal{U}-\lambda)\partial_{\lambda}\xi - \left(\frac{1}{2}+\mu\right)\xi = 0,\label{Extended Gauss manin in matrix form flat Saito}
\end{gather}
where $\mathcal{U}$ and $C_{\beta}$ are defined in \eqref{gradient flat saito coordinates def}.
In the semisimple case, a gauge transformation can be applied
\begin{gather}\label{gauge transform canonical coordinate}
 \phi = \Psi\xi,
\end{gather}
which is equivalent to
\begin{gather}\label{gauge transform canonical coordinate part 2}
 \phi_i = \sum_{\alpha,\beta}\psi_{i\alpha}\eta^{\alpha\beta}\partial_{\beta}x\bigl(t^1-\lambda,t^2,\dots,t^n\bigr).
\end{gather}

Utilising the gauge \eqref{gauge transform canonical coordinate}, we express the system \eqref{Extended Gauss manin in matrix form flat Saito} in the form
\begin{gather}
 (U-\lambda)\partial_{\lambda}\phi - \left(\frac{1}{2}+V\right)\phi = 0,\nonumber\\
 (U-\lambda)\partial_{i}\phi + E_{i}\left(\frac{1}{2}+V\right)\phi = 0, \qquad i=1,\dots,n,\label{Extended Gauss manin in matrix form canonical coordiantes matrix form 0}
\end{gather}
or alternatively
\begin{gather}
 \partial_{i}\phi = -\frac{B_i}{\lambda-u_i}\phi + V_i\phi,\qquad i=1,\dots,n,\qquad
 \partial_{\lambda}\phi = \sum_{i=1}^n \frac{B_i}{\lambda-u_i}\phi,\label{Extended Gauss manin in matrix form canonical coordiantes}
\end{gather}
where
\begin{gather}\label{definition of Bi}
 B_i = -E_i\left(\frac{1}{2}+V\right)
\end{gather}
and $V$, $V_i$, $E_i$ are defined in \eqref{U,V, Vi in canonical coordinates}.

The matrix \(B_i\) in \eqref{definition of Bi} has one eigenvalue \(-\frac{1}{2}\) and \(n-1\) zero eigenvalues. Consequently, there exists a unique basis of solutions \smash{$\phi^{(1)}, \phi^{(2)}, \dots, \phi^{(n)}$} for \eqref{Extended Gauss manin in matrix form canonical coordiantes} such that
\begin{gather}\label{asymptotic behaviour of phi 1}
 \phi_{a}^{(j)}(\lambda) = \frac{\delta_{aj}}{\sqrt{u_j-\lambda}} + O\bigl(\sqrt{u_j-\lambda}\bigr), \qquad \lambda\to u_j.
\end{gather}
Moreover, following the standard theory of Fuchsian systems,
\[
 \bigl(\phi_{a}^{(j)}(\lambda)\bigr) = \left(I + O\left(\frac{1}{\lambda}\right)\right)\lambda^{-\mu-\frac{1}{2}}\lambda^R, \qquad \lambda\to \infty.
\]

Under the semi-simplicity condition, choosing a reference point $\bigl(u^0,\lambda^0\bigr)\in M\times \mathbb{C}$, the monodromy~\eqref{monodromy of GM on the tolal bundle} splits on the short exact sequence
\begin{gather}\label{split of the monodromy group}
\pi_1\bigl(F^0,\lambda^0\bigr)\mapsto \pi_1\bigl(M\times \mathbb{C}\setminus\Sigma_{\lambda},\bigl(u^0,\lambda^0\bigr)\bigr) \mapsto \pi_1\bigl(M,u^0\bigr),
\end{gather}
where
$
p\colon M\times \mathbb{C}\setminus\Sigma_{\lambda}\mapsto M$, $ (u,\lambda)\mapsto u$,
and
$
F^0:=p^{-1}(u)=\mathbb{C}\setminus{\bigl\{u_1^0,u_2^0,\dots,u_n^0\bigr\}}$.

 We refer \cite[Section 1.2]{Milanov} or \cite[Proposition 5.6.4]{Shimada} for this fact.

The image of $\pi_1\bigl(F^0,\lambda^0\bigr)$ under the monodromy representation, which we denote by $W_R$, is obtained as follows. Consider small loops $\gamma_1,\gamma_2,\dots,\gamma_n$ going around $u_1,u_2,\dots,u_n$ in the counterclockwise direction. Then, the generators $R_1,R_2,\dots,R_n$ of the monodromy $W_R$ correspond to the image of the loops $\gamma_1,\gamma_2,\dots,\gamma_n$ under the monodromy representation. In particular, due to the monodromy of $\sqrt{u_i-\lambda}$ in \eqref{asymptotic behaviour of phi 1}, the $R_i$ action on the solutions of \eqref{Extended Gauss manin in matrix form canonical
coordiantes} is given by~${
R_j\phi^{(j)}=-\phi^{(j)}}$.
The general action is given by
\begin{gather}\label{reflection matrix generators}
R_j\phi^{(i)}=\phi^{(i)}-\bigl(S+S^{\mathsf T}\bigr)_{ij}\phi^{(j)},
\end{gather}
where $S$ is the Stokes matrices associated to the Dubrovin connection \eqref{Dubrovin Connection}. See details in \cite[Lemmas 5.3 and 5.4]{B.Dubrovin3}. The compatibility of the system \eqref{Extended Gauss manin in matrix form canonical coordiantes matrix form 0} implies that the system \eqref{Extended Gauss manin in matrix form canonical coordiantes matrix form 0} is isomonodromic.

Choose an angle $\varphi$ such that
$
\arg(u_i-u_j)\neq \frac{\pi}{2}+\varphi \mod{2\pi}$, $ i\neq j$.

Next, define the branch cuts
\begin{gather*}
L_j = \bigl\{ \lambda=u_j+{\rm i}\rho {\rm e}^{{\rm i}\varphi}, \, \rho \geq 0 \bigr\}, \qquad j=1,2,\dots,n,\\
L_j^{\prime} = \bigl\{ \lambda=u_j-{\rm i}\rho {\rm e}^{{\rm i}\varphi},\, \rho \geq 0 \bigr\}, \qquad j=1,2,\dots,n,
\end{gather*}
and the positive and negative side of each $L_j$
\begin{gather*}
L_j^{+} = \left\{ \lambda \mid \arg(u_j-\lambda)=\frac{-\pi}{2}-\varphi \right\}, \qquad L_j^{-} = \left\{ \lambda \mid \arg(u_j-\lambda)=\frac{-\pi}{2}-\varphi +2\pi\right\}.
\end{gather*}

Additionally, consider an infinite contour $C_j$ coming from infinity along the positive side of~$L_j$, encircling $u_j$, and returning to infinity along the negative side of the branchcut $L_j$. Also, consider the contour $C_j^{\prime}$ that encircles the branch cut $L_j^{\prime}$. Then, we state here a lemma which can be derived from the proof of \cite[Lemma 5.4]{B.Dubrovin3}.

 \begin{Lemma}[{\cite{B.Dubrovin3}}]\label{Main lemma Mirror symmetry 1}
 The Fourier--Laplace transform of the solution $\phi^{(j)}$ of \eqref{Extended Gauss manin in matrix form canonical coordiantes}
\begin{gather}
Y^{\rm right}_{aj}(z)=-\frac{1}{2\sqrt{\pi}\sqrt{z}}\int_{C_j}\phi_a^{(j)}(\lambda){\rm e}^{z\lambda}{\rm d}\lambda, \qquad a,j=1,2,\dots,n,\nonumber\\
Y^{\rm left}_{aj}(z)=-\frac{1}{2\sqrt{\pi}\sqrt{z}}\int_{C_j^{\prime}}\phi_a^{(j)}(\lambda){\rm e}^{z\lambda}{\rm d}\lambda, \qquad a,j=1,2,\dots,n,\label{Fourier Laplace transform}
\end{gather}
converges for $z \in \Pi_{\rm right}/\Pi_{\rm left}$, respectively, and give $n$ independent solutions of Dubrovin connection \eqref{Dubrovin connection in canonical coordinates 1}.
\end{Lemma}

\begin{proof}
 Indeed,
\begin{align}
\left(U+\frac{V}{z}\right)_{ia}Y_{aj}={}&-\frac{1}{2\sqrt{\pi}(z)^{\frac{3}{2}}}\int_{C_j}V_{ia}\phi_a^{j}(\lambda){\rm e}^{z\lambda}{\rm d}\lambda-\frac{1}{2\sqrt{\pi}\sqrt{z}}\int_{C_j}u_i\delta_{ia}\phi_a^{j}(\lambda){\rm e}^{z\lambda}{\rm d}\lambda\nonumber\\
={}&-\frac{1}{2\sqrt{\pi}(z)^{\frac{3}{2}}}\int_{C_j}\left(\frac{1}{2}+V_{ia}\right)\phi_a^{j}(\lambda){\rm e}^{z\lambda}{\rm d}\lambda\nonumber\\
&-\frac{1}{2\sqrt{\pi}\sqrt{z}}\int_{C_j}\left(u_i-\lambda\right)\delta_{ia}\phi_a^{j}(\lambda){\rm e}^{z\lambda}{\rm d}\lambda\nonumber\\
&+\frac{1}{4\sqrt{\pi}(z)^{\frac{3}{2}}}\int_{C_j}\phi_a^{j}(\lambda){\rm e}^{z\lambda}{\rm d}\lambda-\frac{1}{2\sqrt{\pi}\sqrt{z}}\int_{C_j}\lambda\phi_a^{j}(\lambda){\rm e}^{z\lambda}{\rm d}\lambda\nonumber\\
={}&-\frac{1}{2\sqrt{\pi}(z)^{\frac{3}{2}}}\int_{C_j}\left(\frac{1}{2}+V_{ia}\right)\phi_a^{j}(\lambda){\rm e}^{z\lambda}{\rm d}\lambda\nonumber\\
&-\frac{1}{2\sqrt{\pi}\sqrt{z}}\int_{C_j}\left(u_i-\lambda\right)\delta_{ia}\phi_a^{j}(\lambda){\rm e}^{z\lambda}{\rm d}\lambda+\partial_zY_{aj}.\label{mirror symmetry canonical coordinates technical lemma}
\end{align}
Substituting \eqref{Extended Gauss manin in matrix form canonical coordiantes matrix form 0} in \eqref{mirror symmetry canonical coordinates technical lemma} and using integration by parts, we obtain
\[
\left(U+\frac{V}{z}\right)_{ia}Y_{aj}=\partial_zY_{aj}.
\]
Since $\lambda=\infty$ is a regular singularity, the solution $\phi_a^{j}(\lambda)$ does not grow faster than $\lambda^k$, for some~$k$. Hence, the \eqref{Fourier Laplace transform} converges absolutely for $z \in \Pi_{\text{right/left}}$.
\end{proof}

\subsection{LG superpotential as family of isomonodromic covering maps}

In this subsection, we aim to construct a family of covering maps by ``inverting'' a solution of the extended Gauss--Manin connection. To achieve this, we require several auxiliary lemmas.

\begin{Lemma}[\cite{B.Dubrovin3}]
Consider the matrix \( G \) whose coefficients are given by
\smash{$
G^{ij} = \bigl(S + S^{\mathsf{T}}\bigr)_{ij}$},
and its inverse \( G^{-1} \) with coefficients \( G_{ij} \). Moreover, consider the solutions \smash{$ \phi^{(1)}, \phi^{(2)}, \dots, \phi^{(n)}$} of equation \eqref{Extended Gauss manin in matrix form canonical coordiantes}. Then, the solution
\begin{gather}\label{phi solution normalized by Gij}
\phi(\lambda,u) = (\phi_a(\lambda,u)) = \sum_{i,j=1}^{n} G_{ij} \phi^{(j)}
\end{gather}
has the following asymptotic behavior:
\begin{gather}\label{asymptotic expansion of phia lemma 0}
 \phi_{a}(\lambda) = \frac{\delta_{aj}}{\sqrt{u_j - \lambda}} + O\bigl(\sqrt{u_j - \lambda}\bigr), \qquad \lambda \to u_j.
\end{gather}

\begin{proof}
The asymptotic expansion \eqref{asymptotic expansion of phia lemma 0} follows from \eqref{asymptotic behaviour of phi 1}. See \cite[Lemma 5.7]{B.Dubrovin3} for details.\looseness=1
\end{proof}
\end{Lemma}

\begin{Lemma}[\cite{B.Dubrovin3}]\label{lemma primitive form}
Let \( (x_1, x_2, \dots, x_n) \) be a solution of the extended Gauss--Manin connection~\eqref{Gauss--Manin connection} with charge \( d \neq 1 \), and let \( (a_1, a_2, \dots, a_n) \in (\mathbb{C}^{*})^n \). Then, the solution of the form
\begin{gather}\label{primitive of volume form}
p\bigl(\lambda, t^1, t^2, \dots, t^n\bigr) := \sum_{i=1}^n a_i x_i\bigl(t^1 - \lambda, t^2, \dots, t^n\bigr)
\end{gather}
has the following behavior:
\begin{gather}
p\bigl(\lambda, t^1, t^2, \dots, t^n\bigr) = p_j + \sqrt{2} \psi_{j1} \sqrt{u_j - \lambda} + O(u_j - \lambda), \nonumber\\
 \lambda \to u_j, \qquad
 \forall j \in \{1, \dots, n\},\label{asymptotic behaviour at the branch locus of p}
\end{gather}
where \( p_j = p(u_j, t) \).
\end{Lemma}

\begin{proof}
According to \cite[equation (3.20)]{B.Dubrovin3}, the transformation law between canonical vector field coordinates and Saito flat vector fields is given by
\begin{gather}\label{change of vector field canonical to flat}
\frac{\partial}{\partial u_i} = \sum_{\alpha, \beta} \psi_{i\alpha} \eta^{\alpha\beta} \psi_{i1} \frac{\partial}{\partial t^{\beta}}.
\end{gather}
Substituting \eqref{change of vector field canonical to flat} into \eqref{gauge transform canonical coordinate part 2}, we obtain
\begin{gather}\label{new gauge transform canonical coordinates}
\phi_i = \frac{1}{\psi_{i1}} \partial_i x,
\end{gather}
where \( x \) in \eqref{new gauge transform canonical coordinates} is a generic solution of \eqref{Extended Gauss manin in matrix form flat Saito}.

The solutions of the extended Gauss--Manin connection \eqref{Extended Gauss manin in matrix form flat Saito} are quasi-homogeneous (see \cite[Appendix H, equation~(H.19)]{B.Dubrovin2}), i.e.,
\begin{gather}\label{quasi homogeneous xa}
\left( \lambda \frac{\rm d}{{\rm d}\lambda} + \text{Lie}_E \right) (x(\lambda, t)) = \frac{1 - d}{2} x(\lambda, t).
\end{gather}
Moreover, recall the representation of the vector fields \( \partial_\lambda \) and \( \text{Lie}_E \) in canonical coordinates:
\begin{gather}\label{action of unity and Euler vector field}
\frac{\rm d}{{\rm d}\lambda} = -\frac{\partial}{\partial t^1} = \sum_{i=1}^n \frac{\partial}{\partial u_i}, \qquad E = \sum_{i=1}^n u_i \frac{\partial}{\partial u_i}.
\end{gather}

Therefore, substituting \eqref{new gauge transform canonical coordinates} and \eqref{action of unity and Euler vector field} into \eqref{quasi homogeneous xa}, we obtain
\begin{gather}
x(\lambda, t) = \frac{2}{1 - d} \sum_{i=1}^n (u_i - \lambda) \partial_i x(\lambda, t) = \frac{2}{1 - d} \sum_{i=1}^n (u_i - \lambda) \psi_{i1} \phi_i(\lambda, t).\label{final form quasi homogeneous}
\end{gather}
Here we choose $\phi$ to be given by \eqref{phi solution normalized by Gij}. Using the asymptotic expansion \eqref{asymptotic expansion of phia lemma 0} in \eqref{final form quasi homogeneous}, we can choose constants \( a_1, \dots, a_n \) such that \eqref{asymptotic behaviour at the branch locus of p} holds. This completes the proof.
\end{proof}

The solution $p(\lambda,t)$ \eqref{primitive of volume form} is analytic functions on $\mathbb{C}\setminus{\bigcup_{j=1}^n L_j }$. In particular, solutions of \eqref{Extended Gauss manin in matrix form flat Saito} are locally invertible outside the discriminant locus. Consider the set of the analytic continuations of $p(\lambda,t)$
$
X:=\bigl\{ p_{\gamma}(\lambda,u) \mid \gamma \in \pi_1\bigl(\mathbb{C}\setminus\bigl\{u_1^0,u_2^0,\dots,u_n^0\bigr\}\bigr)\bigr\}/{\sim},
$
here $\bigl(u_1^0,u_2^0,\dots,u_n^0\bigr)$ is a fixed reference point and the equivalence relation $\sim$ is given by
\begin{gather*}
 p_{\gamma_1}(\lambda,u)= p_{\gamma_2}(\lambda,u) \qquad \text{iff} \qquad \exists g \in \pi_1\bigl(\mathbb{C}\setminus\bigl\{u_1^0,u_2^0,\dots,u_n^0\bigr\}\bigr) \colon \ g\gamma_1=\gamma_2,
\end{gather*}
and the following monodromy representation{\samepage
\begin{gather}\label{monodromy representation superpotential}
\rho\colon\ \pi_1\bigl(\mathbb{C}\setminus\bigl\{u_1^0,u_2^0,\dots,u_n^0\bigr\}\bigr)\mapsto \operatorname{Aut}(X)
\end{gather}
given by
$\rho(\gamma_1)( p_{\gamma_2}(\lambda,u) )=p_{\gamma_1\bullet\gamma_2}(\lambda,u)$.}

For fixed a reference point $\bigl(u_1^0,u_2^0,\dots,u_n^0\bigr)$, consider the universal covering of $\mathbb{C}\setminus\bigl\{u_1^0,u_2^0,\dots,\allowbreak u_n^0\bigr\}$ and denoted it by
\[
\widetilde{\mathbb{C}\setminus\bigl\{u_1^0,u_2^0,\dots,u_n^0\bigr\}}.
\] Then, note that the kernel of the representation~\eqref{monodromy representation superpotential} acts properly discontinuously in
\[
\widetilde{\mathbb{C}\setminus\bigl\{u_1^0,u_2^0,\dots,u_n^0\bigr\}},
\]
 because $\ker\rho$ is a normal subgroup of $\pi_1\bigl(\mathbb{C}\setminus\bigl\{u_1^0,u_2^0,\dots,u_n^0\bigr\}\bigr)$. Then, we construct a Riemann surface given by
\begin{gather*}
S_{(u_1^0,u_2^0,\dots,u_n^0)}:=\widetilde{\mathbb{C}\setminus\bigl\{u_1^0,u_2^0,\dots,u_n^0\bigr\}}/\operatorname{Ker}\rho.
\end{gather*}
The local inverses of the analytic continuation $p_{\gamma}(\lambda,u)$,
$
\lambda\bigl(p_{\gamma}\bigl(\lambda,u^0\bigr),u^0\bigr)=\lambda
$
give rise to a co\-vering map
\smash{$
\lambda \colon S_{(u_1^0,u_2^0,\dots,u_n^0)}\mapsto \mathbb{C}$},
which group of deck transformations is image of the monodromy representation \eqref{monodromy representation superpotential}.

 We can repeat this construction for any $(u_1,u_2,\dots,u_n)$ close enough to $\bigl(u_1^0,u_2^0,\dots,u_n^0\bigr)$. More specifically, it is sufficient that $(u_1,u_2,\dots,u_n)$ do not intersect the branch cut $L_j^0$. As a result, we build a family of Riemann surfaces $S_{(u_1,u_2,\dots,u_n)}$. Then, the family of covering maps~${
\lambda(p_{\gamma}(\lambda,u),u)=\lambda}
$
gives rise to an isomonodromic family of coverings
\begin{gather}\label{family of coverings 0}
\lambda \colon\ S_{(u_1,u_2,\dots,u_n)}\mapsto \mathbb{C}.
\end{gather}

Here, we state \cite[Proposition I.1]{B.Dubrovin2} as follows.

\begin{Lemma}[{\cite{B.Dubrovin2}}]\label{main constructive LG lemma}
Consider the family of coverings
\begin{gather}\label{family of coverings}
\lambda \colon\ S_{(u_1,u_2,\dots,u_n)}\mapsto \mathbb{C}
\end{gather}
defined in \eqref{family of coverings 0}. Then, the family \eqref{family of coverings} exhibits the following behavior:
\begin{gather}\label{local behaviour of LG}
\lambda(p)=u_j+ \frac{1}{2\psi^2_{j1}}( p-p_j)^2+O\bigl(( p-p_j)^3\bigr), \qquad p\mapsto p_j.
\end{gather}
\end{Lemma}

\begin{proof}
The local behavior \eqref{local behaviour of LG} is determined by the local inverse of equation \eqref{asymptotic behaviour at the branch locus of p}. Lemma proved.
\end{proof}

At this stage, we revisit the construction of the Landau--Ginzburg potential in the sense of Dubrovin as presented in \cite[Theorem~5.3]{B.Dubrovin3}.

\begin{Theorem}[{\cite{B.Dubrovin3}}]\label{Dubrovin superpotential theorem}
Consider the family of coverings defined in~\eqref{family of coverings 0}
$
\lambda \colon S_{u_1,u_2,\dots,u_n}\mapsto \mathbb{C}$, $ p\mapsto \lambda(p,u)$.
Then, the function $\lambda(p,u)$, together with the Abelian differential $\phi={\rm d}p$, forms a~Landau--Ginzburg superpotential.
\end{Theorem}

\begin{proof}
Due to Lemma \ref{main constructive LG lemma}, the critical values of $\lambda(p,u)$ are the canonical coordinates $(u_1,\dots,u_n)$, i.e., $(u_1,u_2,\dots,u_n)$ is defined by the following system
$\lambda(p_i)=u_i$, $\frac{{\rm d} \lambda}{{\rm d}p}( p_i)=0$.

Consider the cycles $Z_j=p(C_j)$, which are the lifts of $C_j$ in the correspondent family of Riemann surfaces $S_{u_1,u_2,\dots,u_n}$.

Integrating by parts and changing variables in \eqref{Mirror symmetry},
\begin{gather*}
\tilde t_j=\frac{1}{z^{\frac{3}{2}}}\int_{Z_j} {\rm e}^{z\lambda(p,t)}{\rm d}p
=\frac{1}{z^{\frac{3}{2}}}\int_{C_j} {\rm e}^{z\lambda}\frac{{\rm d}p}{{\rm d}\lambda}{\rm d}\lambda=-\frac{1}{\sqrt{z}}\int_{C_j} p(\lambda,t){\rm e}^{z\lambda}{\rm d}\lambda.
\end{gather*}
Hence, the gradient of \eqref{Mirror symmetry}
\begin{gather*}
Y_{aj}(z)=\frac{1}{\psi_{a1}}\frac{\partial\tilde t_j}{\partial u_a}=\frac{1}{\sqrt{z}}\int_{C_j} \phi_a(\lambda,t){\rm e}^{z\lambda}{\rm d}\lambda
\end{gather*}
solves the Dubrovin connection flat coordinate system due to Lemma \ref{Main lemma Mirror symmetry 1}.

Due to Lemma \ref{main constructive LG lemma} and \eqref{Dubrovin Frobenius structure in canonical coordinates 1}, we have that
\smash{$
\eta_{ii}=\frac{1}{\lambda^{\prime\prime}(p_i)}$}
holds true. The Taylor expansion of~$\lambda$, $\frac{\partial \lambda}{\partial p}$, $\frac{\partial \lambda}{\partial u_j}$ around $p_i$ is given by
\begin{gather}
\lambda(p)=u_i+\lambda^{\prime\prime}(p_i)\frac{(p-p_i)^2}{2}+O\bigl( (p-p_i)^3\bigr),\nonumber\\
\frac{{\rm d} \lambda}{{\rm d}p}(p)=\lambda^{\prime\prime}(p_i)(p-p_i)+O\bigl( (p-p_i)^2\bigr),\qquad
\frac{{\rm d} \lambda}{\partial u_j}(p)=\delta_{ij}+O( (p-p_i)).\label{lambda expansion}
\end{gather}
Then, substituting \eqref{lambda expansion} in the right-hand side of \eqref{residue expression for eta, intersection form and structure constants} in canonical coordinates, we obtain
\begin{gather*}
 \sum \res_{{\rm d}\lambda=0} \frac{ \partial_i\lambda \partial_j\lambda }{{\rm d}_p\lambda }{\rm d}p =\frac{\delta_{kl}}{\lambda^{\prime\prime}(p_l)}=\eta_{kl},\qquad
\sum \res_{{\rm d}\lambda=0} \frac{ \partial_i\log\lambda \partial_j\log\lambda }{{\rm d}_p\log\lambda} {\rm d}p =\frac{\delta_{kl}}{u_l\lambda^{\prime\prime}(p_l)}=g_{kl},\\
\sum \res_{{\rm d}\lambda=0} \frac{ \partial_i\lambda \partial_j\lambda \partial_k\lambda }{{\rm d}_p\lambda}{\rm d}p=\frac{\delta_{kl}\delta_{km}}{\lambda^{\prime\prime}(p_l)}=c_{klm}.
\end{gather*}
Hence, the third statement of \ref{Landau--Ginzburg superpotential definition} is proved for canonical coordinates. In order to compute the coefficients of the tensors \eqref{residue expression for eta, intersection form and structure constants} in a generic coordinate system, it is enough to use the change of coordinates in $\lambda$. Indeed,
\begin{gather*}
\sum \res_{{\rm d}\lambda=0} \frac{\frac{\partial\lambda}{\partial v_i}\frac{\partial\lambda}{\partial v_j} }{\frac{{\rm d}\lambda}{{\rm d}p} }{\rm d}p=\frac{\partial u_k}{\partial v_i}\frac{\partial u_l}{\partial v_j}\left(\sum \res_{{\rm d}\lambda=0} \frac{\frac{\partial\lambda}{\partial u_k}\frac{\partial\lambda}{\partial u_l} }{\frac{{\rm d}\lambda}{{\rm d}p} }{\rm d}p \right)=\eta_{ij},\\
\sum \res_{{\rm d}\lambda=0} \frac{\frac{\partial\lambda}{\partial v_i}\frac{\partial\lambda}{\partial v_j} }{\lambda\frac{{\rm d}\lambda}{{\rm d}p} }{\rm d}p=\frac{\partial u_k}{\partial v_i}\frac{\partial u_l}{\partial v_j}\left(\sum \res_{{\rm d}\lambda=0} \frac{\frac{\partial\lambda}{\partial u_k}\frac{\partial\lambda}{\partial u_l} }{\lambda \frac{{\rm d}\lambda}{{\rm d}p} }{\rm d}p\right)=g_{ij} ,\\
\sum \res_{{\rm d}\lambda=0} \frac{\frac{\partial\lambda}{\partial v_i}\frac{\partial\lambda}{\partial v_j} \frac{\partial\lambda}{\partial v_k}}{\frac{{\rm d}\lambda}{{\rm d}p} }{\rm d}p=\frac{\partial u_k}{\partial v_i}\frac{\partial u_l}{\partial v_j}\frac{\partial u_l}{\partial v_m}\left(\sum \res_{{\rm d}\lambda=0} \frac{\frac{\partial\lambda}{\partial u_k}\frac{\partial\lambda}{\partial u_l} \frac{\partial\lambda}{\partial u_m}}{\frac{{\rm d}\lambda}{{\rm d}p} }{\rm d}p\right)=c_{ijk}.\tag*{\qed}
\end{gather*}
\renewcommand{\qed}{}
\end{proof}

\begin{Remark}
From the data of covering over $\lambda\colon S_{u_1,\dots, u_n}\mapsto\mathbb{CP}^1$ and Abelian differential~${\rm d}p$, one can construct a Dubrovin--Frobenius manifold according Theorem~\ref{Dubrovin superpotential theorem}. The Abelian differential ${\rm d}p$ depend on a choice of a point in $(a_1,a_2,\dots,a_n)\in(\mathbb{C}^{*})^n$. If one choose a~different point~${(b_1,b_2,\dots,b_n)\in(\mathbb{C}^{*})^n}$, the Dubrovin--Frobenius structures of $(a_1,a_2,\dots,a_n)$ and $(b_1,b_2,\dots,b_n)$ are related by a Legendre transform. See \cite[Appendix B, Chapter 5]{B.Dubrovin2} for details.
\end{Remark}

\section[Small quantum cohomology of $\mathbb{CP}^2$]{Small quantum cohomology of $\boldsymbol{\mathbb{CP}^2}$}\label{small Quantum Cohomolology}

\subsection[Monodromy of $\mathbb{CP}^2$]{Monodromy of $\boldsymbol{QH^{*}\bigl(\mathbb{CP}^2\bigr)}$}

In this subsection, we will review the main steps for computing the monodromy of $QH^{*}\bigl(\mathbb{CP}^2\bigr)$. For the convenience of the reader, we summarize the main steps of \cite[Examples~4.4 and~5.5]{B.Dubrovin3}. Consider the Dubrovin connection \eqref{Dubrovin connection flat coordinate system in saito coordinates} associated to the WDVV solution \eqref{generating function of Gromov Witten CP2} at the locus of small quantum cohomoology, i.e., we consider the sublocus $t^1=t^3=0$ in $QH^{*}\bigl(\mathbb{CP}^2\bigr)$. Explicitly such system can be reduced to
\begin{gather}
\partial_2^3\phi=z^3Q\phi,\qquad
(z\partial_z)^3\phi=27z^3Q\phi,\label{small Dubrovin connection 1}
\end{gather}
where the relation of $\phi$ and $\xi=\eta^{-1}\omega=(\xi_1,\xi_2,\xi_3)$ in \eqref{Dubrovin connection flat coordinate system in saito coordinates} is represented as follows:
\[
(\xi_1,\xi_2,\xi_3)=\left(z\phi,\frac{1}{3}z\partial_z\phi,\frac{1}{9}\partial_z(z\partial_z\phi)\right).
\]
Using the quasi homogeneous condition
$
\phi\bigl(z,t^2\bigr)=\phi\bigl(Q^{\frac{1}{3}}z\bigr)$,
the system \eqref{small Dubrovin connection 1} reduces to
\begin{gather}
(z\partial_z)^3\phi=27z^3Q\phi.\label{small Dubrovin connection 2}
\end{gather}
Under this setting, the multiplication by the Euler vector field becomes
\begin{gather}
\bigl(g^{\alpha\beta}\bigr)={\begin{pmatrix}
3Q& 0&0 \\
0& 0& 3\\
0& 3 & 0
\end{pmatrix}}, \qquad \text{or alternatively,} \qquad (g^{\alpha}_{\beta})= \begin{pmatrix}
0& 0&3Q \\
3& 0& 0\\
0& 3 & 0
\end{pmatrix}.\label{small multiplication by the Euler vector field}
\end{gather}
The canonical coordinates, which are the roots of \eqref{spectral curve} with respect \eqref{small multiplication by the Euler vector field}, are
\[
u_1=3Q^{\frac{1}{3}}, \qquad u_2=3Q^{\frac{1}{3}}\zeta_3^2, \qquad u_3=3Q^{\frac{1}{3}}\zeta_3, \qquad \text{where} \quad \zeta_3={\rm e}^{\frac{2\pi {\rm i}}{3}}.
\]
Because the transition function $\psi_{i\alpha}$ diagonalizes the multiplication by the Euler vector field, it can be concluded that
\[
\Psi=\frac{1}{\sqrt{3}}{\begin{pmatrix}
Q^{\frac{-1}{3}}& 1&Q^{\frac{1}{3}} \\
\bar\zeta_6Q^{\frac{-1}{3}}& -1& \zeta_6Q^{\frac{1}{3}}\\
\zeta_6Q^{\frac{-1}{3}}& -1 & \bar\zeta_6Q^{\frac{1}{3}}
\end{pmatrix}}, \qquad \text{where} \quad \zeta_6={\rm e}^{\frac{\pi {\rm i}}{3}}.
\]
The monodromy of the Dubrovin flat coordinate system at $z=0$ is determined by its matrix $\mu$ and its first Chern class $c_1\bigl(\mathbb{CP}^2\bigr)$, see \cite[Examples~1.3 and~2.2]{B.Dubrovin3},
\[
\mu={\begin{pmatrix}
-1& 0&0 \\
0& 0& 0\\
0& 0 & 1
\end{pmatrix}}, \qquad
R={\begin{pmatrix}
0& 0&0 \\
3& 0& 0\\
0& 3 & 0
\end{pmatrix}}.
\]
To compute the monodromy at $z=\infty$, we calculate the asymptotic behaviour of the solutions of \eqref{small Dubrovin connection 2}. Then, it is necessary to choose basis $\tilde t_1$, $\tilde t_2$, $\tilde t_3$ such that the matrix
\smash{$
Y_{ij}=\frac{\partial_i \tilde t_j}{\psi_{i1}}
$}
has the asymptotic behaviour of the form \eqref{formal solution of Dubrovin connection at infinity}, i.e.,
\begin{gather*}
Y_{ij}\sim \left(\delta_{ij}+O\left(\frac{1}{z}\right) \right){\rm e}^{zu_j}, \qquad i,j=1,2,3.
\end{gather*}
Then, we obtain three solutions $\phi_1$, $\phi_2$, $\phi_3$ of \eqref{small Dubrovin connection 2} such that
\begin{gather*}
\phi_j=\frac{\partial_1\tilde t_j}{z}=\frac{1}{z}\sum_{i=1}^3\partial_i\tilde t_j=\frac{1}{z}\sum_{i=1}^3 \psi_{i1}Y_{ij}.
\end{gather*}
Furthermore, the functions $\phi_1$, $\phi_2$, $\phi_3$ have the following asymptotic behaviour:
\begin{gather*}
\phi_1\sim \frac{1}{\sqrt{3}}\frac{{\rm e}^{3Q^{\frac{1}{3}}z}}{Q^{\frac{1}{3}}z}\left(1+O\left(\frac{1}{z}\right) \right),\qquad
\phi_2\sim \frac{\bar\zeta_6}{\sqrt{3}}\frac{{\rm e}^{3Q^{\frac{1}{3}}\bar\zeta_6^2z}}{Q^{\frac{1}{3}}z}\left(1+O\left(\frac{1}{z}\right) \right),\\
\phi_3\sim \frac{\zeta_6}{\sqrt{3}}\frac{{\rm e}^{3Q^{\frac{1}{3}}\zeta_6^2z}}{Q^{\frac{1}{3}}z}\left(1+O\left(\frac{1}{z}\right) \right).
\end{gather*}

At this stage, we can compute the correspondent Stokes matrices of $QH^{*}\bigl(\mathbb{CP}^2\bigr)$ by using \cite[Lemma~4.9]{B.Dubrovin3}. Here, we state as follows.

\begin{Lemma}[{\cite{B.Dubrovin3}}]
The Meijer function
\[
g(z,Q)=\frac{1}{(2\pi)^2{\rm i}}\int_{-c-{\rm i}\infty}^{-c+{\rm i}\infty} \Gamma^3(-s){\rm e}^{\pi {\rm i}s }Q^sz^{3s}{\rm d}s, \qquad -\frac{5\pi}{6}<\arg(z)<\frac{\pi}{6},
\]
defined for $z\neq 0$, where $c$ is any positive number, satisfies the equation \eqref{small Dubrovin connection 2}. The analytic continuation of this function has the asymptotic development
\[
g(z,Q)\sim \frac{1}{\sqrt{3}}\bar\zeta_6\frac{{\rm e}^{3Q^{\frac{1}{3}}\bar\zeta_6^2z}}{Q^{\frac{1}{3}}z}, \qquad |z|\mapsto \infty, \qquad \text{in the sector} \qquad \frac{-5\pi}{3}<\arg(z)<\pi.
\]
Moreover, it satisfies the following identity
\begin{gather}\label{Meijer function id}
g\bigl(z{\rm e}^{2\pi {\rm i}}\bigr)-3g\bigl(z{\rm e}^{\frac{4\pi {\rm i}}{3}}\bigr)+3g\bigl(z{\rm e}^{\frac{2\pi {\rm i}}{3}}\bigr)-g(z)=0.
\end{gather}

\end{Lemma}

To construct $\phi_{\rm left}=\bigl(\phi_{\rm left}^1,\phi_{\rm left}^2,\phi_{\rm left}^3\bigr)$ and $\phi_{\rm right}=\bigl(\phi_{\rm right}^1,\phi_{\rm right}^2,\phi_{\rm right}^3\bigr)$ such that
\begin{gather*}
\phi_{\rm right}^j\cong\phi^j, \qquad |z|\mapsto \infty, \qquad -\pi<\arg(z)<\frac{\pi}{3}, \qquad j=1,2,3,\\
\phi_{\rm left}^j\cong\phi^j, \qquad |z|\mapsto \infty, \qquad 0<\arg(z)<\frac{4\pi}{3}, \qquad j=1,2,3,
\end{gather*}
we take
\begin{gather*}
\phi_{\rm right}=\bigl(-g\bigl(z{\rm e}^{\frac{2\pi {\rm i}}{3}}\bigr),g(z),g\bigl(z{\rm e}^{\frac{-2\pi {\rm i}}{3}}\bigr) \bigr),\\
\phi_{\rm left}=\bigl(-g\bigl(z{\rm e}^{\frac{-4\pi {\rm i}}{3}}\bigr),g\bigl(z{\rm e}^{-2\pi {\rm i}}\bigr)-3g\bigl(z{\rm e}^{\frac{-4\pi {\rm i}}{3}}\bigr),g\bigl(z{\rm e}^{\frac{-2\pi {\rm i}}{3}}\bigr) \bigr).
\end{gather*}
Using the identity \eqref{Meijer function id} in the sector $0<\arg(z)<\frac{\pi}{3}$, we deduce that
$\phi_{\rm left}=\phi_{\rm right}S$,
where
\[
S={\begin{pmatrix}
1& 0&0 \\
3& 1& 0\\
-3& -3 & 1
\end{pmatrix}}
\]
is the Stokes matrices of $QH^{*}\bigl(\mathbb{CP}^2\bigr)$. Here, we used the fact that the Dubrovin flat coordinate system \eqref{Dubrovin connection flat coordinate system in saito coordinates} is isomonodromic. Then, the monodromy computed in a fixed point of $QH^{*}\bigl(\mathbb{CP}^2\bigr)$ is the same as the monodromy computed in a neighbourhood of this reference point.

The monodromy group of the Gauss--Manin connection \eqref{Extended Gauss--Manin connection} splits in short exact sequen\-ce~\eqref{split of the monodromy group}. The monodromy representation of $\pi_1\bigl(F^0,\lambda^0\bigr)$, denoted by $W_R$, is characterized by the reflections given in \eqref{reflection matrix generators}. Hence the generators of $W_R$ associated with $QH^{*}\bigl(\mathbb{CP}^2\bigr)$ are given by
\begin{gather}\label{generators of wr Dubrovin}
R_1={\begin{pmatrix}
-1& -3&3 \\
0& 1& 0\\
0& 0 & 1
\end{pmatrix}}, \qquad R_2={\begin{pmatrix}
1& 0&0 \\
-3& -1& 3\\
0& 0 & 1
\end{pmatrix}}, \qquad R_3={\begin{pmatrix}
1& 0&0 \\
0& 1& 0\\
3& 3 & -1
\end{pmatrix}}.
\end{gather}
The monodromy $\pi_1(D,t)$, where $D$ is defined in \eqref{WDVV domain}, is given by loops around $Q=0$. The image under the monodromy representation of loops around $Q=0$ correspond to loops \smash{$z\mapsto z{\rm e}^{\frac{2\pi {\rm i}}{3}}$}. This loop acts on the solutions of Dubrovin flat coordinate system as follows:
\begin{gather*}
\Phi_{\rm left}\bigl(z{\rm e}^{\frac{2\pi {\rm i}}{3}}\bigr)=\Phi_{\rm left}(z)T ,\qquad \text{where} \quad
T={\begin{pmatrix}
0& -1&0 \\
0& 0& 1\\
-1& -3 & 3
\end{pmatrix}}.
\end{gather*}
Defining
\begin{gather*}
T_0:=TR_1={\begin{pmatrix}
0& -1&0 \\
0& 0& 1\\
1& 0 & 0
\end{pmatrix}},
\end{gather*}
we have the following relations
$
T_0^3=-I$, $ R_2=T_0R_1T_0^{-1}$, $ R_3=T_0^2R_1T_0^{-2}$.
Consequently, the monodromy group $W$ of $QH^{*}\bigl(\mathbb{CP}^2\bigr)$ is generated by
\begin{gather}\label{monodromy of quantum cohomology cp2}
W=\big\langle R_1, T_0^4, T_0^3 \big\rangle .
\end{gather}

To better understand the groups $W_R$, $W$, we make a digression on modular forms and elliptic functions. For comprehensive details, we refer the reader to \cite{Diamond} and \cite{Zagier}. Consider the upper half plane
$\mathbb{H}=\{ \tau\in\mathbb{C}\mid \operatorname{Im}(\tau)>0 \}$.

The special linear group ${\rm SL}_2(\mathbb{Z})$
 acts on $\mathbb{H}$ as follows
 \begin{gather}\label{SL2Z action}
 {\begin{pmatrix}
a& b \\
c& d \\
\end{pmatrix}}\tau=\frac{a\tau+b}{c\tau+d}, \qquad {\begin{pmatrix}
a& b \\
c& d \\
\end{pmatrix}}\in {\rm SL}_2(\mathbb{Z})\qquad \tau\in \mathbb{H}.
 \end{gather}
 The standard generators of ${\rm SL}_2(\mathbb{Z})$ are
$
S= \bigl(\begin{smallmatrix}
0& 1 \\
-1& 0 \\
\end{smallmatrix}\bigr)$, $ T= \bigl(\begin{smallmatrix}
1& 1 \\
0& 1 \\
\end{smallmatrix}\bigr)$,
 which satisfy only the relations~${S^2=(ST)^3=I}$.

 The action \eqref{SL2Z action} acts properly discontinuously in $\mathbb{H}$ having
 \begin{gather*}
 F_0=\left\{ \tau \in \mathbb{H} \mid -\frac{1}{2}\leq \operatorname{Re}(\tau)\leq \frac{1}{2},\, |\tau|>1 \right\}
 \end{gather*}
 as the fundamental domain. The quotient map $\mathbb{H}\mapsto\mathbb{H}/{\rm SL}_2(\mathbb{Z})$ is a branched covering map with branches of order $2$ and $3$ in the ${\rm SL}_2(\mathbb{Z})$ orbit of ${\rm i}$ and \smash{${\rm e}^{\frac{2\pi {\rm i}}{3}}$}, respectively.

 The quotient space $\mathbb{H}/{\rm SL}_2(\mathbb{Z})$ represents the moduli space of isomorphic classes of complex tori. Indeed, consider $\omega_1,\omega_2 \in \mathbb{C}$ be linearly independent over $\mathbb{R}$ and the lattice ${\Lambda=\mathbb{Z}\omega_1+\mathbb{Z}\omega_2}$. The quotient space $\mathbb{C}/\Lambda$ is a torus, and two tori $\mathbb{C}/\Lambda$, $\mathbb{C}/\Lambda^{\prime}$ are biholomorphic iff the correspondents lattices $\Lambda$, $\Lambda^{\prime}$ are homothetics, i.e., there exist $\alpha\in \mathbb{C}^{*}$ such that $\alpha\Lambda=\Lambda^{\prime}$. Therefore, if we multiply the lattice $\Lambda=\mathbb{Z}\omega_1+\mathbb{Z}\omega_2$ by $\alpha=\omega_1^{-1}$, assuming that $\tau:=\frac{\omega_2}{\omega_1}$ has positive imaginary part, we have that $\mathbb{C}/\Lambda\cong \mathbb{C}/\mathbb{Z}+\tau\mathbb{Z}$. Moreover,
 \begin{gather*}
\mathbb{C}/\mathbb{Z}+\tau\mathbb{Z}\cong \mathbb{C}/\mathbb{Z}+\tau^{\prime}\mathbb{Z} \qquad \text{iff} \qquad {\begin{pmatrix}
\tau^{\prime} \\
1 \\
\end{pmatrix}}={\begin{pmatrix}
a& b \\
c&d \\
\end{pmatrix}}{\begin{pmatrix}
\tau \\
1 \\
\end{pmatrix}}, \qquad {\begin{pmatrix}
a& b \\
c&d \\
\end{pmatrix}} \in {\rm SL}_2(\mathbb{Z}).
\end{gather*}

A modular form of weight $k$ is a holomorphic function $f\colon\mathbb{H}\mapsto \mathbb{C}$ such that
\begin{gather*}
 f\left(\frac{a\tau+b}{c\tau+d} \right)=\left( c\tau+d\right)^{k}f(\tau), \qquad {\begin{pmatrix}
a& b \\
c&d \\
\end{pmatrix}} \in {\rm SL}_2(\mathbb{Z}).
 \end{gather*}
 Furthermore, we have some controlled growth condition at $i\infty$. More precisely, its Fourier series around $i\infty$ has the following form:
$
f(\tau)=\sum_{n=0}^{\infty} a_nq^{n}$, $ q={\rm e}^{2\pi {\rm i}\tau}$.

 The ring of meromorphic functions on a torus $ \mathbb{C}/\mathbb{Z}+\tau\mathbb{Z}$ provides examples of modular forms. A function defined on the $ \mathbb{C}/\mathbb{Z}+\tau\mathbb{Z}$ must be double periodic, i.e., $f(z+m+n\tau)=f(z)$. Due to maximum modulus principle and Liouville's theorem, a function defined on the torus must be meromorphic. An example of such function is the Weierstrass $\wp$ function given by
 \begin{gather}\label{Weierstrass p function}
\wp(v,\tau)=\frac{1}{v^2}+\sum_{m^2+n^2\neq 0} \left(\frac{1}{(v+m+n\tau)^2} -\frac{1}{(m+n\tau)^2} \right), \qquad m,n \in \mathbb{Z}.
\end{gather}
The series \eqref{Weierstrass p function} is absolutely convergent, holomorphic in $\mathbb{C}\setminus{\mathbb{Z}+\tau\mathbb{Z}}$ and its poles are in the lattice $\mathbb{Z}+\tau\mathbb{Z}$. Hence, the function \eqref{Weierstrass p function} has the following transformation laws:
 \begin{gather}
\wp(v+m+n\tau,\tau)=\wp(v,\tau),\qquad m,n \in\mathbb{Z},\nonumber\\
\wp\left(\frac{v}{c\tau+d},\frac{a\tau+b}{c\tau+d}\right)=(c\tau+d)^2\wp(v,\tau),\qquad {\begin{pmatrix}
a& b \\
c&d \\
\end{pmatrix}} \in {\rm SL}_2(\mathbb{Z}).\label{Weierstrass p function transformation laws}
\end{gather}

In addition, the function \eqref{Weierstrass p function} satisfies the non-linear differential equation
\begin{gather}\label{differential equation of Weierstrass p function 0}
( \wp^{\prime}(v,\tau))^2=4\wp^3(v,\tau)-g_2(\tau)\wp(v,\tau)-g_3(\tau),
\end{gather}
where
\begin{gather}\label{Einsenstein series 1}
g_2(\tau)=60\sum_{m^2+n^2\neq 0}\frac{1}{(m+n\tau)^4}, \qquad g_3(\tau)=140\sum_{m^2+n^2\neq 0}\frac{1}{(m+n\tau)^6},
\end{gather}
which are absolutely convergent series. The functions \eqref{Einsenstein series 1} are modular forms of weight~$4$ and~$6$, respectively, i.e.,
\begin{gather}
g_2\left(\frac{a\tau+b}{c\tau+d}\right)=(c\tau+d)^4g_2(\tau),\qquad {\begin{pmatrix}
a& b \\
c&d \\
\end{pmatrix}} \in {\rm SL}_2(\mathbb{Z}),\nonumber\\
g_3\left(\frac{a\tau+b}{c\tau+d}\right)=(c\tau+d)^6g_3(\tau),\qquad {\begin{pmatrix}
a& b \\
c&d \\
\end{pmatrix}} \in {\rm SL}_2(\mathbb{Z}).\label{Transformation law Einsenstein series}
\end{gather}
The functions \eqref{Einsenstein series 1} are called Eisenstein series of weight 4 and 6, respectively. Another normalizations of \eqref{Einsenstein series 1} are given by
\begin{gather}
G_k(\tau)=\frac{1}{2}\sum_{m^2+n^2\neq 0}\frac{1}{(m+n\tau)^k}, \qquad E_k(\tau)=\frac{1}{\zeta(k)}G_k(\tau),\label{Einsenstein series 2}
\end{gather}
where $\zeta(k)=\sum_{n=0}^{\infty} \frac{1}{n^k}$.

The differential equation \eqref{differential equation of Weierstrass p function 0} provides a biholomorphism between the torus $\mathbb{C}\setminus{\mathbb{Z}+\tau\mathbb{Z}}$ and the compactification of the algebraic curve
\begin{gather}\label{algebraic of Weierstrass p function}
y^2=4x^3-g_2x-g_3.
\end{gather}
Consider the image of the half periods under \eqref{Weierstrass p function}
\begin{gather}\label{critical values of wp}
e_1(\tau)=\wp\left(\frac{1}{2},\tau\right), \qquad e_2(\tau)=\wp\left(\frac{1+\tau}{2},\tau\right), \qquad e_3(\tau)=\wp\left(\frac{\tau}{2},\tau\right),
\end{gather}
we can show that $e_1$, $e_2$, $e_3$ are the roots of \eqref{algebraic of Weierstrass p function}
\begin{gather}\label{differential equation of Weierstrass p function}
4\wp^3-g_2\wp-g_3=4\left( \wp-e_1 \right)\left( \wp-e_2 \right)\left( \wp-e_3 \right).
\end{gather}
Therefore, the discriminant of the cubic \eqref{differential equation of Weierstrass p function} is given by
\[
\Delta(\tau)=g_2(\tau)^3-27g_3(\tau)^2,
\]
or alternatively
\begin{gather}
\Delta(\tau)=16(e_1-e_2)^2(e_1-e_3)^2(e_3-e_2)^2.\label{modular discriminat}
\end{gather}
The modular discriminant \eqref{modular discriminat} is a modular form of weight 12, which only vanishes at the cusp, i.e., at $\mathbb{Q}\cup \{{\rm i}\infty\}$. The Dedeking $\eta$ function is modular form of weight $\frac{1}{2}$ given by
\begin{gather}\label{Dedeking eta function}
\eta(\tau)=q^{\frac{1}{24}}\prod_{n=1}^{\infty} \left( 1-q^n\right), \qquad \text{where} \quad q={\rm e}^{2\pi {\rm i} \tau}.
\end{gather}
It obeys the following transformation laws:
\begin{gather}
\eta ( \tau+1 )={\rm e}^{\frac{\pi {\rm i}}{12}}\eta(\tau),\qquad
\eta\left( \frac{-1}{\tau} \right)=\sqrt{-{\rm i}\tau}\eta(\tau).\label{transformation law of Dedekind eta function}
\end{gather}

The $24$-th power of \eqref{Dedeking eta function} is the modular discriminant up to a constant, i.e.,
\begin{gather}
\Delta(\tau)=(2\pi)^{12}\eta^{24}(\tau).\label{modular discriminat 2}
\end{gather}

Now define the following group homomorphism:
\begin{gather}
\chi_2\colon\ {\rm SL}_2(\mathbb{Z})\mapsto Z_2, \qquad \chi_2{\begin{pmatrix}
a& b \\
c&d \\
\end{pmatrix}}=\frac{\eta^{12}\bigl( \frac{a\tau+b}{c\tau+d} \bigr)}{ (c\tau+d )^6\eta^{12}(\tau)},\nonumber\\
\chi_3\colon\ {\rm SL}_2(\mathbb{Z})\mapsto Z_3, \qquad \chi_3{\begin{pmatrix}
a& b \\
c&d \\
\end{pmatrix}}=\frac{\eta^8\bigl( \frac{a\tau+b}{c\tau+d} \bigr)}{ (c\tau+d )^4\eta^8(\tau)}.\label{transformation law of powers of Dedeind eta 1}
\end{gather}
More explicitly, substituting \eqref{transformation law of Dedekind eta function} in \eqref{transformation law of powers of Dedeind eta 1}, we obtain
\begin{gather*}
\chi_2{\begin{pmatrix}
0& 1 \\
-1&0 \\
\end{pmatrix}}=\zeta_2,\qquad
\chi_2{\begin{pmatrix}
1& 1 \\
0&1 \\
\end{pmatrix}}=\zeta_2, \qquad \zeta_2={\rm e}^{\pi {\rm i}},\\
\chi_3{\begin{pmatrix}
0& 1 \\
-1&0 \\
\end{pmatrix}}=1, \qquad
 \chi_3{\begin{pmatrix}
1& 1 \\
0&1 \\
\end{pmatrix}}=\zeta_3^2, \qquad \zeta_3={\rm e}^{\frac{2\pi {\rm i}}{3}}.
\end{gather*}
We are interested in the subgroups of ${\rm SL}_2(\mathbb{Z})$ provided by $\Gamma^{(n)}:=\operatorname{Ker}\chi_n$.

For our current purposes, we will consider holomorphic functions $f\colon\mathbb{H}\mapsto \mathbb{C}$ that are invariant under ${\rm SL}_2(\mathbb{Z})$, $\Gamma^{(2)}$, $\Gamma^{(3)}$. Consider the function
$
j\colon\mathbb{H}\mapsto \mathbb{C}
$
given by
\begin{gather}\label{j function}
j(\tau)=1728\frac{g_2^3(\tau)}{\Delta(\tau)}, \qquad \text{or alternatively} \qquad j(\tau)=1728\frac{E_4^3(\tau)}{E_4^3(\tau)-E_6^2(\tau)}.
\end{gather}
The function \eqref{j function} is holomorphic, surjective, and ${\rm SL}_2(\mathbb{Z})$-invariant. Furthermore, it is a~local biholomorphism in the subspace \smash{$\mathbb{H}\setminus\bigl\{{\rm SL}_2(\mathbb{Z})\bigl({\rm e}^{\frac{2\pi {\rm i}}{3}}\bigr)\bigr\}\cup\{{\rm SL}_2(\mathbb{Z})({\rm i})\}$}, providing the following branched covering map
$
j\colon \mathbb{H}\mapsto \mathbb{C}\cong\mathbb{H}/{\rm SL}_2(\mathbb{Z})$.
The Fourier expansion of \eqref{j function}
$j(q)=\frac{1}{q}+744+O(q)$
give a compactification of $\mathbb{H}/{\rm SL}_2(\mathbb{Z}) $ as $\mathbb{H}/{\rm SL}_2(\mathbb{Z})^C\cong \mathbb{CP}^1 $. In this way, up affine shift and rescaling constant, by Liouville's theorem, the $j$-function \eqref{j function} is the unique modular~${\rm SL}_2(\mathbb{Z})$ invariant function, which it is called by Hauptmodul of ${\rm SL}_2(\mathbb{Z})$. A~univalent automorphic function f for a genus zero Fuchsian group is called a Hauptmodul. The Hauptmodul of the group $\Gamma^{(2)}$, $\Gamma^{(3)}$ are given by the following Weber function, see \cite[Section 3.4]{Atkin},
\begin{gather}
\gamma_2(\tau)=j^{\frac{1}{3}}(\tau), \qquad \text{or alternatively,}\qquad \gamma_2(\tau)=\frac{12}{(2\pi)^4}\frac{g_2(\tau)}{\eta^8(\tau)},\nonumber \\
\gamma_3(\tau)=\sqrt{1728-j(\tau)}, \qquad \text{or alternatively,}\qquad \gamma_3(\tau)=i\frac{216}{(2\pi)^6}\frac{g_3(\tau)}{\eta^{12}(\tau)},\label{Weber functions}
\end{gather}
 the cube root being chosen in such a way that \smash{$j^{\frac{1}{3}}(\tilde\tau)$} is positive on the imaginary axis.

 Indeed, consider $\bigl(\begin{smallmatrix}
a_2& b_2 \\
c_2&d_2 \\
\end{smallmatrix}\bigr) \in \Gamma^{(3)}$, $\bigl(\begin{smallmatrix}
a_3& b_3 \\
c_3&d_3 \\
\end{smallmatrix}\bigr) \in \Gamma^{(2)}$. Then, using that $\Gamma^{(n)}:=\operatorname{Ker}\chi_n$ and transformation law of $\eta(\tau)$, $g_2(\tau)$, $g_3(\tau)$ defined in \eqref{Transformation law Einsenstein series} and \eqref{transformation law of Dedekind eta function}, we have the following:
\begin{gather*}
 \gamma_2\left(\frac{a_2\tau+b_2}{c_2\tau+d_2}\right)=\frac{12}{(2\pi)^4}\frac{g_2\bigl(\frac{a_2\tau+b_2}{c_2\tau+d_2}\bigr)}{\eta^8\bigl(\frac{a_2\tau+b_2}{c_2\tau+d_2}\bigr)}
 =\frac{1}{\chi_3{\begin{pmatrix}
a_2& b_2 \\
c_2&d_2 \\
\end{pmatrix}}}\frac{12}{(2\pi)^4}\frac{(c_2\tau+d_2)^4g_2(\tau)}{(c_2\tau+d_2)^4\eta^8(\tau)}= \gamma_2(\tau),
\\
 \gamma_3\left(\frac{a_3\tau+b_3}{c_3\tau+d_3}\right)=i\frac{216}{(2\pi)^6}\frac{g_3\left(\frac{a_3\tau+b_3}{c_3\tau+d_3}\right)}{\eta^{12}\left(\frac{a_3\tau+b_3}{c_3\tau+d_3}\right)} =\frac{1}{\chi_2{\begin{pmatrix}
a_3& b_3 \\
c_3&d_3 \\
\end{pmatrix}}}i\frac{216}{(2\pi)^6}\frac{\left(c_3\tau+d_3\right)^6g_3(\tau)}{\left(c_3\tau+d_3\right)^6\eta^{12}(\tau)}= \gamma_3(\tau).
\end{gather*}
From \cite[Section 3.1]{Booher}, the Weber function $\gamma_2(\tau)$ has the following transformation law under~${\rm SL}_2(\mathbb{Z})$:
\begin{gather}\label{gamma 2 transformation law in section 3.1}
\gamma_2\left( \frac{a\tau+b}{c\tau+d} \right)=\zeta_3^{ac-ab+a^2cd-cd}\gamma_2(\tau), \qquad {\begin{pmatrix}
a& b \\
c&d \\
\end{pmatrix}} \in {\rm SL}_2(\mathbb{Z}),
\end{gather}
and the group $\Gamma^{(3)}$ is represented as follows:
\begin{gather}\label{monodromy group gamma3}
\Gamma^{(3)}:=\operatorname{Ker}\chi_3=\left\{ \begin{pmatrix}
a & b\\
c& d\\
\end{pmatrix} \in {\rm SL}_2(\mathbb{Z}) \mid \begin{pmatrix}
0 & *\\
*& 0\\
\end{pmatrix} \ or \ \begin{pmatrix}
* & b\\
b& *\\
\end{pmatrix} \mod 3\right\}.
\end{gather}
Due to the group homomorphism \eqref{transformation law of powers of Dedeind eta 1}, the group $\Gamma^{(3)}$ is a subgroup of ${\rm SL}_2(\mathbb{Z})$ with index $3$ and is generated by the following three matrices
\begin{gather}\label{generators of gamma 3 2}
r_1= \begin{pmatrix}
0 & -1\\
1& 0\\
\end{pmatrix}, \qquad
r_2= \begin{pmatrix}
1 & -1\\
2& -1\\
\end{pmatrix}, \qquad r_3= \begin{pmatrix}
-1& 2\\
-1& 1\\
\end{pmatrix}.
\end{gather}

Consider the linear map $B\colon\text{M}_{3\times3}(\mathbb{R})\mapsto\text{M}_{3\times 3}(\mathbb{R})
$,
\begin{gather}\label{change of basis of endomorphism}
B\left( {\begin{pmatrix}
a& b&c \\
d& e& f\\
g& h & i
\end{pmatrix}}\right)=\frac{1}{4}{\begin{pmatrix}
1& 2&-1 \\
0& 1& -1\\
1& 1& -2
\end{pmatrix}}{\begin{pmatrix}
a& b&c \\
d& e& f\\
g& h & i
\end{pmatrix}}{\begin{pmatrix}
2& -6&2 \\
2& 2& -2\\
2& -2 & -2
\end{pmatrix}}.
\end{gather}

The change of basis \eqref{change of basis of endomorphism} sends the generators of \eqref{monodromy of quantum cohomology cp2} to
\begin{gather}
B( R_1 )={\begin{pmatrix}
0& 0&1 \\
0& -1& 0\\
1& 0& 0
\end{pmatrix}},\qquad B\bigl( T_0^4\bigr)={\begin{pmatrix}
0& 0&1 \\
0& -1& 1\\
1& -2 & 1
\end{pmatrix}},\nonumber\\
 B\bigl( T_0^3\bigr)={\begin{pmatrix}
-1& 0&0 \\
0& -1& 0\\
0& 0 & -1
\end{pmatrix}}.\label{new basis}
\end{gather}
Moreover, the generators \eqref{generators of wr Dubrovin} are mapped to
\begin{gather}
B( R_1 )={\begin{pmatrix}
0& 0&1 \\
0& -1& 0\\
1& 0& 0
\end{pmatrix}},\qquad B( R_2)={\begin{pmatrix}
-1& 4&-4 \\
-1& 3& -2\\
-1& 2 & -1
\end{pmatrix}},\nonumber\\
 B( R_3)={\begin{pmatrix}
-1& 2&-1 \\
-2& 3& -1\\
-4& 4 & -1
\end{pmatrix}}.\label{new basis 2}
\end{gather}
Hence, the matrices \eqref{new basis} and \eqref{new basis 2} provide another generators for $W$ and $W_R$, respectively.

To define a group action of $W_R$ and $W$ on the upper half-plane, consider the spaces
\begin{gather*}
\left\{ (\omega_1,\omega_1^{\prime},\omega_2,\omega_2^{\prime}) \in {(\mathbb{C}^{*})}^4 \mid \operatorname{Im}\left(\frac{\omega_i^{\prime}}{\omega_i}\right)>0 ,\, i=1,2\right\}\cong \mathbb{C}^{*}\times \mathbb{H}\times \mathbb{C}^{*}\times \mathbb{H},
\end{gather*}
where the above isomorphism is given explicitly by
\begin{gather*}
 (\omega_1,\omega_1^{\prime},\omega_2,\omega_2^{\prime})\mapsto (\omega_1,\tau_1,\omega_2,\tau_2)=\left(\omega_1,\frac{\omega_1^{\prime}}{\omega_1},\omega_2,\frac{\omega_2^{\prime}}{\omega_2}\right).
\end{gather*}
Note that there exists a natural $A_1 \ltimes ({\rm PSL}_2(\mathbb{Z}) \times {\rm PSL}_2(\mathbb{Z}))$ action on the space
$\mathbb{C}^{*}\times \mathbb{H}\times \mathbb{C}^{*}\times \mathbb{H}$ given by
\begin{gather*}
(\gamma_1\times \gamma_2)(\omega_1,\tau_1,\omega_2,\tau_2)=\left((c_1\tau_1+d_1)\omega_1,\frac{a_1\tau_1+b_1}{c_1\tau_1+d_1},(c_2\tau_2+d_2)\omega_2,\frac{a_2\tau_2+b_2}{c_2\tau_2+d_2}\right),\\
\sigma(\omega_1,\tau_1,\omega_2,\tau_2)=(\omega_2,\tau_2,\omega_1,\tau_1),
\end{gather*}
where
\begin{gather*}
\sigma \in A_1, \qquad \gamma_i={\begin{pmatrix}
a_i& b_i \\
c_i& d_i\\
\end{pmatrix}}\in {\rm SL}_2(\mathbb{Z}), \qquad i=1.2.
\end{gather*}

In addition, consider the symmetric square polynomials with respect the periods $\omega_1$,~$\omega_1^{\prime}$,~$\omega_2$,~$\omega_2^{\prime}$%
\begin{gather}
G\colon\ \mathbb{C}^{*}\times \mathbb{H}\times\mathbb{C}^{*}\times \mathbb{H}\mapsto \mathbb{C}^4,\qquad \left(\omega_1,\frac{\omega_1^{\prime}}{\omega_1},\omega_2,\frac{\omega_2^{\prime}}{\omega_2}\right)\mapsto (w_1,w_2,w_3,w_4),\label{G map}
\end{gather}
where
\begin{gather*}
 w_1=\omega_1^{\prime}\omega_2^{\prime}, \qquad w_2=\frac{\omega_1\omega_2^{\prime}+\omega_1^{\prime}\omega_2}{2}, \qquad w_3=\omega_1\omega_2, \qquad
 w_4=\frac{\omega_1\omega_2^{\prime}-\omega_1^{\prime}\omega_2}{2}.
\end{gather*}

Consider the diagonal action of $A_1 \times {\rm PSL}_2(\mathbb{Z})$ given by
\begin{gather*}
\gamma (\omega_1, \tau_1, \omega_2, \tau_2) = \left((c\tau_1 + d)\omega_1, \frac{a\tau_1 + b}{c\tau_1 + d}, (c\tau_2 + d)\omega_2, \frac{a\tau_2 + b}{c\tau_2 + d}\right), \\
\sigma(\omega_1, \tau_1, \omega_2, \tau_2) = (\omega_2, \tau_2, \omega_1, \tau_1),
\end{gather*}
where
\begin{gather*}
\gamma = \begin{pmatrix}
a & b \\
c & d
\end{pmatrix} \in {\rm SL}_2(\mathbb{Z}), \qquad \sigma \in A_1.
\end{gather*}
Hence, the image of this diagonal action under the map \eqref{G map} preserves the space generated by $w_1$, $w_2$, $w_3$ and the space generated by $w_4^2$. This action restricted to the space generated by~$w_1$,~$w_2$, $w_3$, gives rise to the following group homomorphism:
\begin{gather*}
\rho\colon\ {\rm PSL}_2(\mathbb{Z})\mapsto {\rm SL}_3(\mathbb{Z}), \qquad \rho{\begin{pmatrix}
a& b \\
c&d \\
\end{pmatrix}}={\begin{pmatrix}
a^2& 2ab &b^2\\
ac&ad+bc&bd \\
c^2&2cd&d^2 \\
\end{pmatrix}},
\end{gather*}
Observe that the image of the generators of ${\rm PSL}_2(\mathbb{Z})$ are given by
\[
\rho{\begin{pmatrix}
0& -1 \\
1&0 \\
\end{pmatrix}}={\begin{pmatrix}
0& 0&1 \\
0& -1& 0\\
1& 0& 0
\end{pmatrix}},\qquad \rho{\begin{pmatrix}
0& 1 \\
-1&1 \\
\end{pmatrix}}={\begin{pmatrix}
0& 0&1 \\
0& -1& 1\\
1& -2 & 1
\end{pmatrix}}.
\]
Furthermore, the image of the generators of $\Gamma^{(3)}$ are given by
\begin{gather}
\rho{\begin{pmatrix}
0& -1 \\
1&0 \\
\end{pmatrix}}={\begin{pmatrix}
0& 0&1 \\
0& -1& 0\\
1& 0& 0
\end{pmatrix}},\qquad \rho{\begin{pmatrix}
-1& 2 \\
-1&1 \\
\end{pmatrix}}={\begin{pmatrix}
-1& 4&-4 \\
-1& 3& -2\\
-1& 2 & -1
\end{pmatrix}},\nonumber\\
 \rho{\begin{pmatrix}
1& -1 \\
2&-1 \\
\end{pmatrix}}={\begin{pmatrix}
-1& 2&-1 \\
-2& 3& -1\\
-4& 4 & -1
\end{pmatrix}}.\label{new basis 4}
\end{gather}

 Hence, we have group isomorphism given by
 \begin{gather}\label{symmetric square representation}
 A_1\times {\rm PSL}_2(\mathbb{Z})\mapsto W, \qquad \left(\sigma, {\begin{pmatrix}
a& b \\
c& d \\
\end{pmatrix}}\right)\mapsto \sigma\chi_2 {\begin{pmatrix}
a& b \\
c& d \\
\end{pmatrix}}\rho {\begin{pmatrix}
a& b \\
c& d \\
\end{pmatrix}}.
 \end{gather}

 Moreover, the restriction of the group homomorphism \eqref{symmetric square representation} to $\Gamma^{(3)}$ provides an isomorphism from $\Gamma^{(3)}$ to $W_R$ due to the relations \eqref{new basis 4} and \eqref{generators of gamma 3 2}. Summarising, we demonstrated the results of \cite[Section~4.1]{Milanov} and \cite[Theorem~5.5]{B.Dubrovin3}.

 \begin{Lemma}[{\cite{B.Dubrovin3,Milanov}}]\label{Monodromy QHCP2}
 The monodromy groups $W_R$ and $W$ associated to $QH^{*}\bigl(\mathbb{CP}^2\bigr)$ are
 \begin{gather*}
 W_R=\Gamma^{(3)}, \qquad W=A_1\times {\rm PSL}_2(\mathbb{Z}).
 \end{gather*}
\end{Lemma}

\begin{Remark}
In \cite[Theorem 5.12\,(g)]{Hertling2}, Hertling proved that the group $W_R$ is a free Coxeter group with 3 generators.
\end{Remark}

\subsection[LG superpotential for small quantum cohomology of $\mathbb{CP}^2$]{LG superpotential for small quantum cohomology of $\boldsymbol{\mathbb{CP}^2}$}\label{LG superpotential for small Quantum Cohomology}

In this subsection, we construct a LG superpotential for the small quantum cohomology of $\mathbb{CP}^2$ using the Dubrovin superpotential construction discussed in Section~\ref{Construction of Dubrovin superpotential}. To proceed with this approach, we require some results from Milanov \cite[Section~4]{Milanov}.

 Recall that the Dubrovin--Frobenius manifold associated with quantum cohomology of~$\mathbb{CP}^2$ is the domain of the WDVV solution \eqref{generating function of Gromov Witten CP2}, which is given by
 \begin{gather}\label{domain Quantum cohomology 1}
 D=\left\{ \bigl(t^1,t^2,t^3\bigr)\in \mathbb{C}^3\mid \bigl|{\rm e}^{t^2}\bigl(t^3\bigr)^3\bigr|<\frac{1}{a}\right\}.
 \end{gather}
 In this context, the period map is the set of three independent solutions of the system \eqref{Gauss--Manin connection}
\begin{gather}\label{period of quantum cohomology cp2}
w_1\bigl(t^1,t^2,t^3\bigr), \qquad w_2\bigl(t^1,t^2,t^3\bigr), \qquad w_3\bigl(t^1,t^2,t^3\bigr).
\end{gather}

 By introducing the change of coordinates
 \begin{gather}\label{Novikov variable Q}
 t^1=t^1, \qquad Q={\rm e}^{t^2}, \qquad t^3=t^3,
 \end{gather}
 we can express the domain \eqref{domain Quantum cohomology 1} in an alternative way
 \[
 D=\left\{ \bigl(t^1,t^2,Q\bigl(t^3\bigr)^3\bigr)\in \mathbb{C}^3\mid \bigl|Q\bigl(t^3\bigr)^3\bigr|<\frac{1}{a}\right\} \cong \mathbb{C}^2\times\mathbb{H}.
 \]

 In order to simplify the problem, in \cite[Section 4]{Milanov}, Milanov considered the sublocus
 \begin{gather*}
 D_{\rm asmall}=\bigl\{ \bigl(t^1,t^2,0\bigr)\in \mathbb{C}^3\bigr\} \cong \mathbb{C}^2,
 \end{gather*}
which is now denoted as affine small quantum cohomology. In addition, we use the change of coordinates \eqref{Novikov variable Q} to conveniently write the domain $ D_{\rm asmall}$ as
 \begin{gather}\label{D small affine}
 \mathcal{D}_{\rm asmall}=\bigl\{ \bigl(t^1,Q,0\bigr)\in \mathbb{C}^3\bigr\} \cong \mathbb{C}\times\mathbb{C}^{*}.
 \end{gather}
 The discriminant locus \eqref{discriminant} of $QH^{*}\bigl(\mathbb{CP}^2\bigr)$ in the sublocus \eqref{D small affine} is explicitly written as
 \begin{gather*}
 \det{\begin{pmatrix}
3Q& 0&t^1 \\
0& t^1& 3\\
t^1& 3 & 0
\end{pmatrix}}=0.
\end{gather*}
 Hence, the period domain for the affine small quantum cohomology of $\mathbb{CP}^2$ is given by
\begin{gather}\label{affine small quantum cohomology domain}
 \mathcal{D}_{\rm asmall}\setminus{\Sigma_{\rm asmall}}=\bigl\{ \bigl(t^1,Q\bigr) \in \mathbb{C}\times \mathbb{C}^{*} \mid \bigl(t^1\bigr)^3+27Q\neq 0\bigr\}.
\end{gather}
At this stage, we follow the Milanov notation in \cite[Sections 1.3 and 4]{Milanov}. Recall that the solution~${w_E\bigl(t^1,t^2,\dots,t^n,\lambda\bigr)}$ of the extended Gauss--Manin connection flat coordinate system \eqref{Extended Gauss--Manin connection} is obtained from the solution $w\bigl(t^1,t^2,\dots,t^n\bigr)$ of the Gauss--Manin connection flat coordinate system~\eqref{Gauss--Manin connection} by shifting with $-\lambda$. In other words, we have the symmetry relation
\begin{gather}\label{symmetry Gauss Manin connection}
 w_E\bigl(t^1,t^2,\dots,t^n,\lambda\bigr)=w\bigl(t^1-\lambda, t^2,\dots,t^n\bigr).
 \end{gather}
Due to the symmetry \eqref{symmetry Gauss Manin connection}, Milanov considered the following restriction to the extended Gauss--Manin connection of $QH^{*}\bigl(\mathbb{CP}^2\bigr)$ in the small quantum cohomology locus
\begin{gather}\label{period map extended Milanov beg}
w_E(0,Q,0,\lambda)=w(-\lambda,Q,0).
\end{gather}
Note that the period domain of \eqref{period map extended Milanov beg} is given by
\begin{gather}\label{affine small quantum cohomology domain Milanov Extended}
 \mathcal{D}_{\text{asmall},\lambda}\setminus{\Sigma_{\text{asmall},\lambda}}=\bigl\{ (\lambda,Q) \in \mathbb{C}\times \mathbb{C}^{*} \mid (\lambda)^3-27Q\neq 0\bigr\}.
\end{gather}

There is no structural difference by considering the solutions \eqref{Gauss--Manin connection} in the affine small quantum cohomology locus $w\bigl(t^1,Q\bigr)$ or the restricted solutions $w_E(0,Q,0,\lambda)$ of \eqref{Extended Gauss--Manin connection}. The only issue to which we need to pay attention is the sign change from $t^1$ to $-\lambda$. Currently, we follow \cite[Section 4.2]{Milanov} and use $\lambda$ instead of $t^1$ and later we change back to $t^1$.

The system \eqref{Extended Gauss--Manin connection} in the affine small quantum cohomology locus \eqref{affine small quantum cohomology domain Milanov Extended} is reduced to
\begin{gather}
\bigl( (Q\partial_Q)^3-Q\partial^3_{\lambda} \bigr)w(\lambda,Q)=0,\qquad
(\lambda\partial_{\lambda}+3Q\partial_{Q} )w(\lambda,Q)=-\frac{w(\lambda,Q)}{2},\label{Gauss Manin connection of affine quantum cohomology 2}
\end{gather}
where
\begin{gather}\label{basic solution w123}
 w(\lambda,Q)=(w_1(\lambda,Q),w_2(\lambda,Q),w_3(\lambda,Q)) .
\end{gather}

The second equation of \eqref{Gauss Manin connection of affine quantum cohomology 2} implies that $w(\lambda,Q)$ is of the form
\smash{$
w(\lambda,Q)=Q^{-\frac{1}{6}}f(x)$}, \smash{$ x:=\frac{(\lambda)^3}{27Q}$}.
Moreover, due the first equation of \eqref{Gauss Manin connection of affine quantum cohomology 2}, the vector valued function $f(x)$
solves some generalized hypergeometric equation. This equation has a special basis of solution of the symmetric square form. More precisely, writing a basis of solutions in the form
$
f_1=u^2$, $ f_2=uv$, $ f_3=v^2$,
we have that $u$, $v$ are solutions of a differential equation, which is equivalent to the classic hypergeometric equation
\begin{gather}\label{hypergeometric of small quanutm cohomology}
\bigl((1-x)x\partial^2_x+(c-(a+b+1 )x )-ab\bigr)z=0, \qquad a=b=\frac{1}{12},\qquad c=\frac{2}{3}.
\end{gather}
Moreover, due to the \cite[Lemma 4.1]{Milanov}, the solutions \eqref{period of quantum cohomology cp2} in the affine small quantum cohomology locus \eqref{affine small quantum cohomology domain} satisfies a quadratic relation
\begin{gather}\label{Quadratic relation t3 equal 0}
w_2^2=4w_1w_3.
\end{gather}

Therefore, there exist a basis of solution of \eqref{Gauss Manin connection of affine quantum cohomology 2} of the form
\begin{gather*}
w(\lambda,Q)=\bigl( \tau^2(x)f_3(x)Q^{\frac{-1}{6}} , -2\tau(x)f_3(x)Q^{\frac{-1}{6}} , f_3(x)Q^{\frac{-1}{6}} \bigr), \qquad x=\frac{(\lambda)^3}{27Q},
\end{gather*}
where
\begin{gather}\label{Schwarz map}
\tau(x)=-\frac{w_2\bigl(t^1,Q\bigr)}{2w_3\bigl(t^1,Q\bigr)}=-\frac{f_2(x)}{2f_3(x)}.
\end{gather}

The hypergeometric equation \eqref{hypergeometric of small quanutm cohomology} has singularities in 0, 1, $\infty$. As a consequence, the solutions of \eqref{Gauss Manin connection of affine quantum cohomology 2} are multivalued function with branch points at 0, 1, $\infty$. Denote $l_0$, $l_1$, $l_{\infty}$ the loops around 0, 1, $\infty$, in this setting, we write the \cite[Proposition~4.2]{Milanov} as follows.

\begin{Proposition}[{\cite{Milanov}}]\label{Milanov lemma small 1}\quad
\begin{enumerate}\itemsep=0pt
\item[$(1)$]
The analytic continuation of the Schwarz map \eqref{Schwarz map} gives a representation
$\pi_1\bigl(\mathbb{CP}^1\setminus\{0,1,\infty\}\bigr)\mapsto \operatorname{Aut}(\mathbb{H})$
such that the generators $l_1$, $l_{\infty}$ are mapped to the transformations~${
\tau\mapsto \frac{-1}{\tau}}$, $ \tau\mapsto \tau+1$,
respectively.
\item[$(2)$] The image of the lower and upper half plane $\operatorname{Im}(x)\leq 0$, $\operatorname{Im}(x)\geq 0$ under the Schwarz map are the hyperbolic triangles with vertices $\infty$, \smash{${\rm e}^{\frac{\pi {\rm i}}{3}}$}, ${\rm i}$ and $\infty$, \smash{${\rm e}^{2\frac{\pi {\rm i}}{3}}$}, ${\rm i}$, respectively.
\end{enumerate}
\end{Proposition}

Recall that the ${\rm SL}_2(\mathbb{Z})$ orbits of the vertices ${\rm e}^{\frac{\pi {\rm i}}{3}}$ and $i$ can be written as zeros of the Eisenstein series $E_4$, $E_6$ defined in \eqref{Einsenstein series 2}. Indeed
\begin{gather*}
E_4\left(\frac{1}{1-{\rm e}^{\frac{\pi {\rm i}}{3}} }\right)=E_4\bigl({\rm e}^{\frac{\pi {\rm i}}{3}}\bigr)=\bigl({\rm e}^{\frac{\pi {\rm i}}{3}}\bigr)^4E_4\bigl({\rm e}^{\frac{\pi {\rm i}}{3}}\bigr)=0,\qquad
E_6({\rm i})=(-{\rm i})^6E_6({\rm i})=-E_6({\rm i})=0.
\end{gather*}
Hence, the covering map associated to analytic continuation of Schwarz map \eqref{Schwarz map} is given by
\begin{gather}\label{function J covering definition}
J\colon\ \mathbb{H}\setminus\{E_4(\tau)=0\}\cup\{E_6(\tau)=0\}\mapsto \mathbb{C}\setminus\{0,1\}.
\end{gather}
Remarkably, the monodromy and the Fourier expansion around 0, 1, $\infty$ of the hypergeometric equation \eqref{hypergeometric of small quanutm cohomology} are well know. Milanov used these results to prove the \cite[Lemma 4.3]{Milanov}, which we state here.

\begin{Lemma}[{\cite{Milanov}}]\label{Milanov lemma small 2}\quad
\begin{enumerate}\itemsep=0pt

\item[$(1)$]
The function $J$ defined in \eqref{function J covering definition} extends to a holomorphic function on $\mathbb{H}$ and it coincides with the unique ${\rm SL}_2(\mathbb{Z})$ holomorphic invariant function such that $J({\rm i})=1$, i.e.,
\begin{gather}\label{J function other normalization}
J(\tau)=\frac{1}{1728}j(\tau),
\end{gather}
where $j(\tau)$ is defined in \eqref{j function}.
\item[$(2)$]
The pullback of any branch of $\frac{8}{27}(2\pi)^6f_3(x)^{6}$ to $\mathbb{H}$ via the map $J$ extends to a holomorphic function on $\mathbb{H}$ and its coincides with the modular form $E_4^3(\tau)-E_6^2(\tau)$.
\end{enumerate}

\end{Lemma}
From now, we return to the variable $t^1$ by replacing $\lambda$ to $-t^1$.

Using Lemma \ref{Milanov lemma small 2}, we can write the discriminant locus $\bigl(t^1\bigr)^3+27Q=0$ as the zeros of $E_6$. Indeed,
\begin{gather*}
J(\tau)=-\frac{\bigl(t^1\bigr)^3}{27Q} \implies J(i)=-\frac{\bigl(t^1\bigr)^3}{27Q}=1 \qquad \text{iff} \qquad \bigl(t^1\bigr)^3+27Q=0.
\end{gather*}

Due to Proposition~\ref{Milanov lemma small 1}, Lemma~\ref{Milanov lemma small 2} and the quadratic relation \eqref{Quadratic relation t3 equal 0}, the period map~\eqref{Gauss Manin connection of affine quantum cohomology 2}
\begin{gather*}
w\colon\ \mathcal{D}_{\rm asmall}\setminus{\Sigma_{\rm asmall}}\mapsto \Omega_{\rm small}\setminus\left\{ E_6\left(\frac{-w_2}{2w_3}\right)=0\right\}
\end{gather*}
 is given explicitly by
\begin{gather}
\tau=\frac{w_2}{2w_3}\bigl(t^1,Q\bigr)=J^{-1}\left( -\frac{\bigl(t^1\bigr)^3}{27Q} \right),\nonumber\\
r:=w_3\bigl(t^1,Q\bigr)=2\pi\sqrt{\frac{2}{3}}\left(\frac{\bigl(E_4^3(\tau)-E_6^2(\tau) \bigr)}{Q}\right)^{\frac{1}{6}},\label{period map functions}
\end{gather}
where
\[
\Omega_{\rm small}=\left\{ (w_1,w_2,w_3) \in \mathbb{C}^3 \mid w_2^2-4w_1w_3=0,\, \operatorname{Im}\left(\frac{-w_2}{2w_3}\right)> 0\right\}.
\]

In \cite[Section 2.3]{Milanov}, it was introduced the following bijection:
\[
\phi\colon\ \mathbb{H}\times \mathbb{C}^{*} \mapsto \Omega_{\rm small},\qquad
(\tau,r) \mapsto \bigl(\tau^2r,2\tau r, r\bigr).
\]

At this stage, we can state the \cite[Theorem 2.3]{Milanov} as follows.

\begin{Lemma}[{\cite{Milanov}}]\label{Lemma Milanov affine small quantum cohomology}\quad
\begin{enumerate}\itemsep=0pt
\item[$(1)$]
The inverse period map of quantum cohomology of $\mathbb{CP}^2$ for $t^3=0$ is the map
\begin{gather*}
t\colon\ (\mathbb{H}\setminus\{E_6(\tau)=0 \})\times \mathbb{C}^{*} \mapsto \bigl\{ \bigl(t^1,Q\bigr) \in \mathbb{C}\times \mathbb{C}^{*} \mid \bigl(t^1\bigr)^3+27Q\neq 0\bigr\}
\end{gather*}
given by
\begin{gather}
t^1(\tau,r)=-2\frac{ (2\pi )^2}{r^2} E_4(\tau) ,\qquad
Q(\tau,r)=\frac{8}{27}\frac{(2\pi)^6 }{r^6} \bigl(E_4^3(\tau)-E_6^2(\tau) \bigr) ,\nonumber\\
t^3=0,\label{inverse period map of small quantum cohomology}
\end{gather}
where $E_k(\tau)$ are defined in \eqref{Einsenstein series 2}. Moreover, the map \eqref{inverse period map of small quantum cohomology} extends to the holomorphic function
$
t\colon \mathbb{H}\times \mathbb{C}^{*} \mapsto \mathbb{C}\times \mathbb{C}^{*}$.
\item[$(2)$]
The fibers of the map \eqref{inverse period map of small quantum cohomology} are $A_1\times {\rm PSL}_2(\mathbb{Z})$-invariant and its concrete action in the domain of \eqref{inverse period map of small quantum cohomology} is generated by
\begin{gather}
A(\tau,r)=\left( \frac{a\tau+b}{c\tau+d} , (c\tau+d)^2\right)r,\qquad {\begin{pmatrix}
a& b \\
c& d \\
\end{pmatrix}} \in {\rm SL}_2(\mathbb{Z}),\nonumber\\
B(\tau,r)=( \tau , -r).\label{Monodromy action affine small quantum cohomology}
\end{gather}
\end{enumerate}
\end{Lemma}

Using the Lemma \ref{Lemma Milanov affine small quantum cohomology}, we can prove the main result of this section.

\begin{Theorem}\quad
\begin{enumerate}\itemsep=0pt
\item[$(1)$]
The Landau--Ginzburg superpotential of small affine quantum cohomology of $\mathbb{CP}^2$ is a~family of functions
$
\lambda\bigl(\tilde\tau, t^1, Q\bigr)\colon \mathbb{H} \mapsto \mathbb{C}
$
with holomorphic dependence in the parameter space \smash{$
\bigl(t^1,Q^{\frac{1}{3}}\bigr)\in \mathbb{C}\times\mathbb{C}^{*}
$}
 and given by
\begin{gather}\label{lg superpotential of small quantum cohomology j function }
\lambda\bigl(\tilde\tau, t^1, Q\bigr)=t^1+3Q^{\frac{1}{3}}J^{\frac{1}{3}}(\tilde\tau)
\end{gather}
 In addition, the correspondent Abelian differential $\phi$ is given by
\begin{gather}\label{small volume form 1}
\phi=-\frac{2^{\frac{5}{2}}}{2\pi}\frac{\Delta^{\frac{1}{6}}(\tilde\tau )}{Q^{\frac{1}{6}}}{\rm d}\tilde\tau,
\end{gather}
where $\Delta(\tilde\tau)$ is defined in \eqref{modular discriminat 2}.

\item[$(2)$] The Landau--Ginzburg superpotential \eqref{lg superpotential of small quantum cohomology j function } is invariant with respect the group $\Gamma^{(3)}$, i.e.,
$
\lambda\bigl(\frac{a\tilde\tau+b}{c\tilde\tau+d}\bigr)=\lambda(\tilde\tau)$.
\end{enumerate}
\end{Theorem}

\begin{proof}

The Dubrovin construction of Landau--Ginzburg superpotential applies to any Dub\-ro\-vin--Frobenius manifold with $d\neq 1$, in order to construct concretely the Landau--Ginzburg superpotentialone should only choice a suitable solution of the Gauss--Manin connection $p(\lambda,u)$ and invert it in $\lambda$. Consider a basic solution
$
(x_1(\lambda,t),x_2(\lambda,t),x_3(\lambda,t))
$
 of the extended Gauss--Manin connection of $QH^\ast\bigl(\mathbb{CP}^2\bigr)$, which is given by \eqref{asymptotic behaviour of phi 1}, \eqref{reflection matrix generators} and \eqref{final form quasi homogeneous}. Then, consider the change of basis
\begin{gather*}
 {\begin{pmatrix}
\tilde w_1 \\
\tilde w_2\\
\tilde w_3
\end{pmatrix}}={\begin{pmatrix}
1& 2&-1 \\
0& 1& -1\\
1& 1& -2
\end{pmatrix}} {\begin{pmatrix}
x_1 \\
x_2\\
x_3
\end{pmatrix}}
 \end{gather*}
 sends the basic solution $(x_1(\lambda,t),x_2(\lambda,t),x_3(\lambda,t))$ to \eqref{basic solution w123}. In order to obtain the Landau--Ginzburg superpotential of affine small quantum cohomology of $\mathbb{CP}^2$, we choice
$
p\bigl(\lambda,t^1,Q,0\bigr):=\tilde w_2\bigl(\lambda,t^1, Q,0\bigr)
$
and apply the Lemma \ref{lemma primitive form} and Theorem \ref{Dubrovin superpotential theorem} to $p\bigl(\lambda,t^1,Q,0\bigr)$. On another words, we invert $p\bigl(\lambda,t^1,Q,0\bigr):=\tilde w_2\bigl(\lambda,t^1, Q,0\bigr)$ in $\lambda$.

Inverting the map provided by \eqref{inverse period map of small quantum cohomology}, we obtain the functions defined in \eqref{period map functions}.

Next, consider the affine extended $\tilde w_2$, $\tilde w_3$,
\begin{gather}
\tilde\tau:=\frac{-\tilde w_2}{2\tilde w_3}\bigl(t^1-\lambda,Q\bigr)=J^{-1}\left( \frac{\bigl(\lambda-t^1\bigr)^3}{27Q} \right),\nonumber\\
\tilde r:=\tilde w_3\bigl(t^1-\lambda,Q\bigr)=2\pi\sqrt{\frac{2}{3}}\left(\frac{\bigl(E_4^3(\tilde\tau)-E_6^2(\tilde \tau) \bigr)}{Q}\right)^{\frac{1}{6}}.\label{affine extended w2 w3}
\end{gather}

Inverting the first equation of \eqref{affine extended w2 w3} in $\lambda$, we obtain
\begin{gather}\label{final relation}
\lambda\bigl(\tilde\tau, t^1, Q\bigr)=t^1+3Q^{\frac{1}{3}}J^{\frac{1}{3}}(\tilde\tau).
\end{gather}
Hence, by Theorem \ref{Dubrovin superpotential theorem} the data
\smash{$
\bigl( \lambda\bigl(\tilde\tau, t^1, Q\bigr)=t^1+3Q^{\frac{1}{3}}J^{\frac{1}{3}}(\tilde\tau), {\rm d}\tilde w_2 \bigr)$},
 where $\tilde w_2$ is the pullback of the map $\tilde w_2:=w_2\bigl(t^1-\lambda,Q,0\bigr)$ via the map \eqref{lg superpotential of small quantum cohomology j function }, gives a Landau--Ginzburg superpotential for the affine small quantum cohomology of $\mathbb{CP}^2$.

The domain of \eqref{lg superpotential of small quantum cohomology j function } is a family of $\mathbb{H}$ parametrized by
$
(\tau,r)\in (\mathbb{H}\setminus\{E_6(\tau)=0 \} )\times \mathbb{C}^{*}$,
because of the parametrized LG superpotential \eqref{final relation} and the change of coordinates \eqref{inverse period map of small quantum cohomology}.

In order to write ${\rm d}\tilde w_2$ in terms the coordinate $\tilde\tau$, it is convenient to consider the following spaces.

Here consider this family of $\mathbb{H}$ defined by
\[
\mathbb{H}_{\tau,r}=\left\{ (\tilde\tau,\tilde r ,\tau,r) \in \mathbb{H}\times \mathbb{C}^{*}\times\mathbb{H}\times \mathbb{C}^{*}\mid \tilde r=2\pi\sqrt{\frac{2}{3}}\left(\frac{\bigl(E_4^3(\tilde\tau)-E_6^2(\tilde \tau) \bigr)}{Q}\right)^{\frac{1}{6}},\, E_6(\tau)\neq 0 \right\}.
\]

Let the space $\Omega_{w_2,r}$ be defined by
\begin{gather*}
\Omega_{w_2,r}=\left\{ (w_2,r) \in \mathbb{C}^{*}\times \mathbb{C}^{*} \mid \operatorname{Im}\left(\frac{-w_2}{2r} \right)>0 \right\}.
\end{gather*}

Consider the following isomorphism $S\colon\mathbb{H}\times \mathbb{C}^{*}\times\mathbb{H}\times \mathbb{C}^{*} \mapsto \Omega_{\tilde w_2,\tilde r}\times \Omega_{w_2,r}$
$
S(\tilde\tau,\tilde r,\tau, r)=(\tilde w_2, \tilde r, w_2, r)$, where $\tilde w_2=-2\tilde r\tilde\tau$, $ w_2= -2r\tau$.
Note that the section ${\rm d}\tilde w_2 \in \Gamma( T^{*}(\mathbb{H}_{\tau,r} ))$ is the pullback of the section ${\rm d}\tilde w_2 \in \Gamma( T^{*}( \Omega_{\tilde w_2,\tilde r}\times\allowbreak \Omega_{w_2,r} ))$ and it is given by
\begin{gather*}
{\rm d}\tilde w_2=-2\tilde r{\rm d}\tilde\tau=-\frac{2^{\frac{5}{2}}}{2\pi}\frac{\Delta^{\frac{1}{6}}(\tilde\tau )}{Q^{\frac{1}{6}}}{\rm d}\tilde\tau.
\end{gather*}
Here, it was used the following relation
\smash{$
\Delta(\tau)=(2\pi)^{12}\frac{E_4^3(\tau)-E_6^2(\tau)}{1728}$}.
The Landau--Ginzburg po\-tential \eqref{lg superpotential of small quantum cohomology j function } is $\Gamma^{(3)}$ invariant, because the function \smash{$J^{\frac{1}{3}}(\tilde\tau)$} is $\Gamma^{(3)}$ invariant due to \eqref{gamma 2 transformation law in section 3.1} and~\eqref{monodromy group gamma3}, see \cite[Section 3.1]{Booher} for details.
Theorem proved.
\end{proof}

Even though, the construction given by Dubrovin guarantee that the function \eqref{lg superpotential of small quantum cohomology j function } works as superpotential for the affine small quantum cohomology of $\mathbb{CP}^2$. It would be instructive to directly verify the first condition.

Note that the LG variable $p$ in the notation of \eqref{Landau--Ginzburg superpotential definition} is the variable $\tilde w_2$ and not $\tilde\tau$. Due to~\eqref{affine extended w2 w3} and the relation
\smash{$
\Delta(\tau)=(2\pi)^{12}\frac{E_4^3(\tau)-E_6^2(\tau)}{1728}$},
$\tilde\tau$ and $\tilde w_2$ are related by the following:
\[
\tilde\tau=-\frac{\tilde w_2}{2\tilde r} ,\qquad
\frac{\partial}{\partial \tilde w_2}=-\frac{2\pi}{2^{\frac{5}{2}}}\frac{Q^{\frac{1}{6}}}{\Delta^{\frac{1}{6}}(\tilde\tau)}\frac{\partial}{\partial \tilde\tau}.
\]
 Then,
\begin{align}
\frac{\partial\lambda}{\partial \tilde w_2}&=-\frac{(2\pi)\times 3}{2^{\frac{5}{2}}}\frac{Q^{\frac{1}{2}}}{\Delta^{\frac{1}{6}}(\tilde\tau)}\frac{\partial J^{\frac{1}{3}}(\tilde\tau) }{\partial \tilde\tau}=-\frac{(2\pi)\times 3}{2^{\frac{5}{2}}}\frac{Q^{\frac{1}{2}}}{\Delta^{\frac{1}{6}}(\tilde\tau)}\frac{E_4^{\prime}(\tilde\tau)-\frac{E_2(\tilde\tau)E_4(\tilde\tau)}{3}}{\Delta^{\frac{1}{3}}(\tilde\tau)}\nonumber\\
&=\frac{(2\pi)}{2^{\frac{5}{2}}}Q^{\frac{1}{2}}\frac{E_6(\tilde\tau)}{\Delta^{\frac{1}{2}}(\tilde\tau)},\label{small lg superpotential derivative}
\end{align}
where Ramanujan identities \cite{Zagier} were used from the second to the third line of \eqref{small lg superpotential derivative}
\begin{gather}
E_2^{\prime}=\frac{E_2^2-E_4}{12},\qquad E_4^{\prime}=\frac{E_2E_4-E_6}{3},\qquad E_6^{\prime}=\frac{E_2E_6-E_4^2}{2}.\label{Ramanujan identities}
\end{gather}
Then, $
\frac{\partial\lambda}{\partial \tilde w_2}=0 $ iff $ E_6(\tilde\tau)=0 $.
The zeros of $E_6$ are the elements of orbit of the imaginary number~$i$ under the ${\rm SL}_2(\mathbb{Z})$ action~\cite{Zagier}. Recall that the function
\smash{$
\gamma_2(\tilde\tau)= j^{\frac{1}{3}}(\tilde\tau)
$}
is the Weber function \eqref{Weber functions}, which is known to be the Hauptmodul of the group $\Gamma^{(3)}$ defined in \eqref{monodromy group gamma3}.
In~addition, we have that
\smash{$
J^{\frac{1}{3}}(\tau)=\frac{j^{\frac{1}{3}}(\tau)}{12}=\frac{\gamma_2(\tau)}{12}$}.

In this context, the zeros of $E_6$ are $i-1$, $i$, $i+1$ mod $\Gamma^{(3)}$. Moreover, using the relation~\eqref{gamma 2 transformation law in section 3.1}, we derive
\smash{$
\frac{\gamma_2(i)}{12}=1$}, \smash{$ \frac{\gamma_2(i+1)}{12}={\rm e}^{\frac{-2\pi {\rm i}}{3}}$}, \smash{$ \frac{\gamma_2(i-1)}{12}={\rm e}^{\frac{2\pi {\rm i}}{3}}$}.

As a consequence,
\begin{gather}\label{canonical coordinates of small quantum cohomology}
u_k=t^1+3Q^{\frac{1}{3}}\bigl( {\rm e}^{\frac{2\pi {\rm i}}{3}}\bigr)^k, \qquad k=1,2,3.
\end{gather}

\begin{Remark}\label{zero of LG small quantum cohomology}
The zeros of the LG superpotential \eqref{lg superpotential of small quantum cohomology j function } is given by
$
\tilde\tau=\tau$.
Indeed,
\begin{align*}
\lambda\bigl(\tau, t^1, Q\bigr)&=t^1+3Q^{\frac{1}{3}}J^{\frac{1}{3}}(\tau)=-2\frac{ (2\pi )^2}{r^2} E_4(\tau) +2\frac{(2\pi)^2 }{r^2} \bigl(E_4^3(\tau)-E_6^2(\tau) \bigr)^{\frac{1}{3}}J^{\frac{1}{3}}(\tau)\\
&=2\frac{ (2\pi )^2}{r^2}[ -E_4(\tau) + E_4(\tau) ]=0.
\end{align*}
\end{Remark}

\begin{Remark}
The LG superpotential \smash{$\lambda=t^1+\frac{1}{4}Q^{\frac{1}{3}}\gamma_2(\tilde\tau)$} \eqref{lg superpotential of small quantum cohomology j function } can be thought of as being a~family of the covering maps over $\mathbb{CP}^1$ which share the same monodromy. This family is an isomonodromic deformation of the Hauptmodul of the group $\Gamma^{(3)}$, the parameter $Q$, $t^1$ are a~rescaling and affine isomonodromic deformation parameter, respectively.
\end{Remark}

\subsection{Cohn interpretation}\label{Cohn interpretation}

The aim of this subsection is to give a geometric interpretation of LG superpotential \eqref{lg superpotential of small quantum cohomology j function } in terms of elliptic curves with respect to the lattice \smash{$\mathbb{Z}\oplus {\rm e}^{\frac{\pi {\rm i}}{3}}\mathbb{Z}$}.

Consider a family of hyperelliptic curves whose roots are given by the small canonical coordinates \eqref{canonical coordinates of small quantum cohomology}
\begin{gather}
y^2=4(\lambda-u_1)(\lambda-u_2)(\lambda-u_3)=4\big[\bigl(\lambda-t^1\bigr)^3-27Q\big].\label{spectral curve small quantum cohomology}
\end{gather}

With respect to the notation of the equation \eqref{algebraic of Weierstrass p function}, the cubic \eqref{spectral curve small quantum cohomology} do not have the linear term, i.e.,
$
g_2=0$,
 which means that the correspondent lattice of the elliptic curve \eqref{spectral curve small quantum cohomology} is equiharmonic, i.e.,
\smash{$
\mathbb{Z}\oplus {\rm e}^{\frac{\pi {\rm i}}{3}}\mathbb{Z}$}.
Because, $g_2(\tau)$ is proportional to $E_4(\tau)$, whose zeros are the~${\rm SL}_2(\mathbb{Z})$ orbit of \smash{${\rm e}^{\frac{\pi {\rm i}}{3}}$}.
The uniformization of the cubic \eqref{spectral curve small quantum cohomology} is given by
\[
\lambda-t^1=\frac{\wp\bigl(v,{\rm e}^{\frac{\pi {\rm i}}{3}}\bigr)}{(2\omega)^2},\qquad
y=\frac{\wp^{\prime}\bigl(v,{\rm e}^{\frac{\pi {\rm i}}{3}}\bigr)}{(2\omega)^3},\qquad
27Q=-\frac{1}{4(2\omega)^6},
\]
which gives rise to a natural change of coordinates $C(v_0,\omega)=\bigl(t^1,Q\bigr)$
\begin{gather*}
C\colon\ \bigl(\mathbb{C}\setminus\bigl\{ \mathbb{Z}\oplus {\rm e}^{\frac{\pi {\rm i}}{3}}\mathbb{Z}\bigr\}\times\mathbb{C}^{*}\bigr)\setminus\left\{ \frac{\mathbb{Z}}{2}\oplus {\rm e}^{\frac{\pi {\rm i}}{3}}\frac{\mathbb{Z}}{2} \times\mathbb{C}^{*}\right\} \mapsto \mathbb{C}\times\mathbb{C}^{*}\setminus\bigl\{\bigl(t^1\bigr)^3+27Q=0 \bigr\},
\end{gather*}
given by
\begin{gather}\label{small quantum cohomology coordinates cubic}
t^1(v_0,\omega)=-\frac{\wp\bigl(v_0,{\rm e}^{\frac{\pi {\rm i} }{3}}\bigr)}{(2\omega)^2},\qquad
27Q(v_0,\omega)=-\frac{1}{4(2\omega)^6}.
\end{gather}
Indeed, the functions \eqref{small quantum cohomology coordinates cubic} are holomorphic and locally invertible in its domain. In addition, the discriminant \smash{$\bigl(t^1\bigr)^3+27Q=0$}, in this coordinates is written as
\begin{align*}
\bigl(t^1\bigr)^3+27Q&=-\frac{\wp^3\bigl(v_0,{\rm e}^{\frac{\pi {\rm i} }{3}}\bigr)}{(2\omega)^6}-\frac{1}{4(2\omega)^6}=-\frac{1}{4(2\omega)^6}\bigl(4\wp^3\bigl(v_0,{\rm e}^{\frac{\pi {\rm i} }{3}}\bigr)+1 \bigr)\\
&=-\frac{1}{4(2\omega)^6}\bigl( \wp^{\prime}\bigl(v_0,{\rm e}^{\frac{\pi {\rm i} }{3}}\bigr) \bigr)^2=0,
\end{align*}
which is equivalent to the locus \smash{$\frac{\mathbb{Z}}{2}\oplus {\rm e}^{\frac{\pi {\rm i}}{3}}\frac{\mathbb{Z}}{2} \times\mathbb{C}^{*}$} due to \eqref{critical values of wp}.

Then, the uniformization of the cubic \eqref{spectral curve small quantum cohomology} is the following family of $\wp$ functions:
\begin{gather}
\lambda(v,v_0,\omega)=\frac{\wp\bigl(v,{\rm e}^{\frac{\pi {\rm i} }{3}}\bigr)}{(2\omega)^2}-\frac{\wp\bigl(v_0,{\rm e}^{\frac{\pi {\rm i} }{3}}\bigr)}{(2\omega)^2},\label{lg superpotential Weierstrass p function small quantum cohomology}
\end{gather}
parametrized by $(v_0,\omega) \in \bigl(\mathbb{C}\setminus\bigl\{\mathbb{Z}\oplus {\rm e}^{\frac{\pi {\rm i}}{3}}\mathbb{Z} \bigr\}\times\mathbb{C}^{*}\bigr)\setminus\bigl\{ \frac{\mathbb{Z}}{2}\oplus {\rm e}^{\frac{\pi {\rm i}}{3}}\frac{\mathbb{Z}}{2} \bigr\}\times \mathbb{C}^{*}$.
One might initially assume that the family \eqref{lg superpotential Weierstrass p function small quantum cohomology} serves as an LG superpotential for the small quantum cohomology of~$\mathbb{CP}^2$ since it produces the correct canonical coordinates by construction. Nevertheless, it fails to meet the remaining conditions outlined in \eqref{Landau--Ginzburg superpotential definition}. The good news is that we can enrich the family~\eqref{lg superpotential Weierstrass p function small quantum cohomology} to a truly LG superpotential. For this purpose, we have to consider a suitable Abelian differential~$\phi$ which would satisfy the other conditions of \eqref{Landau--Ginzburg superpotential definition}. This Abelian differential~$\phi$ is constructed explicitly by considering the universal covering of \smash{$\mathbb{C}\setminus\bigl\{\mathbb{Z}\oplus {\rm e}^{\frac{\pi {\rm i}}{3}}\mathbb{Z}\bigr\}$}. Let us make a digression on factorisation of ${\rm SL}_2(\mathbb{Z})$ group. Recall that there exist group homomorphisms from ${\rm SL}_2(\mathbb{Z})$ to~${\rm SL}_2(\mathbb{Z}_n)$, i.e.,
\begin{gather}\label{group homomorphism}
\pi_n\colon\ {\rm SL}_2(\mathbb{Z})\mapsto {\rm SL}_2(\mathbb{Z}_n).
\end{gather}
The kernel of \eqref{group homomorphism} provide the congruence subgroups of ${\rm SL}_2(\mathbb{Z})$
\begin{gather*}
\Gamma(n):=\operatorname{Ker}\pi_n=\left\{ {\begin{pmatrix}
a& b \\
c&d
\end{pmatrix}} \in {\rm SL}_2(\mathbb{Z}) \mid {\begin{pmatrix}
a& b \\
c&d
\end{pmatrix}}\equiv{\begin{pmatrix}
1& 0 \\
0&1
\end{pmatrix}} \mod n \right\}.
\end{gather*}
We are particularly interested in the case $n=2$
\begin{gather}
\pi_2\colon\ {\rm SL}_2(\mathbb{Z})\mapsto S_3, \qquad {\rm SL}_2(\mathbb{Z})/\Gamma(2)\cong S_3.\label{group homomorphism particular case n23}
\end{gather}

The factorisation of the ${\rm SL}_2(\mathbb{Z})$ for $\pi_2$ in \eqref{group homomorphism particular case n23} induces a factorisation of the $j$-function~\eqref{J function other normalization} in the following sense:
\begin{gather*}
J\colon\ \mathbb{H}\mapsto \mathbb{C}\setminus\{0,1\}\mapsto \mathbb{C}, \qquad J(\tau)=\frac{4}{27}\frac{\bigl(1-x(\tau)+x^2(\tau)\bigr)^3}{(x(\tau))^2(x(\tau)-1 )^2},
\end{gather*}
where
\begin{gather}\label{modular lambda function 0}
x\colon\ \mathbb{H}\mapsto \mathbb{C}\setminus\{0,1\}, \qquad x(\tau)=\frac{\wp\bigl(\frac{1+\tau}{2},\tau\bigr)-\wp\bigl(\frac{\tau}{2},\tau\bigr)}{\wp\bigl(\frac{1}{2},\tau\bigr) -\wp\bigl(\frac{\tau}{2},\tau\bigr)},
\end{gather}
is the universal covering $ \mathbb{C}\setminus\{0,1\}$, which is a $\Gamma(2)$- invariant function, i.e.,
\begin{gather*}
x( \tau+2 )=x(\tau),\qquad
x\left( \frac{\tau}{2\tau+1} \right)=x(\tau),
\end{gather*}

 and have the following transformation law from under the
 $\operatorname{Aut}( \mathbb{C}\setminus\{0,1\})\cong S_3$ action:
\begin{gather*}
x(\tau+1)=\frac{x(\tau)}{x(\tau)-1}, \qquad x\left(\frac{-1}{\tau}\right)=1-x(\tau), \qquad x\left(\frac{\tau}{1-\tau}\right)=\frac{1}{x(\tau)},\\
x\left(\frac{1}{1-\tau}\right)=\frac{1}{1-x(\tau)}, \qquad x\left(\frac{\tau-1}{\tau}\right)=\frac{x(\tau)-1}{x(\tau)}.
\end{gather*}
The example above is our toy model to the small quantum cohomology case. Indeed, we provide similar results by replacing the universal covering of $\mathbb{C}\setminus\{0,1\}$ to the universal covering of~\smash{$\mathbb{C}\setminus\bigl\{\mathbb{Z}\oplus {\rm e}^{\frac{\pi {\rm i}}{3}}\mathbb{Z}\bigr\}$}. Here, we will summarise some results of \cite[Sections 1 and 2]{M.Batchelor} which is also related to \cite{Cohn}, and \cite{Golubev}. Consider the group of affine transformation on $\mathbb{C}$ which can be represented by
\begin{gather*}
Aff(\mathbb{C})=\left\{{\begin{pmatrix}
a& b \\
0&1 \\
\end{pmatrix}} \right\}\subset {\rm GL}(2,\mathbb{C}),
\end{gather*}
where the action in $\mathbb{C}$ is given by
\begin{gather*}
A(z)={\rm proj}_1{\begin{pmatrix}
a& b \\
0&1 \\
\end{pmatrix}} {\begin{pmatrix}
z \\
1 \\
\end{pmatrix}} =az+b, \qquad z\in\mathbb{C}, \qquad A \in \operatorname{Aff}(\mathbb{C}).
\end{gather*}
Define a group homomorphism
$
\psi\colon {\rm SL}_2(\mathbb{Z})\mapsto \operatorname{Aff}(\mathbb{C})
$
by
\begin{gather*}
S={\begin{pmatrix}
0& 1 \\
-1&0 \\
\end{pmatrix}} \mapsto \bar S={\begin{pmatrix}
-1& \zeta_6+\zeta_6^2 \\
0&1 \\
\end{pmatrix}},\\
T={\begin{pmatrix}
1& 1 \\
0&1 \\
\end{pmatrix}} \mapsto \bar T={\begin{pmatrix}
\zeta_6&0\\
0&1 \\
\end{pmatrix}},\qquad \zeta_6={\rm e}^{\frac{\pi {\rm i}}{3} }.
\end{gather*}
Then, due to the results of \cite[Section 1 and 2]{M.Batchelor}, the following facts hold true:
\begin{enumerate}\itemsep=0pt
\item[(1)] ${\rm SL}_2(\mathbb{Z})/\operatorname{Ker}\psi\cong \operatorname{Im}\psi$.
\item[(2)] The $\operatorname{Im}\psi$ acts on $\mathbb{C}$ and do preserve the lattice \smash{$\mathbb{Z}\oplus {\rm e}^{\frac{\pi {\rm i}}{3}}\mathbb{Z}$}.
\item[(3)] The $\operatorname{Im}\psi$ acts properly discontinuously on \smash{$\mathbb{C}\setminus\bigl\{\mathbb{Z}\oplus {\rm e}^{\frac{\pi {\rm i}}{3}}\mathbb{Z}\bigr\}$} and have the interior of the triangle with corners $0$, $\zeta_6$, $\zeta_6^2$ as fundamental chamber.
\item[(4)] \smash{$\mathbb{H}/\operatorname{Ker} \psi\cong\mathbb{C}\setminus\bigl\{\mathbb{Z}\oplus {\rm e}^{\frac{\pi {\rm i}}{3}}\mathbb{Z}\bigr\}$}, hence \smash{$\pi_1\bigl( \mathbb{C}\setminus\bigl\{\mathbb{Z}\oplus {\rm e}^{\frac{\pi {\rm i}}{3}}\mathbb{Z}\bigr\}\bigr)=\operatorname{Ker}\psi$}.
\item[(5)] The Weierstrass $\wp$ function with respect to the lattice \smash{$\mathbb{Z}\oplus {\rm e}^{\frac{\pi {\rm i}}{3}}\mathbb{Z}$} has the following transformation laws due to \eqref{Weierstrass p function transformation laws}:
 \begin{gather*}
\wp\bigl(v+m+n{\rm e}^{\frac{\pi {\rm i}}{3}},{\rm e}^{\frac{\pi {\rm i}}{3}}\bigr)=\wp\bigl(v,{\rm e}^{\frac{\pi {\rm i}}{3}}\bigr),\qquad m,n \in\mathbb{Z},\\
\wp\left(\frac{v}{-{\rm e}^{\frac{\pi {\rm i}}{3}}},\frac{-1}{{\rm e}^{\frac{\pi {\rm i}}{3}}}\right)=\bigl({\rm e}^{\frac{\pi {\rm i}}{3}}\bigr)^2\wp\bigl(v,{\rm e}^{\frac{\pi {\rm i}}{3}}\bigr).
\end{gather*}
Moreover, the $\operatorname{Im} \psi$ action induces the following:
 \begin{gather*}
\wp\bigl(-v+{\rm e}^{\frac{\pi {\rm i}}{3}}+{\rm e}^{\frac{2\pi {\rm i}}{3}},{\rm e}^{\frac{\pi {\rm i}}{3}}\bigr)=\wp\bigl(v,{\rm e}^{\frac{\pi {\rm i}}{3}}\bigr),\qquad
\wp\left(\frac{v}{-{\rm e}^{\frac{2\pi {\rm i}}{3}}},\frac{-1}{{\rm e}^{\frac{2\pi {\rm i}}{3}}}\right)=-{\rm e}^{\frac{\pi {\rm i}}{3}}\wp\bigl(v,{\rm e}^{\frac{\pi {\rm i}}{3}}\bigr).
\end{gather*}

Then, the function \smash{${\wp^{\prime}\bigl(v,{\rm e}^{\frac{\pi {\rm i}}{3}}\bigr)}^2$} is invariant with respect the $\operatorname{Im} \psi$ action.
\item[(6)] Let $F_{\operatorname{Im}\psi}$ and $F_0$ be the fundamental chamber of $\operatorname{Im}\psi$ and ${\rm SL}_2(\mathbb{Z})$, respectively. Then, the function
\smash{$
v\big|_{F_0}\colon F_0\mapsto F_{\operatorname{Im}\psi}
$}
given by the composition of the $j$-function and the inverse of~\smash{${\wp^{\prime}(v,{\rm e}^{\frac{\pi {\rm i}}{3}})}^2$}, i.e.,
\smash{$
v\big|_{F_0}(\tau)=\bigl(\bigl(\wp_{{\rm e}^{\frac{\pi {\rm i}}{3}}}^{\prime}\bigr)^2\bigr)^{-1}( J(\tau ) )
$}
is analytic and bijective.
\item[(7)]
The function $v\big|_{F_0}$ extends to an analytic covering map
\begin{gather}\label{universal covering map equiharmonic lattice 2}
v(\tau)\colon \ \mathbb{H}\mapsto \mathbb{C}\setminus\bigl\{ \mathbb{Z}\oplus {\rm e}^{\frac{\pi {\rm i}}{3}}\mathbb{Z}\bigr\} ,
\end{gather}
which is the universal covering map of $ \mathbb{C}\setminus\bigl\{ \mathbb{Z}\oplus {\rm e}^{\frac{\pi {\rm i}}{3}}\mathbb{Z}\bigr\} $.
\end{enumerate}

As a consequence, we have the Cohn identities derived in \cite[Section 9]{Cohn} (see also \cite{M.Batchelor,Golubev})
\begin{gather}
1-J(\tau)=4\wp^3\bigl(v(\tau),{\rm e}^{\frac{\pi {\rm i}}{3}}\bigr)+1,\qquad
(v^{\prime}(\tau))^6=\Delta(\tau).\label{universal covering map equiharmonic lattice}
\end{gather}

Here, we state and prove the main result of this subsection.

\begin{Theorem}
The Landau--Ginzburg superpotential of small affine quantum cohomology of $\mathbb{CP}^2$ is a family of functions
$
\lambda(\tilde\tau, v_0, \omega)\colon \mathbb{H} \mapsto \mathbb{C}
$
with holomorphic dependence in the parameter space
\smash{$
(v_0,\omega) \in \bigl(\mathbb{C}\setminus\bigl\{\mathbb{Z}\oplus {\rm e}^{\frac{\pi {\rm i}}{3}}\mathbb{Z} \bigr\}\bigr)\times\mathbb{C}^{*}
$}
 and given by
\begin{gather*}
\lambda(\tilde\tau,v_0,\omega)=\frac{\wp\bigl(v(\tilde\tau),{\rm e}^{\frac{\pi {\rm i} }{3}}\bigr)}{(2\omega)^2}-\frac{\wp\bigl(v_0,{\rm e}^{\frac{\pi {\rm i} }{3}}\bigr)}{(2\omega)^2},
\end{gather*}
where $v(\tilde\tau)$ is the universal covering of \smash{$\mathbb{C}\setminus\bigl\{\mathbb{Z}\oplus {\rm e}^{\frac{\pi {\rm i}}{3}}\mathbb{Z} \bigr\}$}.
 In addition, the correspondent Abelian differential $\phi$ is given by
$
\phi=2\omega {\rm d}v(\tilde\tau)$.

\end{Theorem}

\begin{proof}

 Substituting \eqref{small quantum cohomology coordinates cubic} and \eqref{universal covering map equiharmonic lattice} in \eqref{lg superpotential Weierstrass p function small quantum cohomology}, we obtain the LG superpotential \eqref{lg superpotential of small quantum cohomology j function }. Indeed,
 \begin{align*}
\lambda(\tilde\tau,v_0,\omega)&=\frac{\wp\bigl(v(\tilde\tau),{\rm e}^{\frac{\pi {\rm i} }{3}}\bigr)}{(2\omega)^2}-\frac{\wp\bigl(v_0,{\rm e}^{\frac{\pi {\rm i} }{3}}\bigr)}{(2\omega)^2}=-3Q^{\frac{1}{3}}4^{\frac{1}{3}}\wp\bigl(v(\tilde\tau),{\rm e}^{\frac{\pi {\rm i} }{3}}\bigr)+t^1
 =3Q^{\frac{1}{3}}J^{\frac{1}{3}}(\tilde\tau)+t^1.
 \end{align*}
 Substituting the second equation of \eqref{universal covering map equiharmonic lattice} and \eqref{small quantum cohomology coordinates cubic} in \eqref{small volume form 1}, the Abelian differential $\phi$ becomes
 \begin{align*}
 \phi=-\frac{2^{\frac{5}{2}}}{2\pi}\frac{\Delta^{\frac{1}{6}}(\tilde\tau)}{Q^{\frac{1}{6}}}{\rm d}\tilde \tau=-\frac{2^{\frac{17}{6}}\times 3^{\frac{1}{2}}}{2\pi} (2\omega)\Delta^{\frac{1}{6}}(\tilde\tau){\rm d}\tilde \tau=-\frac{2^{\frac{17}{6}}\times 3^{\frac{1}{2}}}{2\pi} (2\omega)v^{\prime}(\tilde\tau){\rm d}\tilde \tau.
 \end{align*}
 Summarising, up a constant, we have
$
 \phi=2\omega {\rm d}v(\tilde\tau)=2\omega v^{\prime}(\tilde\tau) {\rm d}\tilde\tau$,
as a consequence
 \begin{gather*}
 \phi={\rm d}\tilde w_2=2\omega {\rm d}v(\tilde\tau) \implies \tilde w_2=2\omega v(\tilde\tau).
 \end{gather*}
In other words, the Abelian differential $\phi$ of affine small quantum cohomology of $\mathbb{CP}^2$ is the differential of the universal covering \eqref{universal covering map equiharmonic lattice 2}.
\end{proof}

\begin{Remark}
The composition of universal covering \eqref{universal covering map equiharmonic lattice 2} with the functions $\wp$, $\wp^{\prime}$ with respect to the equianharmonic lattice can be expressed in terms of the Weber functions, i.e.,
\begin{gather*}
\gamma_2(\tau)=3\times4^{\frac{4}{3}}\wp\bigl(v(\tau),{\rm e}^{\frac{\pi {\rm i}}{3}}\bigr),\qquad
\gamma_3(\tau)=3^{\frac{3}{2}} \times 4^{\frac{3}{2}}\wp^{\prime}\bigl(v(\tau),{\rm e}^{\frac{\pi {\rm i}}{3}}\bigr).
\end{gather*}
\end{Remark}

\begin{Remark}
In the Givental setting \cite{Givental}, the elliptic curve \eqref{spectral curve small quantum cohomology}, up affine shift, is precisely the $0$-fiber of the Givental 2D superpotential of small quantum cohomology of $\mathbb{CP}^2$. Indeed, the~$0$ fiber of Givental superpotential
$
\lambda(x,y,Q)=x+y+\frac{Q}{xy}$,
is the cubic
\begin{gather}\label{Givental superpotential cubic zero fiber}
0=x^2y+xy^2+Q.
\end{gather}
Substituting $\tilde y=\sqrt{x}y$ in \eqref{Givental superpotential cubic zero fiber}, we have that
\smash{$
\tilde y^2+x^{\frac{3}{2}}\tilde y+Q=0$}.
 Defining \smash{$\hat y=\tilde y+\frac{x^{\frac{3}{2}}}{2}$} and completing square, we obtain
\smash{$
\hat y^2=\frac{x^3}{4}-Q$}.
After the rescaling \smash{$\hat y\mapsto \frac{\hat y}{2\sqrt{27}}$}, \smash{$x\mapsto \frac{2^{\frac{2}{3}}}{3}x$}, we obtain
$
\hat y^2=4 x^3-(4\times 27)Q$.
\end{Remark}

\section[Big quantum cohomology of $\mathbb{CP}^2$]{Big quantum cohomology of $\boldsymbol{\mathbb{CP}^2}$}\label{Big Quantum Cohomology CP2}

\subsection{Milanov deformation}\label{Milanov deformation}

 The aim of this section is to derive a $t^3$-deformation of the superpotential \eqref{lg superpotential of small quantum cohomology j function }.

In \cite[Section 5]{Milanov}, Milanov considered the $t^3$-deformation of abstract periods of small quantum cohomology in the following sense. Consider the power series
\begin{gather}\label{t3 deformation of periods}
w_{i}\bigl(t^1,t^2,t^3\bigr)=w_i\bigl(t^1,Q,0\bigr)+\sum_{n=1}^{\infty}w_{i,n}\bigl(t^1,Q\bigr) \bigl( t^3\bigr)^n, \qquad i=1,2,3,
\end{gather}
which solves the abstract Gauss--Manin connection flat coordinate system \eqref{Gauss--Manin connection} of quantum cohomology of $\mathbb{CP}^2$.
The solution \eqref{t3 deformation of periods} is a multivalued map that is locally analytic in the locus
\begin{gather}
D\setminus{\Sigma}=\left\{ \bigl(t^1,Q,t^3\bigr) \in \mathbb{C}\times \mathbb{C}^{*}\times \mathbb{C} \mid \big|Q\bigl(t^3\bigr)^3\big|<\frac{1}{a},\, \det\begin{pmatrix}
3F_{33}& 2F_{23}&t^1 \\
2F_{23}& F_{22}& 3\\
t^1& 3 & -t^3
\end{pmatrix} \neq 0 \right\},\label{period domain of Quantum cohomology 1}
\end{gather}
because the coefficients of \eqref{Gauss--Manin connection} with respect $QH^{*}\bigl(\mathbb{CP}^2\bigr)$ are analytic in the domain \eqref{WDVV domain}, excluding the discriminant locus \eqref{discriminant}.

In \cite[Section 5.1]{Milanov}, Milanov constructed an auxiliary period domain, which extends the map~\eqref{inverse period map of small quantum cohomology} for $t^3\neq 0$ aiming to locally invert the map \eqref{t3 deformation of periods}. More specifically, consider the map
\begin{gather}\label{pi auxiliary quantum cohomology}
\pi^{\rm aux}\colon\ \mathcal{U}\mapsto D\setminus{\Sigma}\subset\mathbb{C}\times\mathbb{C}^{*}\times \mathbb{C} ,
\end{gather}
given by
\begin{gather}
t^1= -2\frac{(2\pi)^2}{r^2}E_4(\tau),\qquad
Q= \frac{8}{27}\frac{(2\pi)^6}{r^6}\bigl(E_4^3(\tau)-E_6^2(\tau)\bigr),\qquad
t^3= s\frac{E_6^2(\tau)}{r^6},\label{pi auxiliary concrete}
\end{gather}
where $D\setminus{\Sigma}$ is defined in \eqref{period domain of Quantum cohomology 1} and
\begin{gather*}
\mathcal{U}=\{ (\tau,r,s)\in \mathbb{H}\times\mathbb{C}^{*}\times \mathbb{C}\mid |s|<\delta(\tau,r)\},
\end{gather*}
the function
$
\delta(\tau,r)\colon\mathbb{H}\times\mathbb{C}^{*}\mapsto \mathbb{R}_{>0}
$
 is chosen in such way that the preimage of the discriminant by the map \eqref{pi auxiliary quantum cohomology} is the zero locus of $E_6(\tau)$. More specifically, we express the discriminant as power series in $t^3$ as follows:
\begin{gather}\label{discriminant perturbation 1}
\det E_{\bullet}=\bigl(t^1\bigr)^3+27Q+t^3f\bigl(t^1,Q,t^3\bigr)=0
\end{gather}
for some $f\bigl(t^1,Q,t^3\bigr)$ holomorphic function in $D\setminus{\Sigma}$.
Substituting the change of coordinates \eqref{pi auxiliary concrete} in \eqref{discriminant perturbation 1}, we obtain
\begin{gather*}
E_6^2(\tau)\bigl(8(2\pi)^2 + sf\bullet\pi^{\rm aux}(\tau, r, s) \bigr)=0.
\end{gather*}
For fixed $(\tau,r)\in \mathbb{H}\times \mathbb{C}^{*}$ the function $\delta(\tau,r)$ is chosen such that $\pi^{\rm aux}(\tau, r, s) \in D$ and $|sf\bullet\pi^{\rm aux}(\tau, r, s)|<8(2\pi)^2$.

The domain of \eqref{pi auxiliary quantum cohomology} is a tubular neighbourhood of the domain of \eqref{inverse period map of small quantum cohomology}, which keeps the group of Deck transformations of \eqref{inverse period map of small quantum cohomology} constant. More precisely, the monodromy action \eqref{Monodromy action affine small quantum cohomology} extends trivially on the direction $s$, i.e.,
\[
A(\tau,r,s)=\left( \frac{a\tau+b}{c\tau+d} , r(c\tau+d)^2,s\right),\qquad
B(\tau,r,s)=( \tau , -r,s).
\]

Composing \eqref{pi auxiliary concrete} with \eqref{t3 deformation of periods}, we obtain a power series in \(t^3\)
\begin{gather}\label{t3 deformation of periods 2}
w_{i}\bigl(\tau,r,t^3\bigr) = w_i(\tau,r,0) + \sum_{n=1}^{\infty} w_{i,n}(\tau,r) \bigl( t^3 \bigr)^n, \qquad i=1,2,3,
\end{gather}
where the coefficients are functions
\begin{gather*}
w_{i,n}\colon\ \{(\tau,r) \in \mathbb{H} \times \mathbb{C}^{*} \mid E_6(\tau) \neq 0\} \mapsto \mathbb{C}.
\end{gather*}

Then, we can state in our setting the following \cite[Proposition 5.8]{Milanov}.

\begin{Lemma}[{\cite{Milanov}}]\label{lemma Milanov deformation 1}
Let the power series $w_1$, $w_2$, $w_3$ be defined in \eqref{t3 deformation of periods 2}. Then the power series~$w_1$,~$w_2$,~$w_3$ define holomorphic functions in the domain \eqref{pi auxiliary quantum cohomology}.
 Moreover, its Taylor coefficients $(w_{1,n}(\tau,r),w_{2,n}(\tau,r),w_{3,n}(\tau,r))$ have the following property:
\[
(w_{1,n}(\tau,r),w_{2,n}(\tau,r),w_{3,n}(\tau,r)) \in r^{1-2n}E_6^{-2n}\mathbb{C}[ E_2,E_4,E_6], \qquad n>0 .
\]
\end{Lemma}
The sketch of the proof of Lemma \ref{lemma Milanov deformation 1} is given by the following the key points:
\begin{enumerate}\itemsep=0pt

\item[(1)]
The power series defined in \eqref{t3 deformation of periods 2} are holomorphic in the domain \eqref{pi auxiliary quantum cohomology} since they are analytic continuation of the multivalued functions \eqref{t3 deformation of periods} due to \cite[Proposition 5.3]{Milanov}.

\item[(2)]
Due to the quasi homogeneous condition \eqref{quasi homogeneous condition of the solutions xa and xalambda} and the charge $d$ of $QH^{*}\mathbb{CP}^2$ be equal~$2$, the coordinates $w_1$, $w_2$, $w_3$ have degree $-\frac{1}{2}$ and the coordinates $r$, $\tau$, $t^3$ have degree~$-\frac{1}{2}$,~$0$,~$-1$, respectively. Then $w_{i,n}(\tau,r)$ is proportional to $r^{1-2n}$, i.e., $w_{i,n}(\tau,r)$ factorise as follows
\begin{gather*}
w_{i,n}(\tau,r)=r^{1-2n}w_{i,n}(\tau).
\end{gather*}

\item[(3)]
Under the change of coordinates \eqref{pi auxiliary concrete} and due to the Ramanujan identities \eqref{Ramanujan identities},
the vector fields $t^1\partial_1$, $Q\partial_Q$ have the following form:
\begin{gather}
t^1\partial_1=\frac{E_4}{E_6}\partial_{\tau}+\frac{1}{6E_6}(E_2E_4)r\partial_{r},\qquad
Q\partial_{Q}=\frac{E_4}{E_6}\partial_{\tau}+\frac{1}{6E_6}(E_2E_4-E_6 )r\partial_{r}.\label{change of vector fields new coordinates from t to tau 1}
\end{gather}
The action of vector field $\partial_3$ can also be written in terms of the vector fields $E_2$, $E_4$, $E_6$ due to the vector field \eqref{change of vector fields new coordinates from t to tau 1} and the quasi homogeneous condition \eqref{quasi homogeneous condition of the solutions xa and xalambda}. The coefficients of Gauss--Manin connection of $QH^{*}\mathbb{CP}^2$ are polynomial in $t^1$, $t^3$, $\frac{1}{t^3}$ and $\Phi(X)$, $\Phi^{\prime}(X)$, $\Phi^{\prime\prime}(X)$, $\Phi^{\prime\prime\prime}(x)$ where $\Phi(X)$ is defined in \eqref{main Gromov Witten potential}. The function $\Phi(X)$ is a holomorphic power series around \smash{$X:=\ln \bigl(Q\bigl(t^3\bigr)^3\bigr)\mapsto -\infty$}. Hence, the Gauss--Manin connection of $QH^{*}\mathbb{CP}^2$ gives rise to a infinite list of differential equation for $w_{i,n}(\tau)$ with rational coefficients in $E_2$,~$E_4$,~$E_6$. Since, the ring of quasi-modular forms $\mathbb{C}[E_2,E_4,E_6]$ is closed under derivation, the functions $w_{i,n}(\tau)$ must be rational in $E_2$, $E_4$, $E_6$.

\item[(4)] The function $w_{i,n}(\tau)\in E_6^{-2n}\mathbb{C}[E_2,E_4,E_6]$ , because of its holomorphic behaviour in \eqref{pi auxiliary quantum cohomology}, i.e., the series
\begin{gather*}
w_{i}\bigl(\tau,r,t^3\bigr)=\sum_{n=0}^{\infty}w_{i,n}(\tau,r){\left(\frac{E_6^2(\tau)}{r^6}\right)}^{n} s^n, \qquad i=1,2,3,
\end{gather*}
is holomorphic in the domain \eqref{pi auxiliary concrete}. For more details, see \cite[Proposition 5.8]{Milanov}.

\end{enumerate}

For a deeper understanding of the correspondence between $(\tau,r,s)$ and $(w_1,w_2,w_3)$, it becomes essential to construct an auxiliary coordinate system $(\tau_1,\tau_2,y)$. To achieve this, it is crucial to recall that the intersection form resulting from $S+S^{T}$ in specific flat coordinates is described by an indefinite bilinear form
\begin{gather*}
\begin{pmatrix}
{\rm d}x& {\rm d}y&{\rm d}z \\
\end{pmatrix}\begin{pmatrix}
2& 3&-3 \\
3& 2& -3\\
-3& -3 &2
\end{pmatrix}\begin{pmatrix}
{\rm d}x \\
{\rm d}y\\
{\rm d}z
\end{pmatrix}.
\end{gather*}
Then, through a linear change of coordinates
\begin{gather*}
\begin{pmatrix}
{\rm d}x \\
{\rm d}y\\
{\rm d}z
\end{pmatrix}=\begin{pmatrix}
2& -3&2 \\
2& 1& -2\\
2& -1 & -2
\end{pmatrix}\begin{pmatrix}
{\rm d}w_1 \\
{\rm d}w_2\\
{\rm d}w_3
\end{pmatrix},
\end{gather*}
and after rescaling, the intersection form becomes
\begin{gather}\label{intersection form in flat coordinates}
g^{*}={\rm d}w_2^2-4{\rm d}w_1{\rm d}w_3.
\end{gather}
Due to change of endomorphism \eqref{change of basis of endomorphism}, the monodromy action on the coordinates $(w_1,w_2,w_3)\in {(\mathbb{C}^{*})}^3$ is given by the generators \eqref{new basis}.

Consider the map, which works as a local change of coordinates
\begin{gather}
\phi\colon\ \mathbb{H}^2\times \mathbb{C}^{*}\mapsto {(\mathbb{C}^{*})}^3, \qquad (\tau_1,\tau_2,y)\mapsto (w_1,w_2,w_3) =(\tau_1\tau_2y,(\tau_1+\tau_2)y,y).\label{covering map symmetric square big quantum cohomology}
\end{gather}
Due to symmetric square representation \eqref{symmetric square representation}, the monodromy action in $\mathbb{H}^2\times \mathbb{C}^{*}$ is given by
\begin{gather}
A(\tau_1,\tau_2,y)=\left( \frac{a\tau_1+b}{c\tau_1+d}, \frac{a\tau_2+b}{c\tau_2+d}, y(c\tau_1+d)(c\tau_1+d)\right),\nonumber\\
B(\tau_1,\tau_2,y)=(\tau_1,\tau_2,-y).\label{diagonal action of PSL2Z times a1}
\end{gather}

The change of coordinates of the intersection form \eqref{intersection form in flat coordinates} from its flat coordinates $(w_1,w_2,w_3)$ to Saito flat coordinates $\bigl(t^1,t^2,t^3\bigr)$ give rise to a relationship between both coordinates. In particular, we have \cite[Lemma 5.9]{Milanov}, which we state here for the convenience of the reader.

\begin{Lemma}[{\cite{Milanov}}]\label{lemma quadratic relations Milanov}
Let $\bigl(t^1,t^2,t^3\bigr)$ and $(w_1,w_2,w_3)$ be the Saito flat coordinates and intersection form flat coordinates of $QH^{*}\bigl(\mathbb{CP}^2\bigr)$, respectively. Then
$
t^3=w_2^2-4w_1w_3$.
Moreover, due to the map \eqref{covering map symmetric square big quantum cohomology}, the Saito flat coordinate $t^3$ is described as
\begin{gather}\label{t3 change of coordinate 2}
 t^3=y^2(\tau_1-\tau_2 )^2
\end{gather}
in the coordinates $(\tau_1,\tau_2,y)$.
\end{Lemma}

\begin{Remark}
The Lemma \ref{lemma quadratic relations Milanov} can be understood as consequence of \eqref{quadratic relation between Saito and Intersection form flat coordinates}.

\end{Remark}

\begin{Corollary}[{\cite{Milanov}}]\label{lemma Milanov deformation corollary}
Let the power series $w_1$, $w_2$, $w_3$ be defined in \eqref{t3 deformation of periods 2}. Then the map
\begin{gather}\label{change of coordinate tau,r,s to tau1,tau2,t3}
 F\colon\ \mathcal{U}\mapsto \mathbb{H}\times \mathbb{C}^{*}\times \mathbb{C}, \qquad (\tau,r,s)\mapsto \bigl(\tau_1+\tau_2,y, t^3\bigr)
 \end{gather}
given by the power series
\begin{gather}
\frac{1}{2}(\tau_1+\tau_2)(\tau,r,s)=\frac{-w_2}{2w_3}\bigl(\tau,r,t^3\bigr)=\tau+\sum_{n=1}^{\infty} \tau_{12,n}(\tau) \bigl(t^3r^{-2}\bigr)^n ,\nonumber\\
y(\tau,r,s)=w_3\bigl(\tau,r,t^3\bigr)=r\left( 1+\sum_{n=1}^{\infty} r_n(\tau)\bigl(t^3r^{-2}\bigr)^n \right),\qquad
t^3(\tau,r,s)=s\frac{E_6^2(\tau)}{r^6}\label{t3 deformation of periods in lemma}
\end{gather}
 is holomorphic in the domain \eqref{pi auxiliary quantum cohomology}.
 Moreover,
\begin{gather}\label{coefficients of tau12 inside corollary Milanov}
 \tau_{12,n}(\tau) ,r_n(\tau) \in E_6^{-2n}\mathbb{C}[ E_2,E_4,E_6] .
\end{gather}
\end{Corollary}

At this state, we can sketch the derivation of the inverse period map of quantum cohomology of $\mathbb{CP}^2$ done in \cite[Section 5.4]{Milanov}. Indeed, inverting the second power series \eqref{t3 deformation of periods in lemma} in $r$, we obtain
\begin{gather}
r=y\left( 1+\sum_{n=1}^{\infty} \hat r_n(\tau)\bigl(t^3y^{-2}\bigr)^n \right).\label{inversion power series 1}
\end{gather}
Substituting \eqref{inversion power series 1} in the first equation of \eqref{t3 deformation of periods in lemma},
\begin{gather}\label{inversion power series 2}
\tau_{12}:=\frac{\tau_1+\tau_2}{2}=\tau+\sum_{n=1}^{\infty} \hat\tau_{12,n}(\tau) \bigl(t^3y^{-2}\bigr)^n .
\end{gather}
Inverting \eqref{inversion power series 1}, \eqref{inversion power series 2} in $\tau$ and using equation \eqref{t3 change of coordinate 2}, we obtain the following:
\begin{gather}
\tau(\tau_1,\tau_2,y)=\tau_{12}+\sum_{n=1}^{\infty} \tau_{n}(\tau_{12}) (\tau_1-\tau_2)^{2n} ,\nonumber\\
r(\tau_1,\tau_2,y)=y\left( 1+\sum_{n=1}^{\infty} y_n(\tau_{12})(\tau_1-\tau_2)^{2n} \right),\qquad
t^3(\tau_1,\tau_2,y)=y^2(\tau_1-\tau_2)^2.\label{t3 deformation of periods to compose}
\end{gather}

 Composing \eqref{t3 deformation of periods to compose} with \eqref{pi auxiliary concrete},
\begin{gather*}
t^1=-2 (2\pi )^2\frac{E_4\bigl(\tau_{12}+\sum_{n=1}^{\infty} \tau_{n}(\tau_{12}) (\tau_1-\tau_2)^{2n} \bigr)}{y^2\bigl( 1+\sum_{n=1}^{\infty} y_n(\tau_{12})(\tau_1-\tau_2)^{2n} \bigr)^2},\\
Q=\frac{8}{27}(2\pi)^6 \frac{\bigl(E_4^3-E_6^2\bigr)\bigl(\tau_{12}+\sum_{n=1}^{\infty} \tau_{n}(\tau_{12}) (\tau_1-\tau_2)^{2n} \bigr)}{y^6\bigl( 1+\sum_{n=1}^{\infty} y_n(\tau_{12})(\tau_1-\tau_2)^{2n} \bigr)^6},\qquad
t^3=y^2(\tau_1-\tau_2)^2.
\end{gather*}

 Milanov describes the inverse period map in \cite[Theorem 2.4]{Milanov}, which we state here as follows.

\begin{Theorem}[{\cite{Milanov}}]
Let $\mathcal{D}$ be the image of the following map:
\begin{gather*}
 t\bullet F\colon\ \mathcal{U}\mapsto F(\mathcal{U})\mapsto \mathcal{D}\subset \mathbb{H}\times\mathbb{C}^{*}\times\mathbb{C}
\end{gather*}
given by
\begin{gather*}
(\tau,r,s)\mapsto \bigl(\tau_1+\tau_2,y,t^3\bigr)\mapsto \left(\tau_1+\tau_2,y,\frac{t^3}{y^2}=(\tau_1-\tau_2)^2\right).
\end{gather*}
Then,
\begin{enumerate}\itemsep=0pt
\item[$(1)$]
The inverse period map of big quantum cohomology of $\mathbb{CP}^2$ is a holomorphic map $t\colon \mathcal{D}\subset \mathbb{H}^2\times \mathbb{C}^{*} \mapsto D$
given by
\begin{gather}
t^1(\tau_1,\tau_2,y)=-2\frac{ (2\pi )^2}{y^2}\sum_{n=0}^{\infty} t^1_n(\tau_{12})(\tau_1-\tau_2)^{2n} ,\nonumber\\
Q(\tau_1,\tau_2,y)=\frac{8}{27}\frac{(2\pi)^6 }{y^6}\sum_{n=0}^{\infty} Q_n(\tau_{12})(\tau_1-\tau_2)^{2n},\qquad
t^3=y^2(\tau_1-\tau_2),\label{inverse period map of Big quantum cohomology}
\end{gather}
where $\tau_{12}=\frac{\tau_1+\tau_2}{2}$.
\item[$(2)$]
The coefficients $t^1_n(\tau_{12}), Q_n(\tau_{12})$ are quasi-modular forms, i.e.,
$
t^1_n(\tau_{12}), Q_n(\tau_{12})\in \mathbb{C}[E_2,\allowbreak E_4,E_6]$.
\item[$(3)$] The inverse period map \eqref{inverse period map of Big quantum cohomology} is $A_1\times {\rm PSL}_2(\mathbb{Z})$-invariant with respect to the action \eqref{diagonal action of PSL2Z times a1}.

\end{enumerate}

\end{Theorem}

The first coefficients $t^1_n(\tau_{12})$, $Q_n(\tau_{12})$ are given by
\begin{gather*}
t^1_0(\tau_{12})=E_4(\tau_{12}),\qquad
t^1_1(\tau_{12})=\frac{1}{40}\partial^2_{\tau_{12}}E_4(\tau_{12}),\\
t^1_2(\tau_{12})=\frac{1}{4480}\partial_{\tau_{12}}^4E_4(\tau_{12})-\frac{\pi^4}{2016}Q_0(\tau_{12}),
\end{gather*}
and
\begin{gather*}
Q_0(\tau_{12})=E_4^3(\tau_{12})-E_6^2(\tau_{12}),\qquad
Q_1(\tau_{12})=\frac{1}{140}\partial^2_{\tau_{12}}Q_0(\tau_{12})+\frac{\pi^2}{26}E_4(\tau_{12})Q_0(\tau_{12}),\\
Q_2(\tau_{12})=\frac{1}{24960}\partial^4_{\tau_{12}}Q_0(\tau_{12})+\frac{\pi^2}{2704}E_4(\tau_{12})\partial^2_{\tau_{12}}Q_0(\tau_{12})\\
\phantom{Q_2(\tau_{12})=}{}+\frac{\pi^2}{1040}Q_0(\tau_{12})\partial^2_{\tau_{12}}E_4(\tau_{12})+\frac{17\pi^4}{20280}E_4^2(\tau_{12})Q_0(\tau_{12}).
\end{gather*}

\begin{Remark}\label{largest image of the inverse period map Gauss Manin connection}
The multivalued map $w_{i}\bigl(t^1,Q,t^3\bigr)$ \eqref{t3 deformation of periods} can be locally invertible away from the discriminant locus. Then, the inverse period
map is given by
\[
t\colon\ \hat{\mathcal{D}}\mapsto \left\{ \bigl(t^1,Q,t^3\bigr)\in \mathbb{C}\times\mathbb{C}^{*}\times\mathbb{C} \mid \big|Q\bigl(t^3\bigr)^3\big|<\frac{1}{a},\, \det g\neq 0 \right\}.
\]
 Furthermore, using proof of \cite[Proposition 5.3]{Milanov}, i.e., using the fact that the inverse period map is bounded at the discriminant and Riemann extension theorem, we can extends analytically the inverse period map over the discriminant. Therefore, the largest image of the inverse period map is the set
\[
 \left\{ \bigl(t^1,Q,t^3\bigr)\in \mathbb{C}\times\mathbb{C}^{*}\times\mathbb{C} \mid \big|Q\bigl(t^3\bigr)^3\big|<\frac{1}{a} \right\}.
\]
\end{Remark}

\begin{Remark}
Even though, the flat coordinates with respect to the intersection form of \linebreak ${QH^{*}\allowbreak\times\bigl(\mathbb{CP}^2\bigr)}$ is given by $(w_1,w_2,w_3)$, it is convenient to write the inverse period map in terms of~$(\tau_1,\tau_2,y)$, which is related to $(w_1,w_2,w_3)$ via the map \eqref{covering map symmetric square big quantum cohomology}.
\end{Remark}

Our next goal is to derive a $t^3$ deformation of the Landau--Ginzburg superpotential for affine small quantum cohomology, which is represented by the Weber function $\gamma_2(\tau)$ defined in \eqref{Weber functions}. The $t^3$ deformation involves composing the Weber function $\gamma_2(\tau)$ with the local inverses of $t^3$-deformed periods \eqref{t3 deformation of periods in lemma}.

\begin{Lemma}\label{lemma Milanov deformation 1 on LG superpotential}
Let the power series $w_1$, $w_2$, $w_3$ be defined in \eqref{t3 deformation of periods 2}.
Then, the composition of the functions
\begin{gather}
\tilde \tau_{12}:=-\frac{w_2}{2w_3}\bigl(t^1-\lambda,Q,t^3\bigr),\qquad
\tilde y:=w_3\bigl(t^1-\lambda,Q,t^3\bigr),\label{affine extended w2 w3 Big lemma}
\end{gather}
with the change of coordinates
\begin{gather}
t^1= -2\frac{(2\pi)^2}{r^2}E_4(\tau),\qquad Q= \frac{8}{27}\frac{(2\pi)^6}{r^6}\bigl(E_4^3(\tau)-E_6^2(\tau)\bigr),\qquad t^3= s\frac{E_6^2(\tau)}{r^6},\nonumber\\
\lambda= t^1+2\frac{(2\pi)^2}{\tilde r^2}E_4(\tilde\tau), \qquad \tilde r=r\frac{\bigl(E_4^3(\tilde\tau)-E_6^2(\tilde\tau)\bigr)^{\frac{1}{6}}}{\bigl(E_4^3(\tau)-E_6^2(\tau)\bigr)^{\frac{1}{6}}},\qquad \tilde s=\frac{\tilde r^6t^3}{E_6^2(\tilde\tau)},\label{change of extended coordinates flat intersection and auxiliary}
\end{gather}
are holomorphic functions
\begin{gather*}
\tilde\tau_{12}\colon \ \tilde{\mathcal{U}}=\bigl\{ \bigl(\tilde\tau,Q^{\frac{1}{3}}t^3\bigr)\in\mathbb{H}\times \mathbb{C}\mid \big|Q^{\frac{1}{3}}t^3\big|<\epsilon,\, E_6(\tilde\tau)\neq 0 \bigr\}\mapsto \mathbb{H},\\
\tilde r\colon \ \tilde{\mathcal{U}}=\bigl\{ \bigl(\tilde\tau,Q^{\frac{1}{3}}t^3\bigr)\in\mathbb{H}\times \mathbb{C}\mid \big|Q^{\frac{1}{3}}t^3\big|<\epsilon, \, E_6(\tilde\tau)\neq 0 \bigr\}\mapsto \mathbb{C},
\end{gather*}
given by the power series
\begin{gather}
\tilde\tau_{12}=\tilde\tau+\sum_{n=1}^{\infty}\frac{\tau_{12,n}(\tilde\tau)}{\Delta^{\frac{n}{3}}(\tilde\tau)} \bigl( Q^{\frac{1}{3}}t^3\bigr)^n,\qquad
\tilde y=\tilde r\left( 1+\sum_{n=1}^{\infty} \frac{r_n(\tilde\tau)}{\Delta^{\frac{n}{3}}(\tilde\tau)} \bigl(Q^{\frac{1}{3}} t^3 \bigr)^n \right),\label{affine in lambda t3 deformation of periods new coordinates z2 over z3 lemma}
\end{gather}
where $\epsilon$ has the following properties:
\begin{gather*}
\epsilon\leq \left(\frac{1}{a}\right)^{\frac{1}{3}} ,\qquad
\big|Q^{\frac{1}{3}}t^3\big|<\epsilon \implies \det g=0 \iff E_6(\tilde\tau)=0,
\end{gather*}
 and $\tau_{12,n}(\tilde\tau)$, $r_n(\tilde\tau)$ are given in \eqref{change of coordinate tau,r,s to tau1,tau2,t3}.

Moreover, for fixed \smash{$Q^{\frac{1}{3}}t^3$}, the function
$
\tilde\tau\colon \mathbb{H}\setminus{ {\rm SL}_2(\mathbb{Z})(i) } \mapsto \mathbb{H}\setminus{{\rm SL}_2(\mathbb{Z})(i)}
$
given by the inverse function of \eqref{affine in lambda t3 deformation of periods new coordinates z2 over z3 lemma},
\begin{gather}\label{affine in lambda t3 deformation of periods new coordinates z2 over z3 inverse lemma}
\tilde\tau=\tilde\tau_{12}+\sum_{n=1}^{\infty}\frac{\tilde\tau_{12,n}(\tilde\tau_{12})}{\Delta^{\frac{n}{3}}(\tilde\tau_{12})} \bigl( Q^{\frac{1}{3}}t^3\bigr)^n
\end{gather}
is a holomorphic function.

 \end{Lemma}

\begin{proof}

Consider the $t^3$-deformed affine extended variables
\begin{gather}\label{affine extended w2 w3 Big}
\tilde \tau_{12}:=\frac{w_2}{2w_3}\bigl(t^1-\lambda,Q,t^3\bigr),\qquad
\tilde y:=w_3\bigl(t^1-\lambda,Q,t^3\bigr).
\end{gather}
Then, the composition of \eqref{affine extended w2 w3 Big} with \eqref{change of extended coordinates flat intersection and auxiliary} is equivalent to replace
 $\tau$ to $\tilde\tau$ in the first power series of \eqref{t3 deformation of periods in lemma}. More precisely,
\begin{gather}
\tilde\tau_{12}=\tilde\tau+\sum_{n=1}^{\infty}\tau_{12,n}(\tilde\tau) \left( \frac{t^3}{\tilde r^2}\right)^n,\qquad
\tilde y=\tilde r\left( 1+\sum_{n=1}^{\infty} r_n(\tilde\tau) \left( \frac{t^3}{\tilde r^2} \right)^n \right).\label{affine in lambda t3 deformation of periods new coordinates z2 over z3}
\end{gather}
In particular, the functions \eqref{affine in lambda t3 deformation of periods new coordinates z2 over z3} are holomorphic in the domain
\begin{gather*}
\hat{\mathcal{U}}=\{ (\tilde\tau,\tilde r,\tilde s)\in \mathbb{H}\times\mathbb{C}^{*}\times \mathbb{C}\mid |\tilde s|<\delta(\tilde\tau,\tilde r)\}.
\end{gather*}
Using the last equation of \eqref{change of extended coordinates flat intersection and auxiliary} in \eqref{affine in lambda t3 deformation of periods new coordinates z2 over z3}, we obtain
 \begin{gather}
\tilde\tau_{12}=\tilde\tau+\sum_{n=1}^{\infty}\frac{\tau_{12,n}(\tilde\tau)}{\Delta^{\frac{n}{3}}(\tilde\tau)} \bigl( Q^{\frac{1}{3}}t^3\bigr)^n,\qquad
\tilde y=\tilde r\left( 1+\sum_{n=1}^{\infty} \frac{r_n(\tilde\tau)}{\Delta^{\frac{n}{3}}(\tilde\tau)} \bigl(Q^{\frac{1}{3}} t^3 \bigr)^n \right),\label{two equations for tilde tau and tilde y inside lemma 0}
\end{gather}
which are holomorphic function in the domains
\begin{gather*}
\tilde{\mathcal{U}}=\bigl\{ \bigl(\tilde\tau,Q^{\frac{1}{3}}t^3\bigr)\in\mathbb{H}\times \mathbb{C}\mid \big|Q^{\frac{1}{3}}t^3\big|<\epsilon,\, E_6(\tilde\tau)\neq 0 \bigr\},
\end{gather*}
because the domain $\tilde{\mathcal{U}}$ corresponds to the domain $\hat{\mathcal{U}}$ when written in the coordinates $\bigl(\tilde\tau,Q^{\frac{1}{3}}t^3\bigr)$.
Here $\epsilon$ is such that
\begin{gather*}
\det \bigl(g^{\alpha\beta}-\lambda\eta^{\alpha\beta}\bigr)=\bigl( \bigl(t^1-\lambda\bigr)^3+27Q\bigr)\bigl(1+ O\bigl(t^3\bigr) \bigr)=0 \iff \bigl(t^1-\lambda\bigr)^3+27Q=0,
\end{gather*}
because $\epsilon$ is proportional to $\delta(\tilde\tau,\tilde r)$, which have the property above.
Furthermore, using \eqref{change of extended coordinates flat intersection and auxiliary}, we have
\begin{gather}\label{discriminant in lambda E6 inside lemma}
\bigl(t^1-\lambda\bigr)^3+27Q= -8(2\pi)^6\frac{E_6^2(\tilde\tau)}{\tilde r^6}.
\end{gather}
 In addition, for fixed $Q^{\frac{1}{3}}t^3$, consider the function
 \begin{align}
 \tilde\tau_{12}\colon\ \mathbb{H}\setminus{{\rm SL}_2(\mathbb{Z})(i)}& \mapsto \tilde\tau_{12}( \mathbb{H}\setminus{{\rm SL}_2(\mathbb{Z})(i)})\subset \mathbb{H}, \nonumber\\
 \tilde\tau& \mapsto \tilde\tau_{12}=\tilde\tau+\sum_{n=1}^{\infty}\frac{\tau_{12,n}(\tilde\tau)}{\Delta^{\frac{n}{3}}(\tilde\tau)} \bigl( Q^{\frac{1}{3}}t^3\bigr)^n,\label{affine in lambda t3 deformation of periods new coordinates z2 over z3 lemma 001}
 \end{align}
 and the inverse function of \eqref{affine in lambda t3 deformation of periods new coordinates z2 over z3 lemma 001}
$
 \tilde\tau\colon \tilde\tau_{12}( \mathbb{H}\setminus{{\rm SL}_2(\mathbb{Z})(i)})\mapsto \mathbb{H}\setminus{{\rm SL}_2(\mathbb{Z})(i)}
$
 given by
 \begin{gather}\label{inverse function of tilde tau to tilde tau 12 inside lemma}
 \tilde\tau=\tilde\tau_{12}+\sum_{n=1}^{\infty}\frac{\tilde\tau_{12,n}(\tilde\tau_{12})}{\Delta^{\frac{n}{3}}(\tilde\tau_{12})} \bigl( Q^{\frac{1}{3}}t^3\bigr)^n,
 \end{gather}
 where the coefficients $\tilde\tau_{12,n}(\tilde\tau_{12})$ are obtained by the relation
 \begin{align}
 \tilde\tau_{12}(\tilde\tau( \tilde\tau_{12} ) )&=\tilde\tau_{12}
 =\sum_{n=0}^{\infty}\frac{\tau_{12,n}(\tilde\tau( \tilde\tau_{12}))}{\Delta^{\frac{n}{3}}(\tilde\tau( \tilde\tau_{12} ))} \bigl( Q^{\frac{1}{3}}t^3\bigr)^n\nonumber\\
 &= \sum_{n=0}^{\infty}\left(\sum_{k=0}^n \frac{\partial^k}{\partial\tilde\tau_{12}^k} \left(\frac{\tau_{12,n} ( \tilde\tau_{12} )}{\Delta^{\frac{n}{3}} ( \tilde\tau_{12} )}\right)B_{n,k}\left( \frac{\tilde\tau_{12,1}}{\Delta^{\frac{1}{3}}},\frac{\tilde\tau_{12,2}}{\Delta^{\frac{2}{3}}},\dots ,\frac{\tilde\tau_{12,n}}{\Delta^{\frac{n}{3}}} \right) \right) \frac{\bigl( Q^{\frac{1}{3}}t^3\bigr)^n}{n!},\label{algebraically dependent identity}
 \end{align}
 where $B_{n,k}$ are the Bell polynomials defined in \eqref{Bell polynomials} and we have used the Faa di Bruno formula~\eqref{Faa di Bruno}.
 Since $\tilde\tau_{12,n}$ is algebraically dependent with the coefficients $\tau_{12,n}$ due to equation~\eqref{algebraically dependent identity}, we have that
$
 \tilde\tau_{12,n}\in \mathbb{C}\big[E_2,E_4,E_6,E_6^{-1}\big]$,
 because of condition \eqref{coefficients of tau12 inside corollary Milanov} in the Corollary~\ref{lemma Milanov deformation corollary}.

Composing \eqref{inverse function of tilde tau to tilde tau 12 inside lemma} with the second equation of \eqref{two equations for tilde tau and tilde y inside lemma 0}, we obtain a power series of the following form:
\begin{gather}\label{deformation r part2 0 inside lemma}
\tilde y=\tilde r\left( 1+\sum_{n=1}^{\infty} \frac{\tilde r_n(\tilde\tau_{12})}{\Delta^{\frac{n}{3}}(\tilde\tau_{12})} \bigl(Q^{\frac{1}{3}} t^3 \bigr)^n \right).
\end{gather}

 The image of the discriminant $E_6(\tilde\tau)=0$ under the change of coordinates \eqref{affine in lambda t3 deformation of periods new coordinates z2 over z3 inverse lemma} is obtained by substituting \eqref{inverse function of tilde tau to tilde tau 12 inside lemma} and \eqref{deformation r part2 0 inside lemma} in \eqref{discriminant in lambda E6 inside lemma},
 \begin{align}
\frac{E_6\bigl(\tilde\tau_{	12},Q^{\frac{1}{3}}t^3\bigr)}{\tilde y^3}:={}&\frac{E_6\bigl(\tilde\tau_{12}+\sum_{n=1}^{\infty}\frac{\tilde\tau_{12,n}(\tilde\tau_{12})}{\Delta^{\frac{n}{3}}(\tilde\tau_{12})} \bigl( Q^{\frac{1}{3}}t^3\bigr)^n \bigr)}{\tilde y^3\bigl( 1+\sum_{n=1}^{\infty} \frac{\tilde r_n(\tilde\tau_{12})}{\Delta^{\frac{n}{3}}(\tilde\tau_{12})} \bigl(Q^{\frac{1}{3}} t^3 \bigr)^n \bigr)^{-3}},\nonumber\\
={} &\frac{1}{\tilde y^3}\left[E_6(\tilde\tau_{12})+\sum_{n=1}^{\infty} E_{6,n}(\tilde\tau_{12}) \bigl( Q^{\frac{1}{3}}t^3 \bigr)^n\right].\label{deformed E6 in Qt3}
 \end{align}
 Then,
$
\tilde\tau_{12}( \mathbb{H}\setminus{{\rm SL}_2(\mathbb{Z})(i)})= \mathbb{H}\setminus\bigl\{E_6\bigl(\tilde\tau_{	12},Q^{\frac{1}{3}}t^3\bigr)=0\bigr\}$.
 Moreover, $\tilde\tau_{12}^0$ is zero of \eqref{deformed E6 in Qt3} if the function
 \begin{gather*}
\frac{1}{\tilde y^3}\left[E_6\bigl(\tilde\tau_{12}^0\bigr)+\sum_{n=1}^{\infty} E_{6,n}(\tilde\tau_{12}^0) \bigl( Q^{\frac{1}{3}}t^3 \bigr)^n\right]
 \end{gather*}
 is the zero function in the variable \smash{$Q^{\frac{1}{3}}t^3$}, i.e.,
$
E_{6,n}(\tilde\tau_{12})=0$, $ \forall n\geq 0$.
 In particular, the zero order coefficient of the Taylor expansion \eqref{deformed E6 in Qt3} is zero.

 Therefore,
$ \tilde\tau_{12}\bigl(\tilde{\mathcal{U}}\bigr)=\mathbb{H}\setminus\{{\rm SL}_2(\mathbb{Z})(i)\}$.
 Lemma proved.
 \end{proof}

Here, we make the following deformation of the modular discriminant. Composing the $t^3$ deformations \eqref{affine in lambda t3 deformation of periods new coordinates z2 over z3 inverse lemma} and \eqref{deformation r part2 0 inside lemma} with the last equation of \eqref{change of extended coordinates flat intersection and auxiliary}, we obtain
\[
Q=\frac{\Delta\bigl(\tilde\tau_{12},Q^{\frac{1}{3}} t^3 \bigr)}{\tilde y^6},
\]
where
\begin{align}
\Delta\bigl(\tilde\tau_{12},Q^{\frac{1}{3}}t^3\bigr)&=\frac{\Delta\bigl(\tilde\tau=\tilde\tau_{12}+\sum_{n=1}^{\infty}\frac{\tilde\tau_{12,n}(\tilde\tau_{12})}{\Delta^{\frac{n}{3}}(\tilde\tau_{12})} \bigl( Q^{\frac{1}{3}}t^3\bigr)^n \bigr)}{\bigl( 1+\sum_{n=1}^{\infty} \frac{\tilde r_n(\tilde\tau_{12})}{\Delta^{\frac{n}{3}}(\tilde\tau_{12})} \bigl( Q^{\frac{1}{3}}t^3 \bigr)^n \bigr)^6},\nonumber\\
&=\Delta(\tilde\tau_{12})+\sum_{n=1}^{\infty} \Delta_n(\tilde\tau_{12})\bigl( Q^{\frac{1}{3}}t^3 \bigr)^n.\label{modular deformation tilde 0}
\end{align}

\begin{Lemma}\label{lemma w1 dependence of t3 e Q root 3}
Let the power series $w_1$, $w_2$, $w_3$ be defined in \eqref{t3 deformation of periods 2}.
Then, the function
\begin{gather*}
\tilde\tau_1\tilde\tau_2:= \frac{\tilde w_1\bigl(t^1-\lambda,Q,t^3\bigr)}{\tilde w_3\bigl(t^1-\lambda,Q,t^3\bigr)}
\end{gather*}
satisfies the following relation:
\begin{gather}\label{t3 relation tilde w123}
t^3=\tilde w_2^2-4\tilde w_1\tilde w_3.
\end{gather}
Moreover, consider
$
(\tilde\tau_1-\tilde\tau_2 )^2:=\frac{t^3}{\tilde w_3^2}$.
Hence, the following relation holds true:
\[
\tilde\tau_1\tilde\tau_2=\tilde\tau_{12}^2-\smash{\frac{1}{4}\frac{Q^{\frac{1}{3}}t^3 }{\Delta^{\frac{1}{3}}\bigl(\tilde\tau_{12},Q^{\frac{1}{3}}t^3\bigr)}}.
\]
\end{Lemma}

\begin{proof}
The relation \eqref{t3 relation tilde w123} holds true due to Lemma \ref{lemma quadratic relations Milanov}. Furthermore,
\begin{align*}
\tilde\tau_1\tilde\tau_2&= \tilde\tau_{12}^2-\frac{( \tilde\tau_1-\tilde\tau_2)^2}{4}= \tilde\tau_{12}^2-\frac{\bigl(t^3\bigr)^2}{4\tilde w_3^2}=\tilde\tau_{12}^2-\frac{1}{4}\frac{Q^{\frac{1}{3}}t^3 }{\Delta^{\frac{1}{3}}\bigl(\tilde\tau_{12},Q^{\frac{1}{3}}t^3\bigr)}.\tag*{\qed}
\end{align*} \renewcommand{\qed}{}
\end{proof}

Since, the extended Gauss--Manin connection \eqref{Extended Gauss--Manin connection} has the same monodromy group of the Gauss--Manin connection
\eqref{Gauss--Manin connection}, the periods $\tilde\tau_{1}$, $\tilde\tau_{2}$, $\tilde y$ have the following transformation law due to \eqref{diagonal action of PSL2Z times a1}:
\begin{gather}
A(\tilde\tau_1,\tilde\tau_2,\tilde y)=\left( \frac{a\tilde\tau_1+b}{c\tilde\tau_1+d}, \frac{a\tilde\tau_2+b}{c\tilde\tau_2+d},\tilde y(c\tilde\tau_1+d)(c\tilde\tau_1+d)\right),\nonumber\\
B(\tilde\tau_1,\tilde\tau_2,\tilde y)=(\tilde\tau_1,\tilde\tau_2,-\tilde y),\label{diagonal action of PSL2Z times a1 tilde}
\end{gather}
where $\bigl(\begin{smallmatrix}
a& b \\
c&d \\
\end{smallmatrix}\bigr) \in {\rm SL}_2(\mathbb{Z})$.
Using the change of coordinates \eqref{covering map symmetric square big quantum cohomology}, we have that $\tilde\tau_{12}$ has the following transformation law under \eqref{diagonal action of PSL2Z times a1 tilde}:
\begin{gather}
\text{If} \quad (\tilde\tau_1,\tilde\tau_2)\mapsto (\tilde\tau_1+1,\tilde\tau_2+1), \implies \tilde\tau_{12}\mapsto \tilde\tau_{12}+1,\nonumber\\
\text{If} \quad (\tilde\tau_1,\tilde\tau_2)\mapsto \left(\frac{-1}{\tilde\tau_1},\frac{-1}{\tilde\tau_2}\right), \implies \tilde\tau_{12}\mapsto \frac{-\tilde\tau_{12}}{\tilde\tau_1\tilde\tau_2}.\label{group action in the average of tau}
\end{gather}
Alternatively, due to Lemma \ref{lemma w1 dependence of t3 e Q root 3} the group action \eqref{group action in the average of tau} can be written as
\begin{gather}
\text{If} \quad (\tilde\tau_1,\tilde\tau_2)\mapsto (\tilde\tau_1+1,\tilde\tau_2+1), \implies \tilde\tau_{12}\mapsto \tilde\tau_{12}+1,\nonumber\\
\text{If} \quad (\tilde\tau_1,\tilde\tau_2)\mapsto \left(\frac{-1}{\tilde\tau_1},\frac{-1}{\tilde\tau_2}\right), \implies \tilde\tau_{12}\mapsto \frac{-\tilde\tau_{12}}{\left(\tilde\tau_{12}^2-\frac{1}{4}\frac{Q^{\frac{1}{3}}t^3}{\Delta^{\frac{1}{3}}(\tilde\tau_{12},Q^{\frac{1}{3}}t^3)}\right)}.\label{group action of a1psl2 and related group in tilde and qcubict3}
\end{gather}
Note that for $t^3$ small enough the group action \eqref{group action of a1psl2 and related group in tilde and qcubict3} gives rise to an automorphism of $\mathbb{H}$.
On another hand, the $A_1\times {\rm PSL}_2(\mathbb{Z})$ action in the periods $\tilde\tau$, $\tilde r$, $t^3$ is the following:
\[
A\bigl(\tilde\tau,\tilde r,t^3\bigr)=\left( \frac{a\tilde\tau+b}{c\tilde\tau+d},\tilde r(c\tilde\tau+d)^2,t^3\right),\qquad
B\bigl(\tilde\tau,\tilde r,t^3\bigr)=\bigl(\tilde\tau,-\tilde r,t^3\bigr),
\]
where $\bigl(\begin{smallmatrix}
a& b \\
c&d \\
\end{smallmatrix}\bigr) \in {\rm SL}_2(\mathbb{Z})$.
Therefore, the map \eqref{affine in lambda t3 deformation of periods new coordinates z2 over z3 inverse lemma} has the following transformation law:
\begin{gather*}
\tilde\tau\bigl( \tilde\tau_{12}+1 ,Q^{\frac{1}{3}}t^3 \bigr)=\tilde\tau\bigl( \tilde\tau_{12} ,Q^{\frac{1}{3}}t^3 \bigr)+1,\\
\tilde\tau\left(\frac{-\tilde\tau_{12}}{\left(\tilde\tau_{12}^2-\frac{1}{4}\frac{Q^{\frac{1}{3}}t^3}{\Delta^{\frac{1}{3}}(\tilde\tau_{12},Q^{\frac{1}{3}}t^3)}\right)} ,Q^{\frac{1}{3}}t^3 \right)=\frac{-1}{\tilde\tau\bigl( \tilde\tau_{12} ,Q^{\frac{1}{3}}t^3 \bigr)}.
\end{gather*}

At this stage, we can prove the main theorem of this section.

\begin{Theorem}\quad
\begin{enumerate}\itemsep=0pt
\item[$(1)$]
Let the function
$
\tilde\tau_{12}\colon \mathbb{H}\setminus{ {\rm SL}_2(\mathbb{Z})({\rm i})} \mapsto \mathbb{H}\setminus{ {\rm SL}_2(\mathbb{Z})({\rm i})}
$
given by
\[
\tilde\tau\bigl(\tilde\tau_{12},Q^{\frac{1}{3}}t^3\bigr)=\tilde\tau_{12}+\sum_{n=1}^{\infty}\frac{\tilde\tau_{n}(\tilde\tau_{12})}{\Delta^{\frac{n}{3}}(\tilde\tau_{12})} \bigl( Q^{\frac{1}{3}}t^3\bigr)^n
\]
be defined in \eqref{affine in lambda t3 deformation of periods new coordinates z2 over z3 inverse lemma}.
Then, the Landau--Ginzburg superpotential of big quantum cohomology of $\mathbb{CP}^2$ is a family of functions
$
\lambda\bigl(\tilde\tau_{12}, t^1, Q, t^3\bigr)\colon \mathbb{H} \mapsto \mathbb{C}
$
with holomorphic dependence in the parameter space
\begin{gather*}
\left\{\bigl(t^1,Q^{\frac{1}{3}}, Q^{\frac{1}{3}}t^3\bigr) \in \mathbb{C}\times\mathbb{C}^{*}\times \mathbb{C}\mid \big|Q^{\frac{1}{3}}t^3\big|<\left(\frac{1}{a}\right)^{\frac{1}{3}} \right\}
\end{gather*}
and given by
\begin{gather}\label{lg superpotential of small quantum cohomology j function big quantum cohomology main body}
\lambda\bigl(\tilde\tau_{12}, t^1, Q, t^3\bigr)=t^1+3Q^{\frac{1}{3}}J^{\frac{1}{3}}\bigl(\tilde\tau_{12}, Q^{\frac{1}{3}}t^3\bigr),
\end{gather}
where
\begin{align}
J^{\frac{1}{3}}\bigl(\tilde\tau_{12}, Q^{\frac{1}{3}}t^3\bigr)&:=J^{\frac{1}{3}}\left(\tilde\tau_{12}+\sum_{n=1}^{\infty}\frac{\tilde\tau_{n}(\tilde\tau_{12})}{\Delta^{\frac{n}{3}}(\tilde\tau_{12})} \bigl( Q^{\frac{1}{3}}t^3\bigr)^n\right),\nonumber\\
&=J^{\frac{1}{3}}(\tilde\tau_{12})+\sum_{n=1}^{\infty}J_n^{\frac{1}{3}}(\tilde\tau_{12})\bigl( Q^{\frac{1}{3}}t^3\bigr)^n,\label{definition of t3 deformation of gamma2}
\end{align}
and
\[
J_n^{\frac{1}{3}}(\tilde\tau_{12}):=\frac{1}{n!}\sum_{k=0}^n k!\frac{\partial^kJ^{\frac{1}{3}}}{\partial\tilde\tau_{12}}(\tilde\tau_{12}) B_{n,k}\left( 1!\frac{\tilde\tau_{1}(\tilde\tau_{12})}{\Delta^{\frac{1}{3}}(\tilde\tau_{12})} ,2!\frac{\tilde\tau_{2}(\tilde\tau_{12})}{\Delta^{\frac{2}{3}}(\tilde\tau_{12})},\dots,n!\frac{\tilde\tau_{n}(\tilde\tau_{12})}{\Delta^{\frac{n}{3}}(\tilde\tau_{12})}\right).
\]
The functions $B_{n,k}(x_1,\dots,x_n)$ are the Partial Bell polynomials defined in \eqref{Bell polynomials}.

The coefficients \smash{$J_n^{\frac{1}{3}}(\tilde\tau_{12})$} belong to the following ring
\smash{$
J_n^{\frac{1}{3}}(\tilde\tau_{12}) \in \Delta^{\frac{-n}{3}}\mathbb{C}[E_2,E_4,E_6]$}.

 In addition, the correspondent Abelian differential $\phi$ is given by
\begin{gather}\label{volume form of Big quantum cohomology 1 Milanov deformation section}
\phi=-\frac{2^{\frac{5}{2}}}{2\pi}\frac{\Delta^{\frac{1}{6}}\bigl(\tilde\tau_{12}, Q^{\frac{1}{3}}t^3 \bigr)}{Q^{\frac{1}{6}}}{\rm d}\tilde\tau_{12},
\end{gather}
where $\Delta\bigl(\tilde\tau_{12}, Q^{\frac{1}{3}}t^3 \bigr)$ is defined in \eqref{modular deformation tilde 0}.

\item[$(2)$]
Denote by \smash{$\Gamma^{(3)}_{Q^{1/3}t^3}$} the image of the group $\Gamma^{(3)}$ under the group homomorphism \eqref{group action of a1psl2 and related group in tilde and qcubict3}. Then the Landau--Ginzburg superpotential of big quantum cohomology of $\mathbb{CP}^2$ \eqref{lg superpotential of small quantum cohomology j function big quantum cohomology main body} is \smash{$\Gamma^{(3)}_{Q^{1/3}t^3}$}-invariant.
\end{enumerate}

\end{Theorem}

\begin{proof}

 We obtain the Landau--Ginzburg superpotential of $QH^{*}\bigl(\mathbb{CP}^2\bigr)$ for $t^3\neq 0$ by applying the Dubrovin construction of Landau--Ginzburg superpotential, see Theorem \ref{Dubrovin superpotential theorem}. More specifically, in order to obtain the LG superpotential, we invert \eqref{affine extended w2 w3 Big lemma} in $\lambda$, which is done in practice by composing \eqref{affine in lambda t3 deformation of periods new coordinates z2 over z3 inverse lemma} with \eqref{lg superpotential of small quantum cohomology j function }.

Using Faa di Bruno formula \eqref{Faa di Bruno}
\begin{align}
\lambda\bigl(\tilde\tau_{12}, t^1, Q,t^3\bigr)&=t^1+3Q^{\frac{1}{3}}J^{\frac{1}{3}}\left(\tilde\tau_{12}+\sum_{n=1}^{\infty}\frac{\tilde\tau_{n}(\tilde\tau_{12})}{\Delta^{\frac{n}{3}}(\tilde\tau_{12})} \bigl( Q^{\frac{1}{3}}t^3\bigr)^n\right)\nonumber\\
&=t^1+3Q^{\frac{1}{3}}J^{\frac{1}{3}}\left(\tilde\tau_{12}+\sum_{n=1}^{\infty}n!\frac{\tilde\tau_{n}(\tilde\tau_{12})}{\Delta^{\frac{n}{3}}(\tilde\tau_{12})} \frac{\bigl( Q^{\frac{1}{3}}t^3\bigr)^n}{n!}\right)\nonumber\\
&=t^1+3Q^{\frac{1}{3}}\sum_{n=0}^{\infty} \left(\sum_{k=0}^n k! \frac{\partial^kJ^{\frac{1}{3}}}{\partial\tilde\tau_{12}} B_{n,k}\left( 1!\frac{\tilde\tau_{1}}{\Delta^{\frac{1}{3}}},2!\frac{\tilde\tau_{2}}{\Delta^{\frac{2}{3}}},\dots,n!\frac{\tilde\tau_{n}}{\Delta^{\frac{n}{3}}}\right) \right)\frac{\bigl( Q^{\frac{1}{3}}t^3\bigr)^n}{n!}\nonumber\\
&=t^1+3Q^{\frac{1}{3}}\sum_{n=0}^{\infty}J_n^{\frac{1}{3}}(\tilde\tau_{12})\bigl( Q^{\frac{1}{3}}t^3\bigr)^n\nonumber\\
&=t^1+3Q^{\frac{1}{3}}J^{\frac{1}{3}}\bigl(\tilde\tau_{12}, Q^{\frac{1}{3}}t^3\bigr).\label{lg superpotential of Big quantum cohomology j function first version}
\end{align}
The coefficients \smash{$J_n^{\frac{1}{3}}(\tilde\tau_{12})$} are polynomial in \smash{$\frac{\tau_{n}(\tilde\tau_{12})}{\Delta^{n/3}(\tilde\tau_{12})}$}, which belongs to \smash{$\Delta^{\frac{-n}{3}}\mathbb{C}\big[E_2,E_4,E_6,E_6^{-1}\big]$} due to Corollary~\ref{lemma Milanov deformation corollary} and Lemma \ref{lemma Milanov deformation 1 on LG superpotential}. But the Taylor expansion of any Landau--Ginzburg superpotential near the discriminant is of the form
\begin{gather*}
\lambda=u_i+\lambda^{\prime\prime}(\tilde\tau_{12}^0)\frac{\bigl(\tilde\tau_{12}-\tilde\tau_{12}^0 \bigr)^2}{2}+O\bigl(\tilde\tau_{12}-\tilde\tau_{12}^0 \bigr)^3,
\end{gather*}
which is holomorphic. Hence, we can extend the domain of Landau--Ginzburg superpo\-ten\-tial from $\mathbb{H} \setminus\{E_6((\tilde\tau_{12})\!=\!0 \}$ to $\mathbb{H}$, which also implies that the coefficients \smash{$J_n^{\frac{1}{3}}(\tilde\tau_{12})$} are not~rational in~$E_6$. Moreover, the Landau--Ginzburg superpotential is holomorphic with respect its parameters $\bigl(t^1,Q^{\frac{1}{3}},Q^{\frac{1}{3}}t^3\bigr)$ in the space
\begin{gather*}
\bigg\{\bigl(t^1,Q^{\frac{1}{3}}, Q^{\frac{1}{3}}t^3\bigr) \in \mathbb{C}\times\mathbb{C}^{*}\times \mathbb{C}\mid \big|Q^{\frac{1}{3}}t^3\big|<\left(\frac{1}{a}\right)^{\frac{1}{3}} \bigg\}
\end{gather*}
because, we can extend the parameter space over the discriminant locus as it was discussed in Remark \ref{largest image of the inverse period map Gauss Manin connection}.

Here, we conclude that the pair
\begin{gather*}
\bigl(\lambda=t^1+3Q^{\frac{1}{3}}J^{\frac{1}{3}}\bigl(\tilde\tau_{12},Q^{\frac{1}{3}}t^3\bigr), \phi={\rm d}\tilde w_2\bigr)
\end{gather*}
 serve as a Landau--Ginzburg superpotential for big quantum cohomology of $\mathbb{CP}^2$, where $\tilde w_2$ is the pullback of $\tilde w_2=w_2\bigl(t^1-\lambda,Q,t^3\bigr)$, via the map \eqref{lg superpotential of small quantum cohomology j function big quantum cohomology main body}.

The domain of \eqref{lg superpotential of Big quantum cohomology j function first version} is a family of $\mathbb{H}$ parametrized by
\begin{gather*}
\bigl(\tau_{12},y,Q^{\frac{1}{3}}t^3\bigr)\in \left(\mathbb{H}\setminus\{E_6(\tau_{12})=0 \} \right)\times \mathbb{C}^{*}\times D\biggl(0, \left(\frac{1}{a}\right)^{\frac{1}{3}} \biggr)
\end{gather*}
because of the parametrized LG superpotential \eqref{lg superpotential of Big quantum cohomology j function first version} and the change of coordinates
\eqref{inverse period map of Big quantum cohomology}.
Furthermore, this family of $\mathbb{H}$ is biholomorphic to the following family of manifolds:
\begin{gather}
\mathbb{H}_{(\tau_{12},y,Q^{\frac{1}{3}}t^3)} = \bigg\{ \bigl(\tilde\tau_{12},\tilde y ,\tau_{12},y,Q^{\frac{1}{3}}t^3 \bigr) \in (\mathbb{H}\times \mathbb{C}^{*})^2\times D\biggl(0, \left(\frac{1}{a}\right)^{\frac{1}{3}} \biggr)\colon\nonumber \\
\qquad\tilde y=\frac{2^{\frac{3}{2}}}{(2\pi)^2}\biggl(\frac{\Delta\bigl(\tilde\tau_{12},Q^{\frac{1}{3}}t^3\bigr)}{Q}\biggr)^{\frac{1}{6}}, E_6(\tau_{12})\neq 0 \bigg\}.\label{domain of big lg superpotential}
\end{gather}

Let the space $\Omega_{w_2,y}$ be defined by
\begin{gather*}
\Omega_{w_2,y}=\left\{ (w_2,y) \in \mathbb{C}^{*}\times \mathbb{C}^{*} \mid \operatorname{Im}\left(\frac{-w_2}{2y} \right)>0 \right\}.
\end{gather*}

Consider the following isomorphism $S\colon \mathbb{H}\times \mathbb{C}^{*}\times\mathbb{H}\times \mathbb{C}^{*} \mapsto \Omega_{\tilde w_2,\tilde y}\times \Omega_{w_2,y}$
\begin{gather*}
S(\tilde\tau_{12},\tilde y,\tau_{12}, y)=(\tilde w_2, \tilde y, w_2, y), \qquad \text{where} \quad \tilde w_2=-2\tilde y\tilde\tau_{12}, \qquad w_2= -2y\tau_{12}.
\end{gather*}
Hence, the section ${\rm d}\tilde w_2 \in \Gamma\bigl( T^{*}( \Omega_{\tilde w_2,\tilde y}\times \Omega_{w_2,y} ) \times D\bigl(0, \bigl(\frac{1}{a}\bigr)^{\frac{1}{3}} \bigr)\bigr)$ projected to the submanifold $\mathbb{H}_{(\tau_{12},y,Q^{\frac{1}{3}}t^3)}$ defined in \eqref{domain of big lg superpotential} is the section
\smash{$
 {\rm d}\tilde w_2 \in \Gamma\bigl( T^{*}\bigl(\mathbb{H}_{(\tau_{12},y,Q^{\frac{1}{3}}t^3)} \bigr)\bigr)
$}
 and is given by
\begin{gather*}
{\rm d}\tilde w_2=-2\tilde y{\rm d}\tilde\tau_{12}=-\frac{2^{\frac{5}{2}}}{2\pi}\frac{\Delta^{\frac{1}{6}}\bigl(\tilde\tau_{12}, Q^{\frac{1}{3}}t^3 \bigr)}{Q^{\frac{1}{6}}}{\rm d}\tilde\tau_{12}.
\end{gather*}
Furthermore, by definition the function $J^{\frac{1}{3}}\bigl(\tilde\tau_{12}, Q^{\frac{1}{3}}t^3 \bigr)$ has the following form:
\begin{gather}
J^{\frac{1}{3}}\bigl(\tilde\tau_{12}, Q^{\frac{1}{3}}t^3\bigr)=J^{\frac{1}{3}}\left(\tilde\tau=\tilde\tau_{12}+\sum_{n=1}^{\infty}\frac{\tilde\tau_{n}(\tilde\tau_{12})}{\Delta^{\frac{n}{3}}(\tilde\tau_{12})} \bigl( Q^{\frac{1}{3}}t^3\bigr)^n\right)=J^{\frac{1}{3}}( \tilde\tau(\tilde\tau_{12}) ).\label{definition of t3 deformation of gamma2 inside theorem}
\end{gather}
Here, it is convenient to suppress the dependence of $Q^{\frac{1}{3}}t^3$ in $\tilde\tau(\tilde\tau_{12})$ to avoid a heavy notation.
Then, by the construction of the group homomorphism \eqref{group action of a1psl2 and related group in tilde and qcubict3}, the action of \smash{$\gamma \in \Gamma^{(3)}_{Q^{\frac{1}{3}}t^3}$} in the domain of \eqref{definition of t3 deformation of gamma2 inside theorem} is given by
\begin{align}
J^{\frac{1}{3}}\bigl(\gamma\tilde\tau_{12}, Q^{\frac{1}{3}}t^3\bigr)&=J^{\frac{1}{3}}( \tilde\tau(\gamma\tilde\tau_{12}) ),\nonumber\\
&=J^{\frac{1}{3}}\left( \frac{a\tilde\tau(\tilde\tau_{12})+b}{c\tilde\tau(\tilde\tau_{12})+d} \right),\qquad {\begin{pmatrix}
a& b \\
c&d \\
\end{pmatrix}} \in \Gamma^{(3)},\nonumber\\
&=J^{\frac{1}{3}}( \tilde\tau(\tilde\tau_{12}) ).\label{definition of t3 deformation of gamma2 inside theorem part 2}
\end{align}
In the last line of \eqref{definition of t3 deformation of gamma2 inside theorem part 2}, we have used that the function $J^{\frac{1}{3}}(\tilde\tau)$ is $\Gamma^{(3)}$-invariant and the matrix~${\bigl(\begin{smallmatrix}
a& b \\
c&d \\
\end{smallmatrix}\bigr) \in \Gamma^{(3)}}$ in the second line is the inverse of \smash{$\gamma \in \Gamma^{(3)}_{Q^{\frac{1}{3}}t^3}$} under the group homo\-morphism~\eqref{group action of a1psl2 and related group in tilde and qcubict3}.
Theorem proved.
\end{proof}

\begin{Remark}
Due to Remark \ref{zero of LG small quantum cohomology}, the zeros of the Landau--Ginzburg superpotential \eqref{lg superpotential of small quantum cohomology j function big quantum cohomology main body} is achieved when $\tilde\tau_{12}=\tau_{12}$. Therefore, the coefficients \smash{$J_n^{\frac{1}{3}}(\tilde\tau_{12})$} can be also obtained by the relation
\begin{gather}\label{faster algorithm to compute Big LG superpotential 1}
t^1=3Q^{\frac{1}{3}}J^{\frac{1}{3}}\bigl(\tau_{12}, Q^{\frac{1}{3}}t^3\bigr).
\end{gather}
More specifically, given the coefficients $t^1_n$ and $Q_n$ in \eqref{inverse period map of Big quantum cohomology}, we can derive \smash{$J_n^{\frac{1}{3}}(\tilde\tau_{12})$} by using the Taylor expansions of $t^1$ and $Q$ in $(\tau_1-\tau_2)^2$ and substituting in \eqref{faster algorithm to compute Big LG superpotential 1}.
\end{Remark}

\subsection{From canonical to flat coordinates}\label{From canonical to flat coordinates}

The objective of this subsection is to provide a geometric interpretation of the $t^3$ deformation by considering the analytic continuation of solutions from the flat coordinate system of the Dubrovin connection. To achieve this, we begin with a brief overview of the theory of n-dimensional semisimple Dubrovin--Frobenius manifolds, we refer to \cite{B.Dubrovin3,Guzzetti} for this part.

Let us consider flat sections of the Dubrovin connection constrained by the first condition of~\eqref{Dubrovin Connection}. We can then select the holomorphic part of \eqref{solution of Dubrovin connection near 0} as the fundamental matrix solution
\[
 H(z,t) = \bigl(\eta^{\beta\gamma}\partial_{\gamma}\tilde t_{\alpha}\bigr) = \sum_{p=0}^{\infty} H_{p}(t)z^p,
\]
if we prefer to work in Saito flat coordinates $\bigl(t^1,t^2,\dots,t^n\bigr)$, or alternatively, we can choose the holomorphic part of \eqref{solution of Dubrovin connection near 0 canonical coordinates}
\[
 \Psi(z,t) = (Y_{i\alpha}) = \sum_{p=0}^{\infty}\Psi_p(u)z^p,
\]
if we prefer the canonical coordinates $(u_1,u_2,\dots,u_n)$. The coefficients of the matrices are related by the following:
\begin{gather}\label{relationship coefficients holomorphic matrix}
 \Psi_{i\alpha} = \psi_{i\beta}\eta^{\beta\gamma}\partial_{\gamma}\tilde t_{\alpha}.
\end{gather}

The flat coordinate system of the first condition of \eqref{Dubrovin Connection} deforms the Saito flat coordinates~${\bigl(t^1,t^2,\dots,t^n\bigr)}$ as follows:
\begin{gather}\label{Dubrovin deformed flat coordinates 1}
 \tilde t_{\alpha}(z,t) = t_{\alpha} + \sum_{p=1}^{\infty} H_{\alpha,p}(t)z^p, \qquad t_{\alpha} = \eta_{\alpha\beta}t^{\beta}.
\end{gather}
Here, $H_{\alpha,p}$ is determined by the following recursion:
\begin{gather*}
 H_{\alpha,0} = t_{\alpha} = \eta_{\alpha\beta}t^{\beta}, \qquad
 \partial_{\gamma}\partial_{\beta}H_{\alpha,p+1} = c^{\epsilon}_{\gamma\beta}\partial_{\epsilon}H_{\alpha,p}, \qquad p=0,1,2,\dots.
\end{gather*}
In this setting, based on \cite[Exercise 2.7]{B.Dubrovin3}, the following identity holds true:
\begin{gather}\label{identity exercise 2.7 cargese}
 t_{\alpha} = \langle \nabla H_{\alpha,0}, \nabla H_{1,1} \rangle = \eta^{\mu\lambda}\partial_{\mu}H_{\alpha,0}\partial_{\lambda}H_{1,1}.
\end{gather}
Thus, substituting the equation \eqref{Dubrovin deformed flat coordinates 1} into \eqref{relationship coefficients holomorphic matrix} and the change of basis
\begin{gather}\label{change of basis vector field Saito canonical}
 \frac{\partial}{\partial t^{\alpha}} = \frac{\psi_{i\alpha}}{\psi_{i1}}\frac{\partial}{\partial u_i},
\end{gather}
we obtain
\begin{gather}\label{formula psi H}
 \Psi_{i\alpha,p} = \frac{\partial_i H_{\alpha,p} }{\psi_{i1}}, \qquad \text{where} \quad \Psi_{p} = \left(\Psi_{i\alpha,p} \right).
\end{gather}
Therefore, by substituting \eqref{formula psi H} and \eqref{change of basis vector field Saito canonical} into \eqref{identity exercise 2.7 cargese}, we get
\begin{gather}\label{generic relation flat to canonical coordinates}
 t_{\alpha}(u) = \sum_{i=1}^n \Psi_{i\alpha,0}(u)\Psi_{i1,1}(u), \qquad \alpha=1,2,\dots,n.
\end{gather}

The functions \eqref{generic relation flat to canonical coordinates} are multivalued functions within the domain
\begin{gather*}
 \{ (u_1,u_2,\dots,u_n) \in \mathbb{C}^n \mid u_i \neq u_j, \text{ if } i \neq j\} = \mathbb{C}^n \setminus \{u_i=u_j\},
\end{gather*}
since they are derived from solutions of the Dubrovin flat coordinate system within the semisimple locus. To express \eqref{generic relation flat to canonical coordinates} more clearly, consider the compatibility condition of the Dubrovin connection flat coordinate system in canonical coordinates \eqref{Dubrovin connection in canonical coordinates 1}, \eqref{U,V, Vi in canonical coordinates}, given by
\smash{$
 \frac{\partial^2 Y}{\partial z\partial u_i} = \frac{\partial^2 Y}{\partial u_i\partial z}$},
which specifically appears as
\begin{gather}
 \frac{\partial V}{\partial u_i} = [V_i,V], \qquad
 [U,V_i] = [E_i,V].\label{compatibility of Dubrovin system in canonical coordinates}
\end{gather}

Furthermore, due to the following fact
\[
 [U,V_k] = [E_k,V] \implies (V_k)_{ij} = \frac{\delta_{ki}-\delta_{kj}}{u_i-u_j}V_{ij},
 \]
the system~\eqref{compatibility of Dubrovin system in canonical coordinates} can be expressed as
\begin{gather}
 \frac{\partial V}{\partial u_i} = [V_i,V], \qquad
 \frac{\partial \Psi}{\partial u_i} = V_i\Psi.\label{compatibility of Dubrovin system in canonical coordinates 2}
\end{gather}

\begin{Remark}
The first equation \eqref{compatibility of Dubrovin system in canonical coordinates 2} can be presented as a time-dependent Hamiltonian system
$
 \frac{\partial V}{\partial u_i} = \{V,H_i\}$, $ i=1,2,\dots,n$,
with a quadratic Hamiltonian
\begin{gather}\label{quadratic Hamiltonian}
 H_i(V;u) = \frac{1}{2}\sum_{i \neq j}\frac{V_{ij}^2}{u_i-u_j}, \qquad i=1,2,\dots,n,
\end{gather}
employing the Poisson bracket
\begin{gather}\label{Poisson bracket}
 \{V_{ij} , V_{kl}\} = V_{il}\delta_{jk}-V_{jl}\delta_{ik}+V_{jk}\delta_{il}-V_{ik}\delta_{jl}.
\end{gather}
The Hamiltonians $H_i(V;u)$ \eqref{quadratic Hamiltonian} pairwise commute with respect to the Poisson bracket \eqref{Poisson bracket}, i.e., $\{H_i,H_j\} = 0$ for any $i$, $j$. Hence, the one-form
$
 \sum_{i=1}^n H_i(V;u) {\rm d}u_i
$
forms a closed form for any solution of \eqref{compatibility of Dubrovin system in canonical coordinates 2}. This implies the local existence of a function~$\tau_I(u)$ such that
\begin{gather}\label{isomonodromic tau function}
 \frac{\partial \ln\tau_I(u)}{\partial u_i} = H_i(V;u), \qquad i=1,2,\dots,n.
\end{gather}
The function $\tau_I$ defined in \eqref{isomonodromic tau function} is referred to as the isomonodromic $\tau$ function of the system~\eqref{compatibility of Dubrovin system in canonical coordinates 2}. In the context of semisimple quantum cohomology, $\tau_I$ satisfies a significant identity. Specifically, if $F^1$ represents the generating function of genus $1$ Gromov--Witten invariants, then~\smash{$
 F^1 = \log \frac{\tau_I}{J^{\frac{1}{24}}}$}, where \smash{$J = \det\big(\frac{\partial t^{\alpha}}{\partial u_i}\big) = \psi_{11} \cdots \psi_{n1}$}.
\end{Remark}

Here, our focus centers on the particular case of $3D$ Dubrovin--Frobenius manifolds. Therefore, we will use the results obtained by Guzzetti in \cite[Section 6.1]{Guzzetti}. Consider the metric $\eta$ and matrix~${V(u_1,u_2,u_3)}$,
\begin{gather*}
\eta = \begin{pmatrix}
0 & 0 & 1 \\
0 & 1 & 0 \\
1 & 0 & 0
\end{pmatrix}, \qquad V(u) = \begin{pmatrix}
0 & -\Omega_3 & \Omega_2 \\
\Omega_3 & 0 & -\Omega_1 \\
-\Omega_2 & \Omega_1 & 0
\end{pmatrix}.
\end{gather*}
Recall that the columns of matrix $\Psi$ are the eigenvectors of matrix $V$. Moreover, it satisfies
$
\Psi^{T}\Psi = \eta$.
Therefore, we can write $\Psi$ in the following form:
\begin{gather*}
\Psi(u) = \begin{pmatrix}
\frac{E_{11}}{f} & E_{12} & E_{13}f \\[2mm]
\frac{E_{21}}{f} & E_{22} & E_{23}f \\[2mm]
\frac{E_{31}}{f} & E_{32} & E_{33}f
\end{pmatrix},
\end{gather*}
where
\begin{gather}
E_{11} = \frac{\Omega_1\Omega_2-\mu\Omega_3}{2\mu^2}, \qquad E_{12} = \frac{\Omega_1}{i\mu}, \qquad E_{13} = -\frac{\Omega_1\Omega_2+\mu\Omega_3}{\Omega_1^2+\Omega_3^2}, \nonumber\\
E_{21} = -\frac{\Omega_1^2+\Omega_3^2}{2\mu^2}, \qquad E_{22} = \frac{\Omega_2}{i\mu}, \qquad E_{23} = 1, \nonumber\\
E_{31} = \frac{\Omega_2\Omega_3+\mu\Omega_1}{2\mu^2}, \qquad E_{32} = \frac{\Omega_3}{i\mu}, \qquad E_{33} = -\frac{\Omega_2\Omega_3-\mu\Omega_3}{\Omega_1^2+\Omega_3^2},\label{definition of Ei in terms of omega}
\end{gather}
and $f(u)$ is determined by
$
\frac{\partial \Psi}{\partial u_i} = V_i\Psi$.
Here, it is convenient to make the following global change of coordinates
$
F\colon \mathbb{C}^3\setminus\{u_i=u_j\}\mapsto \mathbb{C}\times\mathbb{C}^{*}\times( \mathbb{C}\setminus\{0,1\})$, $ F(u_1,u_2,u_3)=(v,H,x)$,
where
\begin{gather}\label{global change of coordinates of canonical coordinates}
v=u_1+u_2+u_3, \qquad H=u_2-u_1, \qquad x=\frac{u_3-u_1}{u_2-u_1}.
\end{gather}
From the standard theory of Dubrovin--Frobenius manifolds \cite{B.Dubrovin2,B.Dubrovin3,Guzzetti}, the action of the unit vector field $e$ and Euler vector field $E$ in the matrix $V$ are given by
\begin{gather}\label{symmetries of matrix V}
\sum_{i=1}^3\partial_iV = 0, \qquad \sum_{i=1}^3u_i\partial_iV = 0.
\end{gather}
Hence, considering the change of coordinates \eqref{global change of coordinates of canonical coordinates} and the symmetries of the matrix $V$ in~\eqref{symmetries of matrix V}, we have the following
$
V(u_1,u_2,u_3) = V\big(\frac{u_3-u_1}{u_2-u_1}\big) = V(x)$.
Consequently, the system \eqref{compatibility of Dubrovin system in canonical coordinates} reduces to
\begin{gather}
\frac{{\rm d}\Omega_1}{{\rm d}x} = \frac{\Omega_2\Omega_3}{x}, \qquad
\frac{{\rm d}\Omega_2}{{\rm d}x} = \frac{\Omega_1\Omega_3}{1-x}, \qquad
\frac{{\rm d}\Omega_3}{{\rm d}x} = \frac{\Omega_1\Omega_2}{x(x-1)}.\label{compatibility condition simplified in omega}
\end{gather}

In the specific case of $QH^{*}\bigl(\mathbb{CP}^2\bigr)$, Guzzetti derived the Saito flat coordinates $\bigl(t^1,t^2,t^3\bigr)$ as functions of $(v,H,x)$ in \cite[Section 6.2.2]{Guzzetti}, specifically in his equations (6.7), (6.8), and (6.9)
\begin{gather}
t^1(u)=u_1+a(x)H,\qquad
t^2(u)=3\ln(H)+3\int^x \frac{{\rm d}\zeta}{\zeta+\frac{E_{21}E_{22}}{E_{31}E_{32}}},\nonumber\\
t^3(u)=-9\frac{c(x)}{b(x)^2H},\label{relation flat to canonical coordinates quantum cohomology}
\end{gather}
where
\begin{gather}
a(x)=E_{21}E_{23}+xE_{31}E_{33},\qquad
b(x)=E_{22}E_{21}+xE_{32}E_{31},\qquad
c(x)=E_{21}^2+xE_{31}^2.\label{definition of abc PVI}
\end{gather}

In \cite[Appendix E]{B.Dubrovin2}, Dubrovin proved that any semisimple 3D Dubrovin--Frobenius manifold is associated with a special solution of $P_{\rm VI}$. More concretely, consider the twisted Gauss--Manin connection
$
(U-\lambda I )\partial_{\beta}\chi+C_{\beta} \mu \chi=0$, $
(U-\lambda I )\partial_{\lambda}\chi-\mu \chi=0$.
In canonical coordinates, the corresponding $3\times 3$ system of differential equations is given by
\begin{gather}
\frac{{\rm d}X}{{\rm d}\lambda}=-\mu\left[ \frac{A_1}{\lambda-u_1} +\frac{A_2}{\lambda-u_2}+\frac{A_3}{\lambda-u_3} \right]X,\qquad
\frac{\partial X}{\partial u_i}=\mu \frac{A_i}{\lambda-u_i} X ,\label{Gauss Manin connection of PVI 0}
\end{gather}
where
\begin{gather*}
A_i=\begin{pmatrix}
\psi_{i1}\psi_{i3} &0& -\psi_{i3}^2 \\
\psi_{i1}\psi_{i2} &0& -\psi_{i2}\psi_{i3}\\
\psi_{i1}^2 &0& -\psi_{i1}\psi_{i3}\\
\end{pmatrix},
\qquad
\mu=\mu\begin{pmatrix}
1&0& 0 \\
0 &0& 0\\
0&0& -1\\
\end{pmatrix},
\qquad
X=\begin{pmatrix}
X_1\\
X_2 \\
X_3\\
\end{pmatrix},
\\
\frac{\partial u_i}{\partial t^{\alpha}}=\frac{\psi_{i\alpha}}{\psi_{i1}}.
\end{gather*}

The twisted Gauss--Manin connection \eqref{Gauss Manin connection of PVI 0} can be split into a $2\times 2$ system of differential equations due to the fact that the matrix $\mu=\operatorname{diag}(\mu,0,-\mu)$ has a $0$ eigenvalue. Indeed,
\begin{gather}
\frac{{\rm d}X}{{\rm d}\lambda}=-\mu\left[ \frac{A_1}{\lambda-u_1} +\frac{A_2}{\lambda-u_2}+\frac{A_3}{\lambda-u_3} \right]X,\qquad
\frac{\partial X}{\partial u_i}=\mu \frac{A_i}{\lambda-u_i} X ,\label{Gauss Manin connection of PVI}
\end{gather}
where
\begin{gather*}
A_i=\begin{pmatrix}
\psi_{i1}\psi_{i3} & -\psi_{i3}^3 \\
\psi_{i1}^2 & -\psi_{i1}\psi_{i3}\\
\end{pmatrix},
\qquad
\mu=\mu\begin{pmatrix}
1& 0 \\
0& -1\\
\end{pmatrix},
\qquad
X=\begin{pmatrix}
X_1\\
X_3\\
\end{pmatrix},
\qquad
\frac{\partial u_i}{\partial t^{\alpha}}=\frac{\psi_{i\alpha}}{\psi_{i1}}.
\end{gather*}

Using the new coordinates $q$, $p$ given by the relations
\begin{gather*}
0=\left[\sum_{i=1}^3 \frac{A_i}{q-u_i}\right]_{12},\qquad
p=\left[\sum_{i=1}^3 \frac{A_i}{q-u_i}\right]_{11},
\end{gather*}
the compatibility conditions of the twisted Gauss--Manin connection \eqref{Gauss Manin connection of PVI} are given by
\begin{gather*}
\frac{\partial q}{\partial u_i}=\frac{P(q)}{P^{\prime}(u_i)}\left[ 2p+\frac{1}{q-u_i}\right],\qquad
\frac{\partial p}{\partial u_i}=\frac{P^{\prime}(q)p^2+\bigl(2q+u_i-\sum_{i=1}^3 u_i\bigr)p +\mu(1-\mu)}{P^{\prime}(u_i)},
\end{gather*}
where
$
P(\lambda)=( \lambda-u_1)( \lambda-u_2)( \lambda-u_3)$.

Eliminating $p$ from the system, we obtain a second-order differential equation for the function~${q = q(u_1, u_2, u_3)}$. Using a change of coordinates
\smash{$
x=\frac{u_3-u_1}{u_2-u_1}$}, \smash{$
y=\frac{q-u_1}{u_2-u_1}$},
we have that~${y=y(x)}$ solves a particular one-parameter family of Painlev\'e ${\rm VI}$ equation, denoted by \smash{$P_{{\rm VI}_{\mu}}$},
\begin{align*}
y^{\prime\prime}={}&\frac{1}{2}\left[ \frac{1}{y} +\frac{1}{y-1}+\frac{1}{y-x} \right](y^{\prime})^2-\left[ \frac{1}{x} +\frac{1}{x-1}+\frac{1}{y-x} \right]y^{\prime}\nonumber\\
&+\frac{1}{2}\frac{y(y-1)(y-x)}{x^2(x-1)^2}\left[ (2\mu-1)^2+\frac{x(x-1)}{(y-x)^2}\right].
\end{align*}
The parameter $\mu$ coincides with the $\mu$ of the monodromy data $(\mu,R,S,C)$ of some 3D Dubrovin--Frobenius manifold.

\begin{Remark}\label{Remark Malgrange divisor}
A differential equation is said to possess the \textit{Painlev\'e property} if the problematic singularities of its corresponding solutions, such as essential singularities and branch points, are always at the coefficients of the differential equation. Additionally, the movable singularities—those dependent on the integration constants—are poles. In particular, the \textit{Painlev\'e~{\rm VI}} ($P_{\rm VI}$) equation has the Painlev\'e property. Consequently, the potential locations of essential singularities, branch points, and so on, are confined to 0, 1, and $\infty$, while the positions of the poles depend on the integration constants. In this context, the locus of movable poles of~\smash{$P_{{\rm VI}_{\mu}}$}, known as the Malgrange divisor $MD_{\mu}$, coincides with the zeros and poles divisor associated with the isomonodromic $\tau$ function defined in \eqref{isomonodromic tau function} projected onto $\mathbb{C}\setminus\{0, 1\}$. The functions~$a(x)$,~$b(x)$,~$c(x)$ in \eqref{definition of abc PVI} are rational functions of $x$, $y(x)$, $y^{\prime}(x)$, see \cite[Section~6.1]{Guzzetti}. Due to the Painlev\'e property, the Saito flat coordinates $t^1$, $t^2$, $t^3$ \eqref{relation flat to canonical coordinates quantum cohomology} are meromorphic functions in $\mathbb{C}\times \mathbb{C}^{*}\times (\mathbb{C}\setminus\{0, 1\})$.
\end{Remark}

From here, we return to the case of $QH^{*}\bigl(\mathbb{CP}^2\bigr)$. Note that the hypersurface $t^3 = 0$ in canonical coordinates corresponds to the hypersurface
\smash{$
\bigl\{\bigl(v, H, {\rm e}^{-\frac{\pi {\rm i}}{3}}\bigr) \in \mathbb{C}\times \mathbb{C}^{*}\times \mathbb{C}\setminus\{0, 1\}\bigr\}$}.
Indeed, substituting the small canonical coordinates \eqref{canonical coordinates of small quantum cohomology} in \eqref{global change of coordinates of canonical coordinates}, we obtain the desired result.

Hence, the big quantum cohomology is an expansion around \smash{$x = {\rm e}^{-\frac{\pi {\rm i}}{3}}$}. Following \cite[Section~6.5]{Guzzetti}, define \smash{$s := x - {\rm e}^{\frac{-\pi {\rm i}}{3}}$} and expand $\Omega_1$, $\Omega_2$, $\Omega_3$ in \eqref{compatibility condition simplified in omega} around \smash{${\rm e}^{\frac{\pi {\rm i}}{3}}$}
\begin{gather}
\Omega_1(x) = -\frac{{\rm i}\sqrt{3}}{2} - \left(\frac{1}{6} + \frac{{\rm i}\sqrt{3}}{6}\right)s + \frac{{\rm i}\sqrt{3}}{9}s^2 + O\bigl(s^3\bigr),\nonumber\\
\Omega_2(x) = \frac{{\rm i}\sqrt{3}}{2} + \left(\frac{1}{6} - \frac{{\rm i}\sqrt{3}}{6}\right)s - \frac{{\rm i}\sqrt{3}}{9}s^2 + O\bigl(s^3\bigr),\nonumber\\
\Omega_3(x) = \frac{{\rm i}\sqrt{3}}{2} - \frac{1}{3}s + \frac{2{\rm i}\sqrt{3}}{9}s^2 + O\bigl(s^3\bigr).\label{Omega equiharmonic expansion}
\end{gather}
Furthermore, expanding $y(x)$,
\[
y(x) = \frac{1}{2} - \frac{{\rm i}\sqrt{3}}{6} + \frac{1}{3}s - \frac{{\rm i}\sqrt{3}}{3}s^2 + O\bigl(s^3\bigr).
\]
Substituting \eqref{Omega equiharmonic expansion} in~\eqref{definition of Ei in terms of omega}, \eqref{definition of abc PVI} and \eqref{relation flat to canonical coordinates quantum cohomology}, the Saito flat coordinates have the following expansion:
\begin{gather*}
t^1(u) = u_1 + \left(\frac{1}{2} - \frac{{\rm i}\sqrt{3}}{6}s + O\bigl(s^2\bigr)\right)H,\qquad
Q(u) = \frac{{\rm i}\sqrt{3}}{143}\bigl(1 + {\rm i}\sqrt{3}s + O\bigl(s^2\bigr)\bigr)H^3,\\
t^3(u) = \bigl(-9s + O\bigl(s^2\bigr)\bigr)H^{-1}.
\end{gather*}
Then, the expansion of coordinates \smash{$Q^{\frac{1}{3}}t^3$} is given by
\begin{align}
Q^{\frac{1}{3}}t^3(x)& =\sum_{n=1}^{\infty} Q_n \bigl(x - {\rm e}^{-\frac{\pi {\rm i}}{3}}\bigr)^n\nonumber\\
&= \left(\frac{3}{2} - \frac{i\sqrt{3}}{2}\right)\bigl(x - {\rm e}^{\frac{-\pi {\rm i}}{3}}\bigr) - \left(\frac{1}{2} + \frac{i\sqrt{3}}{2}\right)\bigl(x - {\rm e}^{\frac{-\pi {\rm i}}{3}}\bigr)^2 + O\bigl(\bigl(x - {\rm e}^{\frac{-\pi {\rm i}}{3}}\bigr)^3\bigr)\label{tilde Q as function of x}
\end{align}
is a holomorphic function in some neighbourhood \smash{$U\bigl({\rm e}^{\frac{-\pi {\rm i}}{3}}\bigr)$} of \smash{${\rm e}^{\frac{-\pi {\rm i}}{3}}\in \mathbb{C}\setminus\{0, 1\}$}.

At this stage, we can state the main result of this subsection.

\begin{Lemma}\label{t3 geometric interpretation}
 The composition of the modular lambda function \eqref{modular lambda function 0} restricted to a suitable neighborhood of \smash{${\rm e}^{\frac{2\pi {\rm i}}{3}}$} with the map \eqref{tilde Q as function of x} is a local biholomorphism given by the following power series
 \begin{gather}
 Q^{\frac{1}{3}}t^3(z) = \sum_{n=0}^{\infty}\tilde Q_n \bigl(z - {\rm e}^{\frac{2\pi {\rm i}}{3}} \bigr)^n,\label{local biholomorphism tildeQ to z}
 \end{gather}
 where
\[
 \tilde Q_n = \frac{ \sum_{k=0}^{n} k! Q_k B_{n,k} (x_1,x_2,\dots,x_n)}{n!},
 \]
 $x_m$ are the coefficients of the modular lambda function~\eqref{modular lambda function 0} around \smash{${\rm e}^{\frac{2\pi {\rm i}}{3}}$}, i.e.,
 \begin{gather*}
 x(z) = \sum_{m=0}^{\infty} x_m \frac{\bigl( z - {\rm e}^{\frac{2\pi {\rm i}}{3}} \bigr)^m}{m!},
 \end{gather*}
 and $B_{n,k}$ are the partial Bell polynomials given by the formula \eqref{Bell polynomials}.
\end{Lemma}

\begin{proof}
The map defined by \eqref{tilde Q as function of x} is a local biholomorphism
\[
 Q^{\frac{1}{3}}t^3(x)\colon\ U\bigl({\rm e}^{\frac{-\pi {\rm i}}{3}}\bigr) \subset \mathbb{C}\setminus\{0,1\} \mapsto Q^{\frac{1}{3}}t^3\bigl( U\bigl({\rm e}^{\frac{-\pi {\rm i}}{3}}\bigr)\bigr) \subset D\left(0,\left(\frac{1}{a}\right)^{\frac{1}{3}}\right).
\]
Indeed, the point \smash{${\rm e}^{\frac{-\pi {\rm i}}{3}} \in \mathbb{C}\setminus\{0,1\}$} lies within the semisimple locus and serves as a holomorphic point for the functions given by \eqref{Omega equiharmonic expansion}.
 Since the Jacobian of the transition function between canonical and flat Saito coordinates
\[
 J= \det\left( \frac{\partial t^{\alpha}}{\partial u_i} \right) = \prod_{i=1}^3 \psi_{i1}
\]
does not vanish in the semisimple locus, the map \eqref{tilde Q as function of x} is a local biholomorphism.

Recall that the preimage of \smash{${\rm e}^{\frac{-\pi {\rm i}}{3}}$} under the modular lambda function \eqref{modular lambda function 0} is \smash{${\rm e}^{\frac{2\pi {\rm i}}{3}}$} mod $\Gamma(2)$. Moreover, the modular lambda function \eqref{modular lambda function 0} is a local biholomorphism. In particular, the modular lambda function \eqref{modular lambda function 0} is a local biholomorphism around \smash{${\rm e}^{\frac{2\pi {\rm i}}{3}}$}. Therefore, composing the modular lambda function \eqref{modular lambda function 0} restricted to a suitable neighborhood \smash{$U\bigl({\rm e}^{\frac{2\pi {\rm i}}{3}}\bigr)$} of \smash{${\rm e}^{\frac{2\pi {\rm i}}{3}}\in\mathbb{H}$} with \eqref{tilde Q as function of x}, we obtain the desire local biholomorphism
\begin{gather}
 Q^{\frac{1}{3}}t^3(x(z))\colon \ U\bigl({\rm e}^{\frac{2\pi {\rm i}}{3}}\bigr)\subset\mathbb{H}\mapsto U\bigl({\rm e}^{\frac{-\pi {\rm i}}{3}}\bigr) \subset\mathbb{C}\setminus\{0,1\} \nonumber \\
 \hphantom{Q^{\frac{1}{3}}t^3(x(z))\colon} \ \mapsto Q^{\frac{1}{3}}t^3\bigl(x\bigl( U\bigl({\rm e}^{\frac{2\pi {\rm i}}{3}}\bigr)\bigr)\bigr) \subset D\left(0,\left(\frac{1}{a}\right)^{\frac{1}{3}}\right).\label{local function of tilde Q and z}
\end{gather}
Furthermore, we use \eqref{tilde Q as function of x}, \eqref{modular lambda function 0} and using the Faa di Bruno formula \eqref{Faa di Bruno} to compute the explicit power series composition
\begin{align*}
 Q^{\frac{1}{3}}t^3(x(z)) &= \sum_{n=0}^{\infty} Q_n \bigl(x(z)-{\rm e}^{\frac{-\pi {\rm i}}{3}} \bigr)^n
 = \sum_{n=0}^{\infty} n!Q_n \frac{\bigl(x(z)-{\rm e}^{\frac{-\pi {\rm i}}{3}} \bigr)^n}{n!}\\
 &= \sum_{n=0}^{\infty}n! \frac{Q_n}{n!} \left( \sum_{m=0}^{\infty} x_m\frac{\bigl( z-{\rm e}^{\frac{2\pi {\rm i}}{3}} \bigr)^m}{m!} \right)^n\\
 &= \sum_{n=0}^{\infty} \left(\sum_{k=0}^{n} k!Q_k B_{n,k} (x_1,x_2,\dots,x_n)\right) \frac{\bigl( z-{\rm e}^{\frac{2\pi {\rm i}}{3}} \bigr)^n}{n!}\\
 &= \sum_{n=0}^{\infty} \tilde Q_n\bigl( z-{\rm e}^{\frac{2\pi {\rm i}}{3}} \bigr)^n.
\end{align*}
Lemma proved.
\end{proof}

\begin{Remark}
 The coefficients of \eqref{local biholomorphism tildeQ to z} can be explicitly computed to any desired order using the Guzzetti algorithm \eqref{tilde Q as function of x} or by employing the alternative algorithm derived in Appendix~\ref{Coefficients of the cross ratio function}.
\end{Remark}

\subsection[Canonical coordinates for big quantum cohomology of $\mathbb{CP}^2$]{Canonical coordinates for big quantum cohomology of $\boldsymbol{\mathbb{CP}^2}$}

In this subsection, the goal is to establish a formula expressing the canonical coordinates of quantum cohomology for $\mathbb{CP}^2$ in terms of the Saito flat coordinates.

The canonical coordinates are the roots \eqref{spectral curve}. In particular, $g^{\alpha\beta}$ and $\eta^{\alpha\beta}$ for big quantum cohomology are given by
\begin{gather}
g^{\alpha\beta}=E^{\epsilon}c^{\alpha\beta}_{\epsilon}
=\begin{pmatrix}\!
\frac{3}{\bigl(t^3\bigr)^3}[ 9\Phi^{\prime\prime}(X)-9\Phi^{\prime}(X) +2\Phi(X) ] \!& \!\frac{2}{\bigl(t^3\bigr)^2}[ 3\Phi^{\prime\prime}(X)-\Phi^{\prime}(X) ]\!\!&\!\! t^1\!\\
\frac{2}{\bigl(t^3\bigr)^2}[ 3\Phi^{\prime\prime}(X)-\Phi^{\prime}(X) ] & t^1+\frac{\Phi^{\prime\prime}(X) }{t^3}& 3\\
t^1& 3 &\! -t^3
\end{pmatrix},\label{intersection form in flar coordinates}
\\
\eta^{\alpha\beta}=\begin{pmatrix}
0 & 0& 1\\
0 & 1& 0\\
1&0 & 0
\end{pmatrix},\label{eta in flat coordinates}
\end{gather}
where $\Phi(X)$ is defined by \eqref{main Gromov Witten potential}.

In \cite[Example 3.10.34]{B.Dubrovin4}, Dubrovin and Zhang made the following ansatz:
\begin{gather}\label{Dubrovin Zhang substitution}
u_i=t^1+\frac{9+\Phi^{\prime\prime}(X)-z_i}{t^3}.
\end{gather}
Substituting \eqref{intersection form in flar coordinates}, \eqref{eta in flat coordinates} and \eqref{Dubrovin Zhang substitution} in \eqref{spectral curve}, we obtain that $z_1$, $z_2$, $z_3 $ are the roots of the following cubic equation
$z^3-s_1z^2+s_2z-s_3=0$,
where
\begin{gather}
s_1=27+2\Phi^{\prime\prime},\qquad
s_2=243+6\Phi-15\Phi^{\prime}+27\Phi^{\prime\prime}+\bigl(\Phi^{\prime\prime}\bigr)^2,\nonumber\\
s_3=\bigl(27+2\Phi^{\prime}-3\Phi^{\prime\prime}\bigr)^2.\label{canonical coordinates 2}
\end{gather}
Summarising, we have the following lemma.

\begin{Lemma}
The canonical coordinates of quantum cohomology $u_1$, $u_2$, $u_3$ as function of the Saito flat coordinates $t^1$, $t^2$, $t^3$ are given by
\begin{gather}\label{canonical coordinates lemma1}
u_k=t^1+\frac{1}{t^3}\sum_{n=1}^{\infty} A_n^k \bigl( Q^{\frac{1}{3}}t^3 \bigr)^n,
\end{gather}
where the coefficients $A_{k}$ are given explicitly by
\begin{gather*}
\tilde A_{3k}=\frac{k^2N_k}{(3k-1)! },\\
\sum_{n_2=2}^{3n}3\cos\left( \frac{2\pi (n_2-1)}{3}\right) \tilde A_{3n-n_2+1} \tilde A_{n_2-1}=\bigl(6-15n-9n^2 \bigr)\frac{N_n}{(3n-1)!},\\
\sum_{n_2=1}^{3n-2}\sum_{n_3=1}^{3n-n_2-1}3\cos\left( \frac{2\pi (n_2+2n_3 )}{3} \right) \tilde A_{3n-n_2-n_3} \tilde A_{n_2} \tilde A_{n_3}=(54-243n )\frac{N_n}{(3n-1)!}+\delta_{n} ,
\end{gather*}
where
\begin{gather*}
A_n^k=\tilde A_n \bigl( {\rm e}^{\frac{2\pi {\rm i}}{3}}\bigr)^{nk},
\qquad
\delta_n=
\begin{cases}
0& \text{if } n=1,\\
 \tilde\delta_n & \text{otherwise},
\end{cases}
\\
\tilde\delta_n=\sum_{n_2=2}^{n}\frac{\bigl(6(n_2-1)-3(n-n_2+1)(n_2-1)^2 \bigr)}{(3n-3n_2+2)!(3n_2-4)!} N_{n-n_2+1}N_{n_2-1}\\
\phantom{\tilde\delta_n=}{}+\sum_{n_2=2}^{n}\frac{(-4(n-n_2+1)(n_2-1) )}{(3n-3n_2+2)!(3n_2-4)!} N_{n-n_2+1}N_{n_2-1}\\
\phantom{\tilde\delta_n=}{}+\sum_{n_2=2}^{n}\frac{(-9(n_1-n_2+1)^2(n_2-1)^2 )}{(3n-3n_2+2)!(3n_2-4)!} N_{n-n_2+1}N_{n_2-1}.
\end{gather*}
\end{Lemma}

\begin{proof}
See Appendix~\ref{Canonical coordinates as functions of Gromov Witten invariants}.
\end{proof}

Moreover, the transition function matrix $\psi_{i \alpha}$ can be obtained by finding the eigenvectors related to \eqref{spectral curve} as done in \cite[Example 3.10.34]{B.Dubrovin4}. More precisely, the transition function matrix $\psi_{i \alpha}$ is the matrix such that the columns are eigenvectors of
$g^{*}\eta$
which is given by
\begin{gather*}
\psi_{i1}=h_i,\qquad
\psi_{i2}=\frac{h_i\bigl( -27-2\Phi^{\prime}+3\Phi^{\prime\prime} +3z_i \bigr)}{z_it^3},\\
\psi_{i3}=\frac{h_i\bigl( 81+6\Phi^{\prime}-9\Phi^{\prime\prime}-18z_i+z_i^2-z_i^2\Phi^{\prime\prime} \bigr)}{z_i\bigl(t^3\bigr)^2},
\end{gather*}
where $h_i$ have the form
\begin{gather*}
h_1=\frac{t^3\sqrt{z_1}}{\sqrt{(z_1-z_2)}\sqrt{(z_1-z_3)}}, \qquad h_2=\frac{t^3\sqrt{z_2}}{\sqrt{(z_2-z_1)}\sqrt{(z_2-z_3)}}, \\ h_3=\frac{t^3\sqrt{z_3}}{\sqrt{(z_3-z_1)}\sqrt{(z_3-z_2)}}.
\end{gather*}

\subsection{Geometric isomonodromic deformation}\label{Geometric Isomonodromic deformation}

The aim of this subsection is to give a geometric interpretation for the isomonodromic deformation of the Hauptmodul $\gamma_2(\tilde\tau)$~\eqref{Weber functions}. For this purpose, we derive a~$t^3$ deformation of the equianharmonic elliptic curve
\[
y^2=4(x-u_1)(x-u_2)(x-u_3)=4\bigl(x-t^1\bigr)^3-27Q.
\]

\begin{Lemma}
Consider the canonical coordinates $(u_1, u_2, u_3)$ of the quantum cohomology \linebreak of~$\mathbb{CP}^2$, obtained from \eqref{Dubrovin Zhang substitution} and \eqref{canonical coordinates 2}, along with the cubic equation
\begin{gather}\label{spectral curve brut form}
y^2 = 4( \lambda - u_1)( \lambda - u_2)( \lambda - u_3).
\end{gather}
Then, the pair $(\lambda,y)$ satisfying \eqref{spectral curve brut form} can be expressed as follows:
\begin{gather}
\lambda=t^1+\frac{\Phi^{\prime\prime}\bigl(Q\bigl(t^3\bigr)^3\bigr)}{3t^3}-\frac{\wp(v,z)}{(2\tilde\omega)^2},\qquad
y=i\frac{\wp^{\prime}(v,z)}{(2\tilde\omega)^3}.\label{LG superpotential elliptic form first piece}
\end{gather}
\end{Lemma}
\begin{proof}
Substituting \eqref{Dubrovin Zhang substitution} in \eqref{spectral curve brut form}, we obtain
\begin{gather}
\tilde y^2=4\bigl( \tilde\lambda-z_1\bigr)\bigl( \tilde\lambda-z_2\bigr)\bigl( \tilde\lambda-z_3\bigr)
=4\bigl( \tilde\lambda^3-s_1\tilde\lambda^2+s_2\tilde\lambda-s_3\bigr),\label{spectral curve brut form part 1}
\end{gather}
where $s_1$, $s_2$, $s_3$ are given by \eqref{canonical coordinates 2} and
\begin{gather}
 y = \frac{i\tilde y}{ \bigl(t^3\bigr)^{\frac{3}{2}}},\qquad
\lambda=t^1+\frac{9+\Phi^{\prime\prime}-\tilde\lambda }{t^3}.\label{first substitution 1}
\end{gather}
To express \eqref{spectral curve brut form part 1} in the depressed form of the cubic equation, we make in \eqref{spectral curve brut form part 1} the following substitution:
\begin{gather}
\tilde\lambda=\hat\lambda +\frac{s_1}{3},\qquad
\hat y= \tilde y,\qquad
z_i=\hat e_i+\frac{s_1}{3},\label{substitution hat}
\end{gather}
then we obtain
\begin{gather}
\hat y^2=4\bigl( \hat\lambda-\hat e_1\bigr)\bigl( \hat\lambda-\hat e_2\bigr)\bigl( \hat\lambda-\hat e_3\bigr)=4\hat\lambda^3-g_2\hat\lambda-g_3,\label{spectral curve brut form depressed form}
\end{gather}
where
$
\hat g_2=-4( \hat e_1\hat e_2+\hat e_1\hat e_3+\hat e_2\hat e_3)$, $
\hat g_3=4\hat e_1\hat e_2\hat e_3$.
Moreover, consider the following rescaling:
\begin{gather}
\hat\lambda=\frac{\lambda^{r}}{(2\tilde\omega)^2}, \qquad \hat e_i=\frac{e_i}{(2\tilde\omega)^2} , \qquad \hat y= \frac{y^r}{(2\tilde\omega)^3}.\label{spectral curve brut form depressed form rescalling}
\end{gather}
Hence, substituting \eqref{spectral curve brut form depressed form rescalling} in \eqref{spectral curve brut form depressed form},
\begin{gather}
 (y^r)^2=4( \lambda^r- e_1)( \lambda^r- e_2)( \lambda^r- e_3)=4(\lambda^r)^3-g_2\lambda^r-g_3,\label{spectral curve brut form depressed form 1}
\end{gather}
where
$
g_2=-4( e_1e_2+e_1e_3+e_2e_3)$, $
 g_3=4 e_1 e_2 e_3$.

Due to the uniformization of elliptic curves, we can parametrize the cubic equation \eqref{spectral curve brut form depressed form 1} using the Weierstrass $\wp$ function and its derivatives. In other words, we express \eqref{spectral curve brut form depressed form 1} in the form
\begin{gather}\label{spectral curve brut form depressed form uniformization}
(\wp^{\prime}(v,z))^2=4\wp(v,z)^3-g_2(z)\wp(v,z)-g_3(z).
\end{gather}

Substituting \eqref{spectral curve brut form depressed form uniformization}, \eqref{spectral curve brut form depressed form rescalling}, \eqref{substitution hat} in \eqref{first substitution 1}, we obtain
\begin{gather*}
\lambda=t^1+\frac{\Phi^{\prime\prime}\bigl(Q\bigl(t^3\bigr)^3\bigr)}{3t^3}-\frac{\wp(v,z)}{(2\tilde\omega)^2},\qquad
y={\rm i}\frac{\wp^{\prime}(v,z)}{(2\tilde\omega)^3}.
\end{gather*}
Lemma proved.
\end{proof}

Our objective is to extend the family of elliptic curves \eqref{LG superpotential elliptic form first piece} to a LG superpotential for $QH^{*}\bigl(\mathbb{CP}^2\bigr)$. To begin, we examine the following auxiliary lemma.

\begin{Lemma}
The Landau--Ginzburg superpotential of quantum cohomology of $\mathbb{CP}^2$ is a family of functions
\smash{$
\lambda\bigl(\tilde\tau_{12}, t^1, Q^{\frac{1}{3}}, z\bigr)\colon \mathbb{H} \mapsto \mathbb{C}
$}
parametrized by
\begin{gather*}
\bigl(t^1,Q^{\frac{1}{3}},z\bigr)\in 	\bigl\{ \bigl(t^1,Q^{\frac{1}{3}},z\bigr) \in \mathbb{C}\times\mathbb{C}^{*}\times U\bigl({\rm e}^{\frac{2\pi {\rm i}}{3}}\bigr) \bigr\}\subset \mathbb{C}\times\mathbb{C}^{*}\times\mathbb{H}
\end{gather*}
and given by
\smash{$
\lambda\bigl(\tilde\tau_{12}, t^1, Q, z\bigr)=t^1+3Q^{\frac{1}{3}}J^{\frac{1}{3}}(\tilde\tau_{12}, z)$},
where \smash{$J^{\frac{1}{3}}(\tilde\tau_{12}, z)$} is the composition of the function \smash{$J^{\frac{1}{3}}\bigl(\tilde\tau_{12}, Q^{\frac{1}{3}}t^3\bigr)$} defined in \eqref{definition of t3 deformation of gamma2} and \smash{$Q^{\frac{1}{3}}t^3(z)$} defined in \eqref{local function of tilde Q and z}, i.e.,
\begin{align*}
J^{\frac{1}{3}}(\tilde\tau_{12}, z)&:=\sum_{n=0}^{\infty}J_n^{\frac{1}{3}}(\tilde\tau_{12})\bigl( Q^{\frac{1}{3}}t^3(z)\bigr)^n\\
&=\sum_{n=0}^{\infty}\left( \sum_{k=0}^{n} k! J_k^{\frac{1}{3}}(\tilde\tau_{12}) B_{n,k}\bigl(1!\tilde Q_1,\dots,m!\tilde Q_m\bigr) \right)\frac{\bigl(z-{\rm e}^{\frac{2\pi {\rm i}}{3}} \bigr)^n}{n!},
\end{align*}
where $B_{n,k}$ are the Bell polynomials defined in \eqref{Bell polynomials}.
 In addition, the correspondent Abelian differential $\phi$ is given by
\begin{gather*}
\phi=\frac{\Delta^{\frac{1}{6}}(\tilde\tau_{12},z)}{Q^{\frac{1}{6}}} {\rm d}\tilde\tau_{12},
\end{gather*}
where $\Delta(\tilde\tau_{12},z)$ is the composition of \eqref{local function of tilde Q and z} with \eqref{modular deformation tilde 0}.
\end{Lemma}

\begin{proof}
Doing the change of coordinates \eqref{local function of tilde Q and z} in the Landau--Ginzburg superpotential \eqref{lg superpotential of small quantum cohomology j function big quantum cohomology main body} and in Abelian differential \eqref{volume form of Big quantum cohomology 1 Milanov deformation section} we get the desired result. More explicitly, substituting \eqref{local function of tilde Q and z} in \eqref{definition of t3 deformation of gamma2}, we obtain
\begin{align}
\lambda\bigl(\tilde\tau_{12},t^1,Q,t^3\bigr)&=t^1+3Q^{\frac{1}{3}}\sum_{n=0}^{\infty}J_n^{\frac{1}{3}}(\tilde\tau_{12})\bigl( Q^{\frac{1}{3}}t^3\bigr)^n\nonumber\\
&=t^1+3Q^{\frac{1}{3}}\sum_{n=0}^{\infty}n!J_n^{\frac{1}{3}}(\tilde\tau_{12})\frac{\bigl( \sum_{k=0}^{\infty}\tilde Q_k \bigl(z-{\rm e}^{\frac{2\pi {\rm i}}{3}} \bigr)^k \bigr)^n}{n!}\nonumber\\
&=t^1+3Q^{\frac{1}{3}}\sum_{n=0}^{\infty}\left( \sum_{k=0}^{n} k! J_k^{\frac{1}{3}}(\tilde\tau_{12}) B_{n,k}\bigl(1!\tilde Q_1,\dots,m!\tilde Q_m\bigr) \right)\frac{\bigl(z-{\rm e}^{\frac{2\pi {\rm i}}{3}} \bigr)^n}{n!}\nonumber\\
&=t^1+3Q^{\frac{1}{3}}J^{\frac{1}{3}}(\tilde\tau_{12},z).\label{Lg superpotential big quantum cohomology geo iso dif 12}
\end{align}
That concludes the proof.
\end{proof}

Our next goal is to generalise the Cohn identity \eqref{universal covering map equiharmonic lattice} for $z \in \mathbb{H}$ near \smash{${\rm e}^{\frac{2\pi {\rm i}}{3}}$}.

\begin{Theorem}
The Landau--Ginzburg superpotential of big quantum cohomology of $\mathbb{CP}^2$ is a~family of functions
$
\lambda(\tilde\tau_{12}, \tau_{12}, \tilde\omega,z)\colon \mathbb{H} \mapsto \mathbb{C}
$
parametrized by
\begin{gather*}
\bigl\{(\tau_{12},\tilde\omega,z) \in \mathbb{H}\times\mathbb{C}^{*}\times\mathbb{H}\mid z \in U\bigl({\rm e}^{\frac{2\pi {\rm i}}{3}}\bigr) \bigr\}
\end{gather*}
and given by
\begin{gather*}
\lambda(v(\tilde\tau_{12},z),\tau_{12},\tilde\omega,z)=\frac{\wp(v(\tau_{12},z),z)}{\left(2\tilde\omega\right)^2}-\frac{\wp(v(\tilde\tau_{12},z),z)}{\left(2\tilde\omega\right)^2},
\end{gather*}
where for $z\in \mathbb{H}$ close enough to \smash{${\rm e}^{\frac{2\pi {\rm i}}{3}}$}, $v(\tilde\tau_{12},z)$ is the universal covering of $\mathbb{C}\setminus\{\mathbb{Z}\oplus z\mathbb{Z} \}$.
 In addition, the correspondent Abelian differential $\phi$ is given by
$
\phi= 2\tilde\omega {\rm d}v(\tilde\tau_{12},z)$,
where $\Delta(\tilde\tau_{12},z)$ is the composition of \eqref{local function of tilde Q and z} with \eqref{modular deformation tilde 0}.
\end{Theorem}

\begin{proof}
Recall that the critical values of the uniformization \eqref{LG superpotential elliptic form first piece} are the canonical coordinates~\eqref{canonical coordinates lemma1} by construction. The Abelian differential which gives the correct Landau--Ginzburg superpotential is $\phi={\rm d}\tilde w_2$. But, in order to implement the Abelian differential in the spectral curve \eqref{spectral curve brut form}, we need to construct, for a fixed $z\in\mathbb{H}$, a change of coordinates
\begin{gather}\label{family of universal coverings of torus}
v_z\colon \ \mathbb{H}\mapsto \mathbb{C}\setminus{\mathbb{Z}\oplus z\mathbb{Z}}, \qquad \tilde\tau_{12} \mapsto v_z(\tilde\tau_{12}):=v(\tilde\tau_{12},z)
\end{gather}
such that
\begin{gather}
\lambda\bigl(\tilde\tau_{12},t^1,Q,z\bigr)=t^1+Q^{\frac{1}{3}}\frac{\Phi^{\prime\prime}\bigl(Q\bigl(t^3\bigr)^3(z) \bigr)}{3Q^{\frac{1}{3}}t^3(z)}-\frac{\wp(v(\tilde\tau_{12},z),z)}{(2\tilde\omega)^2}.\label{universal covering map equiharmonic lattice t3 deformation part 3}
\end{gather}
Then, the change of coordinates \eqref{family of universal coverings of torus} is concretely constructed by comparing \eqref{universal covering map equiharmonic lattice t3 deformation part 3} with \eqref{Lg superpotential big quantum cohomology geo iso dif 12}. More explicitly, by the relation below
\begin{gather}
3Q^{\frac{1}{3}}J^{\frac{1}{3}}(\tilde\tau_{12},z)=Q^{\frac{1}{3}}\frac{\Phi^{\prime\prime}\bigl(Q\bigl(t^3\bigr)^3(z) \bigr)}{3Q^{\frac{1}{3}}t^3(z)}-\frac{\wp(v(\tilde\tau_{12},z),z)}{(2\tilde\omega)^2},\nonumber\\
\frac{\Delta^{\frac{1}{6}}(\tilde\tau_{12},z)}{Q^{\frac{1}{6}}}=2\tilde\omega v^{\prime}(\tilde\tau_{12},z).\label{universal covering map equiharmonic lattice t3 deformation part 4}
\end{gather}
The identity \eqref{universal covering map equiharmonic lattice t3 deformation part 4} can be understood as generalised Cohn identity. Indeed, the map \eqref{family of universal coverings of torus}
is defined by a composition of the deformed $J^{\frac{1}{3}}$ function \eqref{Lg superpotential big quantum cohomology geo iso dif 12} restricted to $\mathbb{H}\setminus{{\rm SL}_2(\mathbb{Z})({\rm i})}$ and the inverse of the Weierstrass $\wp$ function restricted to $\mathbb{C}\setminus\{e_1(z),e_2(z),e_3(z) \}$, for fixed $z$ close enough to \smash{${\rm e}^{\frac{2\pi {\rm i}}{3}}$}. More precisely, consider the functions
\begin{gather*}
J_z^{\frac{1}{3}}\colon\ \mathbb{H}\setminus{{\rm SL}_2(\mathbb{Z})(i)}\mapsto \mathbb{C}\setminus\{e_1^{*}(z),{\rm e}^{*}_2(z),{\rm e}^{*}_3(z) \}, \qquad \tilde\tau_{12}\mapsto J^{\frac{1}{3}}(\tilde\tau_{12},z),\\
\wp\colon\ ( \mathbb{C}\setminus{\mathbb{Z}\oplus z\mathbb{Z}})\setminus\left\{ \frac{\mathbb{Z}}{2}\oplus z\frac{\mathbb{Z}}{2} \right\}\mapsto \mathbb{C}\setminus\{e_1(z),e_2(z),e_3(z) \}, \qquad v\mapsto \wp(v,z),
\end{gather*}
where
\begin{gather*}
e_i^{*}(z)=\frac{\Phi^{\prime\prime}\bigl(Q\bigl(t^3\bigr)^3(z) \bigr)}{9Q^{\frac{1}{3}}t^3(z)}-4^{\frac{1}{3}}e_i(z).
\end{gather*}
Then, for small enough $z$, the composition
\begin{gather}
v(\tilde\tau_{12},z)\colon\ \mathbb{H}\setminus\{ E_6=0\}\mapsto( \mathbb{C}\setminus{\mathbb{Z}\oplus z\mathbb{Z}})\setminus\left\{ \frac{\mathbb{Z}}{2}\oplus z\frac{\mathbb{Z}}{2} \right\}\label{restricted family of universal coverings}
\end{gather}
is a holomorphic surjective function, because it is a composition of holomorphic surjective functions. Moreover, for $z$ close enough to \smash{${\rm e}^{\frac{2\pi {\rm i}}{3}}$}, the derivative of the function \eqref{restricted family of universal coverings}
\begin{gather*}
v^{\prime}(\tilde\tau_{12},z)=v^{\prime}(\tilde\tau_{12})+\sum_{n=1}^{\infty} v_n(\tau_{12})\bigl(z-{\rm e}^{\frac{2\pi {\rm i}}{3}} \bigr)^{n}
\end{gather*}
is close to its leading term
\smash{$
 v^{\prime}(\tilde\tau_{12})=\Delta^{\frac{1}{6}}( \tilde\tau_{12})$},
which is a non-vanishing function on $\mathbb{H}$. Therefore, we can extend \eqref{restricted family of universal coverings} to a local biholomorphism on $\mathbb{H}$
\begin{gather}\label{family of universal coverings in theorem}
v(\tilde\tau_{12},z)\colon\ \mathbb{H}\mapsto\ \mathbb{C}\setminus{(\mathbb{Z}\oplus z\mathbb{Z})}.
\end{gather}
The function \eqref{family of universal coverings in theorem} is a $\pi_1( \mathbb{C}\setminus{(\mathbb{Z}\oplus z\mathbb{Z}}))$-invariant function. Indeed, let $\gamma \in \Gamma^{(3)}$ and $\tilde\tau_{12}(\tilde\tau)$ the function defined in \eqref{affine in lambda t3 deformation of periods new coordinates z2 over z3 lemma}, then consider the group homomorphism defined by
\begin{gather}\label{Milanov group homomorphism}
\gamma_z \tilde\tau_{12}(\tilde\tau):=\tilde\tau_{12}(\gamma\tilde\tau).
\end{gather}
Since \eqref{affine in lambda t3 deformation of periods new coordinates z2 over z3 lemma} is an isomonodromic deformation equation, the following property holds:
\begin{gather*}
J^{\frac{1}{3}}(\tilde\tau_{12}(\tilde\tau),z)=J^{\frac{1}{3}}(\tilde\tau).
\end{gather*}
In addition,
\[
J^{\frac{1}{3}}(\gamma_z\tilde\tau_{12}(\tilde\tau),z)=J^{\frac{1}{3}}(\tilde\tau_{12}(\gamma\tilde\tau),z)=J^{\frac{1}{3}}(\gamma\tilde\tau)=J^{\frac{1}{3}}(\tilde\tau).
\]
The group homomorphism \eqref{Milanov group homomorphism} induces another group homomorphism
\begin{gather}\label{Milanov group homomorphism 2}
\bar\gamma_z v(\tilde\tau_{12}(\tilde\tau),z ):=v(\tilde\tau_{12}(\gamma\tilde\tau), z).
\end{gather}

Denoting the image of \eqref{Milanov group homomorphism} by $\Gamma^{(3)}_z$ and $\operatorname{ Im}\Psi_z$, $\operatorname{Ker}\Psi_z$ the image and kernel of~\eqref{Milanov group homomorphism 2}, respectively, we have that
\begin{gather*}
\mathbb{H}/\Gamma^{(3)}_z =(\mathbb{H}/ \operatorname{Ker}\Psi_z) /\operatorname{Im}\Psi_z= \mathbb{C}\setminus{(\mathbb{Z}\oplus z\mathbb{Z})}/\operatorname{Im}\Psi_z,
\\
\mathbb{H}/ \operatorname{Ker}\Psi_z= \mathbb{C}\setminus{(\mathbb{Z}\oplus z\mathbb{Z})}=\mathbb{H}/ \pi_1( \mathbb{C}\setminus{(\mathbb{Z}\oplus z\mathbb{Z})}).
\end{gather*}
Hence, the map \eqref{family of universal coverings in theorem} is a~local biholomorphism, which is $\pi_1( \mathbb{C}\setminus{(\mathbb{Z}\oplus z\mathbb{Z}}))$-invariant. Then, the map \eqref{family of universal coverings in theorem} can be identified with the quotient map
$
\pi\colon \mathbb{H}\mapsto \mathbb{H}/ \pi_1( \mathbb{C}\setminus{\mathbb{Z}\oplus z\mathbb{Z}})$,
which is a~covering map, since $\pi_1( \mathbb{C}\setminus{\mathbb{Z}\oplus z\mathbb{Z}})$ acts properly discontinuously on $\mathbb{H}$.

 Therefore, \eqref{family of universal coverings of torus} is a family of universal covering of $\mathbb{C}\setminus{\mathbb{Z}\oplus z\mathbb{Z}}$, which is an isomonodromic deformation of the universal covering of the equianharmonic lattice.

Moreover, the correspondent Abelian differential $\phi={\rm d}\tilde w_2$ in these coordinates is given by
\begin{align*}
\phi={\rm d}\tilde w_2=\tilde y {\rm d}\tilde\tau_{12}=\frac{\Delta^{\frac{1}{6}}(\tilde\tau_{12},z)}{Q^{\frac{1}{6}}}{\rm d}\tilde\tau_{12}=2\tilde\omega\frac{\partial v(\tilde\tau_{12},z)}{\partial \tilde\tau_{12}} {\rm d}\tilde\tau_{12}=2\tilde\omega{\rm d}v(\tilde\tau_{12},z).
\end{align*}
Theorem proved.
\end{proof}

\begin{Remark}
The title of this subsection was inspired by Doran's paper \cite{Doran}, which utilizes the moduli space of elliptic curves over $\mathbb{CP}^1$ to construct algebraic-geometric solutions of Painlev\'e~VI. In the context of this manuscript, we also have a Painlev\'e~VI associated with isomonodromic deformation induced by the Gauss--Manin connection of $QH^{*}\bigl(\mathbb{CP}^2\bigr)$ as mentioned in Section~\ref{From canonical to flat coordinates}. It would be interesting to compare both constructions in a more general setting, since geometric isomonodromic deformation coming from Gromov--Witten potential could give a good source of interesting moduli space of elliptic curves over $\mathbb{CP}^1$.
\end{Remark}

\appendix
\section{Appendix}

\subsection{Composition of power series and Faa di Bruno formula}

In this appendix, we state the necessary definitions and results on composition of power series, referencing \cite[Chapter~11]{Charalambides}.

\begin{Definition}[{\cite{Charalambides}}]
The polynomials $B_n(x_1,x_2,\dots,x_n)$, $B_{n,k}(x_1,x_2,\dots,x_n)$ in the variables~${x_1,x_2,\dots,x_n}$, defined by the sum
\begin{gather}
B_{n}=\sum_{ k_1+2k_2+3k_3+\dots +nk_{n}=n} \frac{n!}{k_1!(1!)k_2!(2!)^{k_2}\cdots k_n!(n!)^{k_n}} x_1^{k_1}x_2^{k_2}\cdots x_n^{k_n},\nonumber\\
B_{n,k}=\sum_{\substack{ k_1+\dots +k_n=k\\ k_1+2k_2+3k_3+\dots +nk_{n}=n}} \frac{n!}{k_1!(1!)k_2!(2!)^{k_2}\cdots k_n!(n!)^{k_n}} x_1^{k_1}x_2^{k_2}\cdots x_n^{k_n}\label{Bell polynomials}
\end{gather}
 are called exponential Bell partition polynomial and partial Bell partition polynomial respectively.
\end{Definition}

\begin{Theorem}[Faa di Bruno formula, \cite{Charalambides}]
Let $f(u)$ and $g(t)$ be two functions of real variables for which all the derivatives
\begin{gather*}
g_r=\left[\frac{{\rm d}^rg(t)}{{\rm d}t^r}\right]_{t=a}, \qquad r=0,1,\dots, \qquad f_k=\left[\frac{{\rm d}^kf(u)}{{\rm d}u^k}\right]_{u=(g(a)}, \qquad k=0,1,\dots,
\end{gather*}
exist. Then the derivatives of the composite function $h(t)= f(g(t))$,
\begin{gather*}
h_n=\left[\frac{{\rm d}^nh(u)}{{\rm d}t^n}\right]_{t=a}, \qquad n=0,1,\dots,
\end{gather*}
are given by
\begin{gather}
h_n=\sum_{k=0}^n f_kB_{n,k}(g_1,g_2,\dots,g_n)=B_n(fg_1,fg_2,\dots,fg_n).\label{Faa di Bruno}
\end{gather}
\end{Theorem}

\begin{Definition}[{\cite{Charalambides}}]
The polynomial $C_{n,s}=C_{n,s}(x_1,x_2,\dots,x_n)$ in the variables $x_1,x_2,\dots,\allowbreak x_n$, defined for a real (or complex) number s by the sum
\begin{gather*}
C_{n,s}=\sum_{ k_1+2k_2+3k_3+\dots +nk_{n}=n} \frac{n!}{k_1!(1!)k_2!(2!)^{k_2}\cdots k_n!(n!)^{k_n}}{{s}\choose{k}} x_1^{k_1}x_2^{k_2}\cdots x_n^{k_n}
\end{gather*}
 is called potential partition polynomial,
 where
 \begin{gather*}
 {{s}\choose{k}}=\frac{s(s-1)(s-2)\cdots (s-k+1))}{k!}.
 \end{gather*}
\end{Definition}

\begin{Theorem}[\cite{Charalambides}]\label{theorem potential partition polynomials}
The generating function of the potential partition polynomials $C_{n,s}(x_1,\dots,\allowbreak x_n)$, $n=0,1,\dots$, for fixed $k$, is given by
 \[
C_{s}(t)=\sum_{n=0}^{\infty} C_{n,s}(x_1,\dots, x_n)\frac{t^n}{n!}=[ 1+(g(t)-x_0) ]^s ,
\]
where
$g(t)=\sum_{r=0}^{\infty} x_r \frac{t^r}{r!}$.
\end{Theorem}

\subsection{Canonical coordinates as functions of Gromov--Witten invariants}\label{Canonical coordinates as functions of Gromov Witten invariants}

\begin{Lemma}\label{lemma of exp 2pi 3}
Let be \smash{$\zeta={\rm e}^{\frac{2\pi {\rm i}}{3}}$} and
\begin{gather*}
b_{k_1,k_2}:=\zeta^{k_1}+\zeta^{k_1+2k_2}+\zeta^{2k_1},\qquad
c_{n_1,n_2,n_3}=\zeta^{n_1+2n_2},\\
\tilde c_{n_1,n_2,n_3}:=c_{n_1,n_2,n_3}+c_{n_1,n_3,n_2}+c_{n_3,n_2,n_1}+c_{n_3,n_1,n_2}+c_{n_2,n_1,n_3}+c_{n_2,n_3,n_1}
\end{gather*}
for $k,k_1,k_2 k_3 \in \mathbb{Z}$. Then the following identities holds:
\begin{gather}
\zeta^{k}+\zeta^{2k}+\zeta^{3k}=1+2\cos\left(\frac{2\pi k}{3} \right),\qquad
b_{k,k}=1+2\cos\left(\frac{2\pi k}{3} \right),\qquad
c_{k,k,k}=1,
\nonumber\\
 b_{k_1-k_2,k_2}+b_{k_2,k_1-k_2}=
\begin{cases}
 6\cos\bigl( \frac{2\pi k_2}{3}\bigr)& \text{if } k_1\in 3\mathbb{Z},\\
 0 & \text{otherwise},
\end{cases}\nonumber
\\
\tilde c_{k_1-k_2-k_3,k_2,k_3}=
\begin{cases}
 6\cos\bigl( \frac{2\pi (k_2+2k_3)}{3}\bigr)& \text{if } k_1\in 3\mathbb{Z},\\
 0 & \text{otherwise}.
\end{cases}\label{exp 2pi 3 id}
\end{gather}
\end{Lemma}

\begin{proof}
\[
\zeta^{k}+\zeta^{2k}+\zeta^{3k}=\zeta^{k}+\zeta^{2k}+1=2\left(\frac{{\rm e}^{\frac{2\pi {\rm i}k}{3}}+{\rm e}^{\frac{-2\pi {\rm i}k}{3}}}{2} \right)+1=1+2\cos\left(\frac{2\pi k}{3} \right).
\]

Define $b_{k_1,k_2}$,
$
b_{k_1,k_2}:=\zeta^{k_1}+\zeta^{k_1+2k_2}+\zeta^{2k_1}$.
Then,
\[
b_{k,k}=\zeta^{k}+\zeta^{3k}+\zeta^{2k}=1+2\cos\left(\frac{2\pi k}{3} \right).
\]
Moreover,
\[
b_{k_1-k_2,k_2}=\zeta^{k_1-k_2}+\zeta^{k_1+k_2}+\zeta^{2(k_1-k_2)},\qquad
b_{k_2,k_1-k_2}=\zeta^{k_2}+\zeta^{2k_1-k_2}+\zeta^{2k_2}.
\]

If $k_1 \in 3\mathbb{Z}$,
\begin{align*}
b_{k_1-k_2,k_2}+b_{k_2,k_1-k_2}&=\zeta^{-k_2}+\zeta^{k_2}+\zeta^{-2k_2}+\zeta^{k_2}+\zeta^{-k_2}+\zeta^{2k_2}\\
&=\zeta^{2k_2}+\zeta^{k_2}+\zeta^{k_2}+\zeta^{k_2}+\zeta^{2k_2}+\zeta^{2k_2}\\
&=3\bigl( \zeta^{2k_2}+\zeta^{k_2}\bigr)=6\cos\left( \frac{2\pi k_2}{3}\right).
\end{align*}

If $k_1 \in 1+3\mathbb{Z}$,
\begin{align*}
b_{k_1-k_2,k_2}+b_{k_2,k_1-k_2}&=\zeta^{1-k_2}+\zeta^{k_2+1}+\zeta^{2-2k_2}+\zeta^{k_2}+\zeta^{2-k_2}+\zeta^{2k_2}\\
&=\zeta^{1+2k_2}+\zeta^{k_2+1}+\zeta^{k_2+2}+\zeta^{k_2}+\zeta^{2(k_2+1)}+\zeta^{2k_2}\\
&=2\cos\left( \frac{2\pi (1-k_2)}{3}\right)+2\cos\left( \frac{2\pi k_2}{3}\right)+2\cos\left( \frac{2\pi (1+k_2)}{3}\right)\\
&=0.
\end{align*}

If $k_1 \in 2+3\mathbb{Z}$,
\begin{align*}
b_{k_1-k_2,k_2}+b_{k_2,k_1-k_2}&=\zeta^{2-k_2}+\zeta^{k_2+2}+\zeta^{1-2k_2}+\zeta^{k_2}+\zeta^{1-k_2}+\zeta^{2k_2}\\
&=\zeta^{2-2k_2}+\zeta^{k_2+2}+\zeta^{2(2-k_2)}+\zeta^{k_2}+\zeta^{2(k_2+2)}+\zeta^{2k_2}\\
&=2\cos\left( \frac{2\pi (2-k_2)}{3}\right)+2\cos\left( \frac{2\pi k_2}{3}\right)+2\cos\left( \frac{2\pi (2+k_2)}{3}\right)\\
&=0.
\end{align*}

Let be the tensor defined by
$c_{n_1,n_2,n_3}=\zeta^{n_1+2n_2}$.
It is straightforward the following
$c_{n,n,n}=\zeta^{3n}=1$.

Moreover, setting
\begin{gather*}
\tilde c_{n_1,n_2,n_3}:=c_{n_1,n_2,n_3}+c_{n_1,n_3,n_2}+c_{n_3,n_2,n_1}+c_{n_3,n_1,n_2}+c_{n_2,n_1,n_3}+c_{n_2,n_3,n_1},
\end{gather*}
we have
\begin{gather*}
\tilde c_{n_1,n_2,n_3}:=\zeta^{n_1+2n_2}+\zeta^{n_1+2n_3}+\zeta^{n_3+2n_2}+\zeta^{n_3+2n_1}+\zeta^{n_2+2n_1}+\zeta^{n_2+2n_3}.
\end{gather*}
In particular,
\begin{align*}
\tilde c_{n_1-n_2-n_3,n_2,n_3}={}&\zeta^{n_1+n_2-n_3}+\zeta^{n_1-n_2+n_3}+\zeta^{n_3+2n_2}\\
&+\zeta^{-n_3+2n_1-2n_2}+\zeta^{-n_2+2n_1-2n_3}+\zeta^{n_2+2n_3}.
\end{align*}

If $n_1 \in 3\mathbb{Z}$,
\begin{align*}
\tilde c_{n_1-n_2-n_3,n_2,n_3}&=\zeta^{n_2-n_3}+\zeta^{-n_2+n_3}+\zeta^{n_3+2n_2}+\zeta^{-n_3-2n_2}+\zeta^{-n_2-2n_3}+\zeta^{n_2+2n_3}\\
&=\zeta^{n_2+2n_3}+\zeta^{2n_2+n_3}+\zeta^{n_3+2n_2}+\zeta^{2n_3+n_2}+\zeta^{2n_2+n_3}+\zeta^{n_2+2n_3}\\
&=3\bigl(\zeta^{n_2+2n_3}+\zeta^{2n_2+n_3}\bigr)
=3\bigl(\zeta^{n_2+2n_3}+\zeta^{-n_2-2n_3}\bigr)\\
&=6\cos\left( \frac{2\pi \left(n_2+2n_3 \right)}{3} \right).
\end{align*}

If $n_1 \in 1+3\mathbb{Z}$,
\begin{align*}
\tilde c_{n_1-n_2-n_3,n_2,n_3}={}&\zeta^{1+n_2-n_3}+\zeta^{1-n_2+n_3}+\zeta^{n_3+2n_2}\\
&+\zeta^{-n_3-2n_2+2}+\zeta^{-n_2-2n_3+2}+\zeta^{n_2+2n_3}\\
={}&\zeta^{1+n_2-n_3}+\zeta^{2+2n_2-2n_3}+\zeta^{n_3+2n_2}+\zeta^{-n_3-2n_2+2}
 +\zeta^{-n_2-2n_3+1}+\zeta^{n_2+2n_3}\\
={}&2\cos\left( \frac{2\pi (1+n_2-n_3)}{3}\right)+2\cos\left( \frac{2\pi (2n_2+n_3)}{3}\right)\\
&+2\cos\left( \frac{2\pi (1-n_2+n_3)}{3}\right)\\
={}&0.
\end{align*}

If $n_1 \in 2+3\mathbb{Z}$.
\begin{align*}
\tilde c_{n_1-n_2-n_3,n_2,n_3}={}&\zeta^{2+n_2-n_3}+\zeta^{2-n_2+n_3}+\zeta^{n_3+2n_2}
 +\zeta^{-n_3-2n_2+1}+\zeta^{-n_2-2n_3+1}+\zeta^{n_2+2n_3}\\
={}&\zeta^{2+n_2-n_3}+\zeta^{1+2n_2-2n_3}+\zeta^{n_3+2n_2}
 +\zeta^{-n_3-2n_2+1}+\zeta^{-n_2-2n_3+2}+\zeta^{n_2+2n_3}\\
={}&2\cos\left( \frac{2\pi (2+n_2-n_3)}{3}\right)+2\cos\left( \frac{2\pi (2n_2+n_3)}{3}\right)\\
&+2\cos\left( \frac{2\pi (2-n_2+n_3)}{3}\right)\\
={}&0.\tag*{\qed}
\end{align*}
\renewcommand{\qed}{}
\end{proof}

\begin{Lemma}
Let be $u_1$, $u_2$, $u_3$ the canonical coordinates of quantum cohomology as function of the Saito flat coordinates $t^1$, $t^2$, $t^3$, i.e.,
\begin{gather}\label{canonical coordinates lemma appendix}
u_k=t^1+\frac{1}{t^3}\sum_{n=1}^{\infty} A_n^k \bigl( Q^{\frac{1}{3}}t^3 \bigr)^n.
\end{gather}
Then, the coefficients $A_{k}$ are given explicitly by
\begin{gather*}
\tilde A_{3n}=\frac{n^2N_n}{(3n-1)! },\\
\sum_{n_2=2}^{3n}3\cos\left( \frac{2\pi (n_2-1)}{3}\right) \tilde A_{3n-n_2+1} \tilde A_{n_2-1}=\bigl(6-15n-9n^2 \bigr)\frac{N_n}{(3n-1)!},\\
\sum_{n_2=1}^{3n-2}\sum_{n_3=1}^{3n-n_2-1}3\cos\left( \frac{2\pi (n_2+2n_3 )}{3} \right) \tilde A_{3n-n_2-n_3} \tilde A_{n_2} \tilde A_{n_3}\\
\qquad=\bigl(54-243n +243n^2\bigr)\frac{N_n}{(3n-1)!}+\delta_{n} ,
\end{gather*}
where
\begin{gather*}
A_n^k=\tilde A_n \bigl( {\rm e}^{\frac{2\pi {\rm i}}{3}}\bigr)^{nk},
\qquad
\delta_n=
\begin{cases}
0& \text{if } n=1,\\
 \tilde\delta_n & \text{otherwise},
\end{cases}
\\
\tilde\delta_n=\sum_{n_2=2}^{n}\frac{\bigl(6(n_2-1)-3(n-n_2+1)(n_2-1)^2 \bigr)}{(3n-3n_2+2)!(3n_2-4)!} N_{n-n_2+1}N_{n_2-1}\\
\phantom{\tilde\delta_n=}{}+\sum_{n_2=2}^{n}\frac{(-4(n-n_2+1)(n_2-1) )}{(3n-3n_2+2)!(3n_2-4)!} N_{n-n_2+1}N_{n_2-1}\\
\phantom{\tilde\delta_n=}{}+\sum_{n_2=2}^{n}\frac{\bigl(-9(n_1-n_2+1)^2(n_2-1)^2 \bigr)}{(3n-3n_2+2)!(3n_2-4)!} N_{n-n_2+1}N_{n_2-1}.
\end{gather*}
\end{Lemma}

\begin{proof}
The canonical coordinates $u_1$, $u_2$, $u_3$ can be written as
\begin{gather}\label{canonical coordinates in zi}
u_i=t^1+\frac{9+\Phi^{\prime\prime}-z_i}{t^3},
\end{gather}
where $z_i$ is are the roots of
$
 (z-z_1)(z-z_2)(z-z_3)=
 z^3-s_1z^2+s_2z-s_3=0$,
where
\begin{gather}
s_1=z_1+z_2+z_3=27+2\Phi^{\prime\prime},\nonumber\\
s_2=z_1z_2+z_2z_3+z_1z_3=243+6\Phi-15\Phi^{\prime}+27\Phi^{\prime\prime}+\bigl(\Phi^{\prime\prime}\bigr)^2,\nonumber\\
s_3=z_1z_2z_3=\bigl(27+2\Phi^{\prime}-3\Phi^{\prime\prime}\bigr)^2.\label{canonical coordinates 2 in appendix}
\end{gather}
Define
\[
f_k:=\frac{1}{t^3}\sum_{n=1}^{\infty} A_n^k \bigl( Q^{\frac{1}{3}}t^3 \bigr)^n.
\]
Substituting \eqref{canonical coordinates in zi} in \eqref{canonical coordinates lemma appendix},
\begin{gather}\label{relation zi fi}
z_i=9+\Phi^{\prime\prime}-f_i.
\end{gather}
Substituting \eqref{relation zi fi} in \eqref{canonical coordinates 2 in appendix}, we obtain
\begin{gather}
f_1+f_2+f_3=\Phi^{\prime\prime},\qquad
f_1f_2+f_2f_3+f_1f_3=6\Phi-15\Phi^{\prime}-9\Phi^{\prime\prime},\nonumber\\
f_1f_2f_3=54\Phi-243\Phi^{\prime}+243\Phi^{\prime\prime}+6\Phi\Phi^{\prime\prime}- 4{\Phi^{\prime}}^2-3\Phi^{\prime}\Phi^{\prime\prime}-9{\Phi^{\prime\prime}}^2.\label{relation fi generating function phi and derivatives}
\end{gather}
Using \eqref{generating function of Gromov Witten CP2} in the first two equation of the right-hand side of \eqref{relation fi generating function phi and derivatives}, we have
\begin{gather}
\Phi^{\prime\prime}=\sum_{n=1}^{\infty}n^2 \frac{N_n}{(3n-1)!} \bigl(Q \bigl(t^3\bigr)^3\bigr)^n ,\nonumber\\
6\Phi-15\Phi^{\prime}-9\Phi^{\prime\prime}=\sum_{n=1}^{\infty}\bigl(6-15n-9n^2 \bigr)\frac{N_n}{(3n-1)!}\bigl(Q \bigl(t^3\bigr)^3\bigr)^n.\label{relation fi generating function phi and derivatives right hand side gromov witten}
\end{gather}
The third equation of the right-hand side of \eqref{relation fi generating function phi and derivatives} is bit more involved
\begin{gather*}
54\Phi-243\Phi^{\prime}+243\Phi^{\prime\prime}+6\Phi\Phi^{\prime\prime}-3\Phi^{\prime}\Phi^{\prime\prime}-4\bigl(\Phi^{\prime}\bigr)^2-9\bigl(\Phi^{\prime\prime}\bigr)^2\nonumber\\
\qquad=\sum_{n=1}^{\infty}\bigl(54-243n+243n^2 \bigr)\frac{N_n}{(3n-1)!}\bigl(Q \bigl(t^3\bigr)^3\bigr)^n\nonumber\\
\phantom{\qquad=}{}+\sum_{n_1=1}^{\infty}\sum_{n_2=1}^{\infty}\bigl(6n_2-3n_1n_2^2-4n_1n_2-9n_1^2n_2^2 \bigr)
\frac{N_{n_1}}{(3n_1-1)!} \frac{N_{n_2}}{(3n_2-1)!}\bigl(Q \bigl(t^3\bigr)^3\bigr)^{n_1+n_2}.
\end{gather*}
Using the following double infinite sum identity
\[
\sum_{k_1=1}^{\infty}\sum_{k_2=1}^{\infty} C_{k_1,k_2}=\sum_{k_1=2}^{\infty} \Bigg( \sum_{k_2=2}^{k_1} C_{k_1-k_2+1,k_2-1} \Bigg)
\]
in the equation \eqref{relation fi generating function phi and derivatives},
\begin{gather}
554\Phi-243\Phi^{\prime}+243\Phi^{\prime\prime}+6\Phi\Phi^{\prime\prime}- 4{\Phi^{\prime}}^2-3\Phi^{\prime}\Phi^{\prime\prime}-9{\Phi^{\prime\prime}}^2\nonumber\\
\qquad=\sum_{n=1}^{\infty}\bigl(54-243n +243n^2\bigr)\frac{N_n}{(3n-1)!}\bigl(Q \bigl(t^3\bigr)^3\bigr)^n\nonumber\\
\phantom{\qquad=}{}+\sum_{n_1=2}^{\infty}\left(\sum_{n_2=2}^{n_1}\frac{\bigl(6(n_2-1)-3(n_1-n_2+1)(n_2-1)^2 \bigr)}{(3n_1-3n_2+2)!(3n_2-4)!} N_{n_1-n_2+1}N_{n_2-1}\right)\bigl(Q \bigl(t^3\bigr)^3\bigr)^{n_1}\nonumber\\
\phantom{\qquad=}{}+\sum_{n_1=2}^{\infty}\left(\sum_{n_2=2}^{n_1}\frac{(-4(n_1-n_2+1)(n_2-1) )}{(3n_1-3n_2+2)!(3n_2-4)!} N_{n_1-n_2+1}N_{n_2-1}\right)\bigl(Q \bigl(t^3\bigr)^3\bigr)^{n_1}\nonumber\\
\phantom{\qquad=}{}+\sum_{n_1=2}^{\infty}\left(\sum_{n_2=2}^{n_1}\frac{\bigl(-9(n_1-n_2+1)^2(n_2-1)^2 \bigr)}{(3n_1-3n_2+2)!(3n_2-4)!} N_{n_1-n_2+1}N_{n_2-1}\right)\bigl(Q \bigl(t^3\bigr)^3\bigr)^{n_1}.\label{relation fi generating function phi and derivatives part 2}
\end{gather}
On another hand, setting
\smash{$
A_n^k=\tilde A_n \bigl( {\rm e}^{\frac{2\pi {\rm i}}{3}}\bigr)^{nk}
$}
and using the first equation of \eqref{exp 2pi 3 id}, we have
the first equation of the left-hand side of \eqref{relation fi generating function phi and derivatives} can be written as
\begin{align}
f_1+f_2+f_3&=\sum_{n=1}^{\infty} \tilde A_n \bigl( \bigl( {\rm e}^{\frac{2\pi {\rm i}}{3}}\bigr)^{n}+\bigl( {\rm e}^{\frac{2\pi {\rm i}}{3}}\bigr)^{2n}+\bigl( {\rm e}^{\frac{2\pi {\rm i}}{3}}\bigr)^{3n}\bigr)\bigl(Q^{\frac{1}{3}}t^3\bigr)^{n}\nonumber\\
&=\sum_{n=1}^{\infty} \tilde A_n \left( 1+\cos\left(\frac{2\pi n}{3}\right)\right)\bigl(Q^{\frac{1}{3}}t^3\bigr)^{n}
=\sum_{n=1}^{\infty} 3\tilde A_{3n} \bigl(Q \bigl(t^3\bigr)^3\bigr)^{n}.\label{relation fi generating function phi and derivatives left hand side}
\end{align}
Comparing equation \eqref{relation fi generating function phi and derivatives left hand side} with \eqref{relation fi generating function phi and derivatives right hand side gromov witten} and \eqref{relation fi generating function phi and derivatives}, we have
\smash{$
\tilde A_{3n}=\frac{n^2N_n}{3(3n-1)!}$}.

The second equation of the left-hand side of \eqref{relation fi generating function phi and derivatives} can be written as
\begin{align}
f_1f_2+f_2f_3+f_1f_3&=\sum_{n_1=1}^{\infty}\sum_{n_2=1}^{\infty} \tilde A_{n_1} \tilde A_{n_2} b_{n_1,n_2}\bigl(Q^{\frac{1}{3}}t^3\bigr)^{n_1+n_2}\nonumber\\
&=\sum_{n_1=2}^{\infty}\left(\sum_{n_2=2}^{n_1} \tilde A_{n_1-n_2+1} \tilde A_{n_2-1} b_{n_1-n_2+1,n_2-1}\right)\bigl(Q^{\frac{1}{3}}t^3\bigr)^{n_1},\label{relation fi generating function phi and derivatives left hand side part 2}
\end{align}
where
\smash{$
 b_{n_1,n_2}= \bigl( {\rm e}^{\frac{2\pi {\rm i}}{3}}\bigr)^{n_1+2n_2}+\bigl( {\rm e}^{\frac{2\pi {\rm i}}{3}}\bigr)^{n_1}+\bigl( {\rm e}^{\frac{2\pi {\rm i}}{3}}\bigr)^{2n_1}
$}
using Lemma \ref{lemma of exp 2pi 3} in the equation \eqref{relation fi generating function phi and derivatives left hand side part 2}
\begin{gather}
\sum_{n_1=2}^{\infty}\left(\sum_{n_2=2}^{n_1} \tilde A_{n_1-n_2+1} \tilde A_{n_2-1} b_{n_1-n_2+1,n_2-1}\right)\bigl(Q^{\frac{1}{3}}t^3\bigr)^{n_1}\nonumber\\
\qquad=\sum_{n_1=1}^{\infty}\left(\sum_{n_2=2}^{3n_1}3\cos\left( \frac{2\pi (n_2-1)}{3}\right) \tilde A_{3n_1-n_2+1} \tilde A_{n_2-1} \right)\bigl(Q^{\frac{1}{3}}t^3\bigr)^{3n_1}\nonumber\\
\qquad=\sum_{n_1=1}^{\infty}\left(\sum_{n_2=2}^{3n_1}3\cos\left( \frac{2\pi (n_2-1)}{3}\right) \tilde A_{3n_1-n_2+1} \tilde A_{n_2-1} \right)\bigl(Q\bigl(t^3\bigr)^3\bigr)^{n_1}.\label{relation fi generating function phi and derivatives left hand side part 3}
\end{gather}
Comparing \eqref{relation fi generating function phi and derivatives left hand side part 3} with \eqref{relation fi generating function phi and derivatives right hand side gromov witten},
\[
\sum_{n_2=2}^{3n}3\cos\left( \frac{2\pi (n_2-1)}{3}\right) \tilde A_{3n-n_2+1} \tilde A_{n_2-1}=\bigl(6-15n-9n^2 \bigr)\frac{N_n}{(3n-1)!}.
\]
	Using the triple infinite sum identity was used in
\[
\sum_{k_1=1}^{\infty}\sum_{k_2=1}^{\infty} \sum_{k_3=1}^{\infty} C_{k_1,k_2,k_3}=\sum_{k_1=3}^{\infty} \Bigg(\sum_{k_2=1}^{k_1-2} \sum_{k_3=1}^{k_1-k_2-1} C_{k_1-k_2-k_3,k_2,k_3} \Bigg).
\]
The third equation of the left-hand side of \eqref{relation fi generating function phi and derivatives} can be written as
\begin{align*}
f_1f_2f_3&=\sum_{n_1=1}^{\infty}\sum_{n_2=1}^{\infty}\sum_{n_3=1}^{\infty} \tilde A_{n_1} \tilde A_{n_2} \tilde A_{n_3} c_{n_1,n_2,n_3}\bigl(Q^{\frac{1}{3}}t^3\bigr)^{n_1+n_2+n_3}\\
&=\sum_{n_1=3}^{\infty}\left(\sum_{n_2=1}^{n_1-2}\sum_{n_3=1}^{n_1-n_2-1} \tilde A_{n_1-n_2-n_3} \tilde A_{n_2} \tilde A_{n_3} c_{n_1-n_2-n_3,n_2,n_3}\right)\bigl(Q^{\frac{1}{3}}t^3\bigr)^{n_1},
\end{align*}
where
\smash{$
c_{n_1,n_2,n_3}=\bigl( {\rm e}^{\frac{2\pi {\rm i}}{3}} \bigr)^{n_1+2n_2+3n_3}$}.
Then
\begin{gather}
\sum_{n_1=3}^{\infty}\left(\sum_{n_2=1}^{n_1-2}\sum_{n_3=1}^{n_1-n_2-1} \tilde A_{n_1-n_2-n_3} \tilde A_{n_2} \tilde A_{n_3} c_{n_1-n_2-n_3,n_2,n_3}\right)\bigl(Q^{\frac{1}{3}}t^3\bigr)^{n_1}\nonumber\\
\qquad=\sum_{n_1=3}^{\infty}\left(\sum_{n_2=1}^{n_1-2}\sum_{n_3=1}^{n_1-n_2-1} \tilde A_{n_1-n_2-n_3} \tilde A_{n_2} \tilde A_{n_3} \bigl( {\rm e}^{\frac{2\pi {\rm i}}{3}} \bigr)^{n_1+n_2+2n_3}\right)\bigl(Q^{\frac{1}{3}}t^3\bigr)^{n_1}\nonumber\\
\qquad=\sum_{n_1=1}^{\infty}\left(\sum_{n_2=1}^{3n_1-2}\sum_{n_3=1}^{3n_1-n_2-1}3\cos\left( \frac{2\pi (n_2+2n_3 )}{3} \right) \tilde A_{3n_1-n_2-n_3} \tilde A_{n_2} \tilde A_{n_3} \right)\bigl(Q^{\frac{1}{3}}t^3\bigr)^{3n_1}\nonumber\\
\qquad=\sum_{n_1=1}^{\infty}\left(\sum_{n_2=1}^{3n_1-2}\sum_{n_3=1}^{3n_1-n_2-1}3\cos\left( \frac{2\pi (n_2+2n_3 )}{3} \right) \tilde A_{3n_1-n_2-n_3} \tilde A_{n_2} \tilde A_{n_3} \right)\bigl(Q\bigl(t^3\bigr)^3\bigr)^{n_1}.\label{relation fi generating function phi and derivatives left hand side part 3 almost final part 1}
\end{gather}
Comparing \eqref{relation fi generating function phi and derivatives left hand side part 3 almost final part 1} with \eqref{relation fi generating function phi and derivatives part 2},
\begin{gather*}
\sum_{n_2=1}^{3n-2}\sum_{n_3=1}^{3n-n_2-1}3\cos\left( \frac{2\pi (n_2+2n_3 )}{3} \right) \tilde A_{3n-n_2-n_3} \tilde A_{n_2} \tilde A_{n_3}\\
\qquad=\bigl(54-243n+243n^2 \bigr)\frac{N_n}{(3n-1)!}+\delta_{n} ,
\end{gather*}
where
\begin{gather*}
\delta_n=
\begin{cases}
0& \text{if } n=1,\\
 \tilde\delta_n & \text{otherwise},
\end{cases}
\\
\tilde\delta_n=\sum_{n_2=2}^{n}\frac{\bigl(6(n_2-1)-3(n-n_2+1)(n_2-1)^2 \bigr)}{(3n-3n_2+2)!(3n_2-4)!} N_{n-n_2+1}N_{n_2-1}\\
\phantom{\tilde\delta_n=}{}+\sum_{n_2=2}^{n}\frac{(-4(n-n_2+1)(n_2-1) )}{(3n-3n_2+2)!(3n_2-4)!} N_{n-n_2+1}N_{n_2-1}\\
\phantom{\tilde\delta_n=}{}+\sum_{n_2=2}^{n}\frac{(-9(n_1-n_2+1)^2(n_2-1)^2 )}{(3n-3n_2+2)!(3n_2-4)!} N_{n-n_2+1}N_{n_2-1}.\tag*{\qed}
\end{gather*} \renewcommand{\qed}{}
\end{proof}

\begin{Corollary}
Let be $u_1$, $u_2$, $u_3$ the canonical coordinates of quantum cohomology as function of the Saito flat coordinates $t^1$, $t^2$, $t^3$, i.e.,
\begin{gather*}
u_k=t^1+\frac{1}{t^3}\sum_{n=1}^{\infty} A_n^k \bigl( Q^{\frac{1}{3}}t^3 \bigr)^n.
\end{gather*}
Then, the coefficients $A_{k}$ are given recursively by
\begin{gather*}
\tilde A_{3n}=\frac{n^2N_n}{(3n-1)! },\\
 \tilde A_{3n-1}=\frac{1}{3 \tilde A_{1}}\left[\bigl(6-15n-9n^2 \bigr)\frac{N_n}{(3n-1)!}-\sum_{n_2=3}^{3n-1}3\cos\left( \frac{2\pi (n_2-1)}{3}\right) \tilde A_{3n-n_2+1} \tilde A_{n_2-1}\right],\\
 \tilde A_{3n-2}=\frac{1}{9\tilde A_1^2}\left[\bigl(54-243n +243n^2\bigr)\frac{N_n}{(3n-1)!}+\delta_{n} \right.\\
 \phantom{ \tilde A_{3n-2}=}{}\left.-\sum_{n_3=2}^{3n-2}3\cos\left( \frac{2\pi (1+2n_3 )}{3} \right) \tilde A_{3n-1-n_3} \tilde A_{1} \tilde A_{n_3}\right]\\
\phantom{ \tilde A_{3n-2}=}{}-\frac{1}{9\tilde A_1^2}\left[\sum_{n_2=2}^{3n-3}\sum_{n_3=1}^{3n-n_2-1}3\cos\left( \frac{2\pi (n_2+2n_3 )}{3} \right) \tilde A_{3n-n_2-n_3} \tilde A_{n_2} \tilde A_{n_3}\right],
\end{gather*}
where
\begin{gather*}
A_n^k=\tilde A_n \bigl( {\rm e}^{\frac{2\pi {\rm i}}{3}}\bigr)^{nk}
,\qquad
\delta_n=
\begin{cases}
0& \text{if } n=1,\\
 \tilde\delta_n & \text{otherwise},
\end{cases}
\\
\tilde\delta_n=\sum_{n_2=2}^{n}\frac{\bigl(6(n_2-1)-3(n-n_2+1)(n_2-1)^2 \bigr)}{(3n-3n_2+2)!(3n_2-4)!} N_{n-n_2+1}N_{n_2-1}\\
\phantom{\tilde\delta_n=}{}+\sum_{n_2=2}^{n}\frac{(-4(n-n_2+1)(n_2-1) )}{(3n-3n_2+2)!(3n_2-4)!} N_{n-n_2+1}N_{n_2-1}\\
\phantom{\tilde\delta_n=}{}+\sum_{n_2=2}^{n}\frac{\bigl(-9(n_1-n_2+1)^2(n_2-1)^2 \bigr)}{(3n-3n_2+2)!(3n_2-4)!} N_{n-n_2+1}N_{n_2-1}.
\end{gather*}
\end{Corollary}

\subsection{Coefficients of the cross ratio function}\label{Coefficients of the cross ratio function}

\begin{Lemma}
Let the cross ratio function be defined by the formula
\begin{gather}
f\bigl(Q^{\frac{1}{3}}t^3\bigr)=\frac{u_3-u_1}{u_2-u_1} =\sum_{n=0}^{\infty} f_n \bigl(Q^{\frac{1}{3}}t^3\bigr)^n.\label{identity cross ratio qt3}
\end{gather}
Then, the Taylor series of \eqref{identity cross ratio qt3} is given by
\[
f\bigl(Q^{\frac{1}{3}}t^3\bigr)=\sum_{n=0}^{\infty} \sum_{m=0}^{n}\left(\frac{A_{n+1-m}^3-A_{n+1-m}^1}{A_1^2-A_1^1}\right)C_{m,-1}(y_1,y_2,\dots,y_m) \bigl( Q^{\frac{1}{3}}t^3 \bigr)^{n},
\]
where
\smash{$
y_n=\frac{A_{n+1}^2-A_{n+1}^1}{A_1^2-A_1^1}$},
and $A_n$ is defined in \eqref{canonical coordinates lemma1}.
\end{Lemma}

\begin{proof}
The coefficients $f_n$ \eqref{identity cross ratio qt3} is obtained by substituting \eqref{canonical coordinates lemma1} in \eqref{identity cross ratio qt3} and using Theorem~\ref{theorem potential partition polynomials}. Indeed,
\begin{align*}
\frac{u_3-u_1}{u_2-u_1}&= \frac{\sum_{n=1}^{\infty} \bigl(A_n^3-A_n^1\bigr) \bigl( Q^{\frac{1}{3}}t^3 \bigr)^n}{\sum_{n=1}^{\infty} \bigl(A_n^2-A_n^1\bigr) \bigl( Q^{\frac{1}{3}}t^3 \bigr)^n}\\
&=\left(\sum_{n=0}^{\infty} \left(\frac{A_{n+1}^3-A_{n+1}^1}{A_1^2-A_1^1}\right) \bigl( Q^{\frac{1}{3}}t^3 \bigr)^n\right)\left(1+\sum_{n=1}^{\infty} \left(\frac{A_{n+1}^2-A_{n+1}^1}{A_1^2-A_1^1}\right) \bigl( Q^{\frac{1}{3}}t^3 \bigr)^n\right)^{-1}\\
&=\left(\sum_{n=0}^{\infty} \left(\frac{A_{n+1}^3-A_{n+1}^1}{A_1^2-A_1^1}\right) \bigl( Q^{\frac{1}{3}}t^3 \bigr)^n\right)\left(\sum_{n=0}^{\infty} C_{n,-1}(y_1,y_2,\dots,y_n) \bigl( Q^{\frac{1}{3}}t^3 \bigr)^n\right)\\
&=\sum_{n=0}^{\infty} \sum_{m=0}^{\infty}\left(\frac{A_{n+1}^3-A_{n+1}^1}{A_1^2-A_1^1}\right)C_{m,-1}(y_1,y_2,\dots,y_m) \bigl( Q^{\frac{1}{3}}t^3 \bigr)^{n+m}\\
&=\sum_{n=0}^{\infty} \sum_{m=0}^{n}\left(\frac{A_{n+1-m}^3-A_{n+1-m}^1}{A_1^2-A_1^1}\right)C_{m,-1}(y_1,y_2,\dots,y_m) \bigl( Q^{\frac{1}{3}}t^3 \bigr)^{n},
\end{align*}
where
\smash{$
y_n=\frac{A_{n+1}^2-A_{n+1}^1}{A_1^2-A_1^1}$},
and the following Double sum identity was used:
\[
\sum_{n=0}^{\infty}\sum_{m=0}^{\infty} C_{n,m}=\sum_{n=0}^{\infty}\sum_{m=0}^{n} C_{n-m,m}.
\]
Lemma proved.
\end{proof}

\subsection{Coefficients of modular lambda function}\label{Coefficients of Modular lambda function}

Recall the Eisenstein series $E_2(\tau)$, $E_4(\tau)$, $E_6(\tau)$ at the point $\tau={\rm e}^{\frac{2\pi {\rm i}}{3}}$ have the following values:
\begin{gather}\label{special values of Eisenstein series in the equiharmonic lattice }
E_2\bigl({\rm e}^{\frac{2\pi {\rm i}}{3}}\bigr)=\frac{2\sqrt{3}}{\pi}, \qquad E_4\bigl({\rm e}^{\frac{2\pi {\rm i}}{3}}\bigr)=0, \qquad E_6\bigl({\rm e}^{\frac{2\pi {\rm i}}{3}}\bigr)=\frac{\bigl(\Gamma\bigl(\frac{1}{3}\bigr)\bigr)^{18}}{4^6}.
\end{gather}
Moreover, there exist the following relation between $x$, $x^{\prime}$ and $\Delta$:
\begin{gather}
\frac{2^8}{3^3}(2\pi )^{-12}\Delta(\tau)=\frac{(x^{\prime})^6}{x^4(x-1)^4}.\label{main relationship between x,xprime and delta}
\end{gather}
The special values of $x$ at $\tau={\rm e}^{\frac{2\pi {\rm i}}{3}}$ is
\begin{gather}
x\bigl({\rm e}^{\frac{2\pi {\rm i}}{3}}\bigr)={\rm e}^{\frac{\pi {\rm i}}{3}}.\label{special value of x in the equiharmonic lattice}
\end{gather}
Substituting \eqref{special value of x in the equiharmonic lattice}, \eqref{special values of Eisenstein series in the equiharmonic lattice } in \eqref{main relationship between x,xprime and delta}, we compute $x^{\prime}\bigl({\rm e}^{\frac{2\pi {\rm i}}{3}}\bigr)$. Computing
\begin{gather}
\frac{2^8}{3^3}(2\pi )^{-12}\frac{{\rm d}^n\log\Delta(\tau)}{{\rm d}\tau^n}=\frac{{\rm d}^n}{{\rm d}\tau^n}\left(\frac{(x^{\prime})^6}{x^4(x-1)^4}\right)n,\label{main relationship between x,xprime and delta derivatives }
\end{gather}
we obtain in the left-hand side \eqref{main relationship between x,xprime and delta derivatives } a polynomial expression in terms of $E_2$, $E_4$, $E_6$ and in the right-hand side a rational expression of \smash{$\frac{{\rm d}^nx(\tau)}{{\rm d}\tau^n}$}. Then, we can derive \smash{$\frac{{\rm d}^nx(\tau)}{{\rm d}\tau^n}$} for any order $n$ by using the recursive relation \eqref{main relationship between x,xprime and delta derivatives }.

\subsection*{Acknowledgements}
I am grateful to Professor Hertling for his remarkable advice, guidance, and for proofreading this manuscript. Furthermore, I would also like to thank Professors Doran and Milanov for helpful discussions. I am also thankful to the anonymous referees for their valuable comments and suggestions. This work was funded by the Deutsche Forschungsgemeinschaft (DFG, German Research Foundation) -- 494849004.
The work was carried out while I was at the University of Mannheim, and I now work at the Max Planck Institute of Molecular Cell Biology and Genetics, Dresden.

\LastPageEnding

\end{document}